\documentclass[12pt]{article}
\usepackage[paper=letterpaper, left=1.25in, right=1.25in, top=1.25in, bottom=1.25in]{geometry}

\usepackage{cmap}
\usepackage{amsmath, amsfonts, amssymb, amsthm, array, graphicx, float, morefloats, setspace, wrapfig}
\usepackage{icomma}
\usepackage[utf8]{inputenc}
\usepackage[english]{babel}

\usepackage{titlesec}
\titleformat{\section}{\normalfont\large\scshape\centering}{\thesection.}{0.5em}{}
\titleformat{\subsection}{\normalfont\normalsize\scshape\centering}{\thesubsection.}{0.5em}{}
\titleformat{\subsubsection}{\normalfont\normalsize\scshape\centering}{\thesubsubsection.}{0.5em}{}
\titlespacing*{\section}{0pt}{2.5ex plus 1ex minus .2ex}{1.5ex plus .2ex}
\titlespacing*{\subsection}{0pt}{2ex plus 1ex minus .2ex}{1ex plus .2ex}
\titlespacing*{\subsubsection}{0pt}{1.5ex plus 1ex minus .2ex}{0.5ex plus .2ex}

\usepackage[title]{appendix}
\usepackage{chngcntr}
\usepackage{apptools}
\AtAppendix{\counterwithin{lemma}{section}}

\usepackage[protrusion=true,expansion=true]{microtype}

\usepackage{rotating}
\usepackage{threeparttable}
\usepackage{booktabs}
\usepackage{tablefootnote}

\usepackage{color}

\usepackage[section]{placeins}
\usepackage{subcaption}

\usepackage{xr-hyper}
\usepackage[colorlinks=true]{hyperref}
\hypersetup{colorlinks, citecolor=red, linkcolor=red, urlcolor=red, filecolor=red}

\usepackage{tikz}

\usepackage{verbatim}

\usepackage[font={small},labelfont=bf,skip=5pt]{caption}
\newcommand{\E}{{\mathbb E}}

\newcommand{\cN}{{\mathcal N}}

\newcommand{\cT}{{\mathcal T}}
\newcommand{\cU}{{\mathcal U}}

\newcommand{\R}{{\mathbb R}}
\renewcommand{\P}{{\mathbb P}}

\newcommand{\argmin}{\operatornamewithlimits{argmin}}

\usepackage{bm}
\usepackage{dsfont}  

\usepackage{algorithm}
\usepackage{algpseudocode}

\newtheoremstyle{sctheorem}
  {\topsep}
  {\topsep}
  {\itshape}
  {}
  {\scshape}
  {.}
  { }
  {}%

\theoremstyle{sctheorem}
\newtheorem{theorem}{Theorem}
\newtheorem{lemma}{Lemma}
\newtheorem{proposition}{Proposition}
\newtheorem{corollary}{Corollary}
\newtheorem{assumption}{Assumption}

\newtheoremstyle{scexample}
  {\topsep}
  {\topsep}
  {\normalfont}
  {}
  {\scshape}
  {.}
  { }
  {}%
\theoremstyle{scexample}
\newtheorem{exmp}{Example}

\theoremstyle{remark}
\newtheorem{remark*}{Remark}

\makeatletter
\def\@fnsymbol#1{\ensuremath{\ifcase#1\or \dagger\or \ddagger\or
   \mathsection\or \mathparagraph\or \|\or **\or \dagger\dagger
   \or \ddagger\ddagger \else\@ctrerr\fi}}
\makeatother

\usepackage[round]{natbib}
\renewenvironment{abstract}{%
    \small
    \begin{center}
        \textsc{Abstract}\vspace{-0.5em}
    \end{center}
    \begin{quotation}\noindent\ignorespaces
}{%
    \end{quotation}\vspace{1em}
}

\title{Penalized GMM Framework for Inference on Functionals of Nonparametric Instrumental Variable Estimators\thanks{Previously circulated as ``Automatic Debiased Machine Learning in Presence of                       Endogeneity''.}\footnote{I am very grateful to Amit Gandhi, Xu Cheng, and Karun Adusumilli for their guidance and support, as well as Frank Diebold, Wayne Gao, and participants of the UPenn Econometrics Workshop 2021 and Young Scholars Conference on Machine Learning in Economics and Finance 2021 for insightful comments and suggestions. I would like to thank IRI for making the data available. All estimates and analysis in this paper, based on data provided by IRI, are by the authors and not by IRI. Replication code is available at \url{https://github.com/edbakhitov/ADMLIV}.}}
\author{Edvard Bakhitov\thanks{University of Pennsylvania. e-mail: edbakhitov@gmail.com}}
\date{\today}

\begin{document}

\maketitle

\begin{abstract}
	\singlespacing
	This paper develops a penalized GMM (PGMM) framework for automatic debiased inference on functionals of nonparametric instrumental variable estimators. We derive convergence rates for the PGMM estimator and provide conditions for root-$n$ consistency and asymptotic normality of debiased functional estimates, covering both linear and nonlinear functionals. Monte Carlo experiments on average derivative show that the PGMM-based debiased estimator performs on par with the analytical debiased estimator that uses the known closed-form Riesz representer, achieving 90--96\% coverage while the plug-in estimator falls below 5\%. 
	
	We apply our procedure to estimate mean own-price elasticities in a semiparametric demand model for differentiated products. Simulations confirm near-nominal coverage while the plug-in severely undercovers. Applied to IRI scanner data on carbonated beverages, debiased semiparametric estimates are approximately 20\% more elastic compared to the logit benchmark, and debiasing corrections are heterogeneous across products, ranging from negligible to several times the standard error.
\end{abstract}

\newpage

\addtocontents{toc}{\protect\setcounter{tocdepth}{-1}}

\section{Introduction}

Instrumental variables (IV) methods are widely used in applied research for estimation and inference in models containing endogenous regressors. In many cases, economic theory does not impose any functional form restrictions motivating nonparametric instrumental variables (NPIV) methods, where the function of interest is not assumed to be known up to a finite-dimensional parameter. In many cases, structural parameters of economic interest appear as functionals of that underlying unknown function. Examples are policy effects, average (weighted) partial effects, consumer surplus, measures of substitution patterns, and various counterfactuals from structural models. It is quite common for the estimation problem to be high-dimensional. There might be many control variables which we want to include in a flexible way along with the endogenous regressor, or a structural model may depend on many variables, e.g. in the demand for differentiated goods framework, the demand function depends on the vector of prices and product characteristics of all products in the market. In this paper, we are interested in estimation and inference on structural economic objects in presence of endogeneity when the dimensionality of the problem is (moderately) high. 

Machine learning (ML) literature provides a collection of modern statistical tools for flexible estimation of various statistical objects that are especially powerful in high-dimensional settings.  However, standard ML estimators, such as Lasso, boosting, or Neural Networks are unable to pick up causal relationships when endogenous regressors are present \citep[see e.g.,][]{hartford2017deepIV}. On the other hand, there is a new line of research in machine learning and computer science communities that offers a series of new algorithms that both addresses endogeneity and can be applied  in high-dimensional environments, we refer to them as MLIV estimators. These algorithms are data-driven and exploit various forms of regularization to ameliorate the ill-posedness of the problem while maintaining the functional form flexibility. Examples include the DeepIV estimator \citep{hartford2017deepIV} and its regularized version \citep{li2024regdeepiv}, the Kernel IV regression \citep{singh2019kiv}, the Dual IV regression \citep{muandet2020dual}, the  DeepGMM estimator \citep{bennett2019deepGMM}, the Double Lasso estimator of \cite{gold2020}, a series of estimators constructed using the minimax framework of \cite{dikkala2020minimax}, the deep feature estimator \citep{xu2020learning}, the boostIV estimator \citep{bakhitov2021boostIV}, the minimax IV regression \citep{bennett2023minimax}, and the two-stage ML estimator \citep{bruns-smith2025tsml}. The goal of this paper is to use these novel methods to estimate and perform inference on various economic objects of interest that appear as functionals of the underlying structural function under endogeneity.

As standard ML algorithms, MLIV estimators produce inherently biased estimates. The main source of bias is regularization and/or model selection needed to balance out squared bias and variance to obtain overall small mean squared errors. In the NPIV context, regularization is particularly important as it plays a dual role. First, it allows to deal with the curse of dimensionality, as in the case of standard ML estimators. Second, it is necessary to solve the ill-posed problem. As a result, regularization and/or model selection bias leads to poor coverage unless it is corrected for. Furthermore, the bias term will propagate into the functional estimate if we simply plug-in an MLIV estimator into the functional formula. As \cite{chernozhukov2022adml} point out, squared bias of plug-in estimators can shrink slower than the variance, leading to extremely poor confidence interval coverage. 

In this paper, we provide an approach for performing valid asymptotic inference on functionals of MLIV estimators. Building on the automatic debiased machine learning framework of \cite{chernozhukov2022adml}, hereafter CNS, we construct Neyman-orthogonal moment functions by adding the influence function adjustment term for the NPIV estimator derived by \cite{ichimura2022influence} to the identifying moment conditions. The resulting debiased moment function has a zero derivative with respect to the MLIV estimator, ensuring insensitivity to local perturbations around the true value of the estimated function and allowing plug-in of noisy MLIV estimates without strongly violating the moment condition. We focus on regular functionals with a finite semiparametric asymptotic variance bound necessary for root-$n$ estimability, allowing for both linear and nonlinear functionals, though the conditions for root-$n$ rate are much tighter for the nonlinear case. 
  
The influence function adjustment depends on the Riesz representer (RR) for the identifying moment function in the linear case, or the derivative of the identifying moment condition in the nonlinear case. In the NPIV framework, the closed-form RR is typically either very complicated to derive or even unknown. Our main methodological contribution is a penalized GMM (PGMM) framework that estimates the RR directly from the identifying moment conditions by exploiting their orthogonality, without requiring knowledge of the closed-form expression. The debiasing is therefore automatic in the sense that it depends only on the form of the identifying moment function but not on the form of the bias correction term. The PGMM estimator is novel and, to the best of our knowledge, is the only non-minimax automatic estimator of the RR in the NPIV framework, generalizing the Lasso minimum distance estimator of CNS to allow for a more general form of the influence function.  

We derive the convergence rate for the PGMM estimator and provide conditions for root-$n$ consistency and asymptotic normality of the debiased MLIV estimator of the functional of interest. To accommodate for a large variety of MLIV estimators, we only require certain mean square consistency and convergence rates for MLIV estimators. The required conditions differ quite drastically for linear and nonlinear functionals. For linear functionals it is sufficient to require the MLIV estimator to converge at some positive rate in the projected mean square norm. It is well-known that NPIV estimators exhibit much faster convergence rates in the projected norm rather than in the standard mean square norm due to ill-posedness \citep[see e.g.,][]{blundell2007, ChenPouzo2012, ChenPouzo2015QLR}. However, for nonlinear functionals it is necessary to account for the linearization bias which requires the convergence rate to be faster than $n^{-1/4}$, which is a standard condition in the semiparametric literature \citep{newey1994}. Moreover, the presence of nonlinearities in the identifying moment function results in the convergence rate condition in the standard mean square norm rather than the projected norm, which makes it harder to satisfy in practice. 

In Monte Carlo experiments based on the average derivative functional with multiple endogenous regressors, we demonstrate that the PGMM-estimated Riesz representer performs on par with the known analytical expression. The plug-in estimator exhibits severely deteriorating coverage as the sample size grows, falling from 64\% at $n=100$ to below 5\% at $n=10000$. In contrast, the automatic debiased estimator achieves near-nominal coverage of 90--96\% throughout, matching the performance of the analytical debiased estimator that uses the closed-form Riesz representer \citep{chen2023ann_npiv}. Moreover, the automatic debiasing procedure exhibits greater numerical stability in small samples and high-dimensional settings, where the analytical approach can suffer from instability in the intermediate estimation steps required to evaluate the closed-form expression.

We apply our approach to estimate the mean own-price elasticity functional in the nonparametric demand for differentiated products framework \citep{berryhaile2016,compiani2022market} that has been gaining popularity in the last years as an alternative to the standard parametric procedure of \cite{BerryLevinsohnPakes1995}, hereafter, BLP. Building on the dimensionality reduction idea of \cite{gandhi_houde2019}, we derive a semiparametric inverse demand estimation equation that remains tractable as the number of products grows. We then construct own-price elasticities as nonlinear functionals of the estimated inverse demand system using the implicit function theorem. To our knowledge, this is the first application of automatic debiased machine learning to nonlinear functionals in the endogenous setting.

In our Monte Carlo experiments, we demonstrate that the plug-in estimator exhibits small bias but severely undercovers, with empirical coverage ranging from 31\% to 55\%, reflecting that its standard errors dramatically understate the true sampling variability. In contrast, the debiased estimator achieves near-nominal coverage of 91--95\% across all configurations, demonstrating that the automatic debiasing procedure delivers valid inference even for nonlinear functionals in moderately sized samples.

We use the Monte Carlo results as a basis for our empirical application where we estimate own-price elasticities using scanner data. We show that semiparametric elasticity estimates are $\approx 20\%$ larger in absolute value across products compared to the logit benchmark, reflecting richer substitution patterns. Moreover, the debiasing corrections vary across products in both magnitude and direction, with several products exhibiting corrections that substantially exceed the associated standard errors, underscoring the practical importance of the debiasing procedure for valid inference in empirical applications.

This paper connects several strands of literature. First, since the focus of the paper is functionals of nonparametric quantities, our methodology relates to the literature on semiparametric statistical theory \citep{van_der_vaart1991,bickel1993efficient,newey1994,robins_rotnitzky1995,van_der_vaart2000}. These papers focus on functionals of densities or regressions in low dimensional settings, while in our paper we focus on functionals of MLIV estimators over domains that may include low, moderate, and high dimensional objects. A more recent work by \cite{chernozhukov2022locally} generalizes and extends the insights from the classical theory by constructing Neyman-orthogonal moment conditions allowing for a wide range of ML estimators\footnote{\cite{chernozhukov2022locally} provide high-level conditions for inference on functionals for conditional moment restriction models that nest the NPIV problem (see Theorem 19). Our results are complementary as we provide an estimator of the RR and give low-level conditions to derive its convergence rate.}. We follow \cite{chernozhukov2022locally} and use Neyman-orthogonal moment functions with the influence function adjustment term for the NPIV estimator from \cite{ichimura2022influence}.

 Riesz representers are important objects in semiparametric theory as they appear in calculations of the asymptotic variance of functionals of nonparametric quantities \citep{ichimura2022influence,chernozhukov2022locally}. For the same reason they appear in the influence function calculations, which makes estimation of RRs a cornerstone of the debiased machine learning literature.  \cite{chernozhukov2022dml_gl} and CNS propose Lasso and Dantzig minimum distance estimators of the RR based on the sparse approximation assumption. While the latter provides asymptotic results for regular functionals, the former provides finite sample analysis and also allows for irregular functionals. A recent paper by \cite{chernozhukov2022riesznet} proposes to use neural networks and random forests to estimate the RR. On the other hand, \cite{chernozhukov2025adversarial} take a different approach and allow for a more general estimator of the RR based on the minimax framework of \cite{dikkala2020minimax}. \cite{ghassami2022minimax} extend the doubly robust framework to settings where nuisance functions solve integral equations and propose a minimax kernel machine learning approach for their estimation, with applications to proximal causal inference. \cite{bennett2025inference} use penalized minimax estimators for both the primary and debiasing nuisance functions and derive conditions for asymptotic normality of linear functionals that do not require closedness or strong identification of $\gamma$. Our work is complementary: while we focus on the NPIV framework, we provide an automatic PGMM-based approach to estimate the influence function adjustment and give conditions for valid asymptotic inference on debiased estimated for both linear and nonlinear regular functionals.

This work also contributes to the literature on estimation and inference on conditional restrictions models which nest the NPIV regression problem as a special case. Several NPIV estimators are now available including kernel-based estimators \citep{hall_horowitz2005,darolles2011npiv} and series or sieve estimators \citep{newey_powell2003,blundell2007,ChenPouzo2012,chen2023ann_npiv}. There are several papers focusing on linear regular functionals of NPIV estimators, see e.g., \cite{AiandChen2003}, \cite{santos2011_IV}, and \cite{severini_tripathi2012} among others. \cite{ChenPouzo2015QLR} and \cite{chen_christensen2018optimal} give conditions for pointwise and uniform asymptotic normality, respectively, of possibly nonlinear functionals of the sieve NPIV estimator. The results presented in this paper are complementary to the results on inference on functionals of NPIV estimators.

The paper also touches on the growing literature on flexible demand estimation in differentiated product markets. \cite{compiani2022market} follows the nonparametric identification arguments of \cite{berry2014identification} and demonstrates the performance of the NPIV estimator in a very simple case of two products with two characteristics. He uses Bernstein polynomials along with shape restrictions to alleviate the curse of dimensionality and nonparametrically estimate the inverse demand function. In that regard, our modeling approach is complementary: we leverage the dimensionality reduction idea of \cite{gandhi_houde2019} to construct a semiparametric inverse demand estimation equation that combined with our debiasing procedure allows the practitioner to apply it to more realistic settings. \cite{lu2023semi} consider a similar framework, but they focus on applications with large amounts of products instead of large amount of markets. A closely related paper by \cite{singh2023choice} models choice probabilities directly instead of demand inversion and applies the control function approach to account for price endogeneity. \cite{fosgerau2020inverse} and \cite{monardo2021measuring} consider a different class of inverse product differentiation models which generalize the inverse demand function of the nested logit model.

The remainder of the paper is organized as follows. Section \ref{sec:npiv_issues} briefly introduces the NPIV framework,  discusses practical issues and MLIV estimators. In Section \ref{sec:mliv_functionals}, we describe the objects of interest and provide several economic examples. We also illustrate how to construct the debiased estimator and the estimator of its asymptotic variance. Finally, we introduce the PGMM estimator of the RR. Section \ref{sec:pgmm_prop} gives conditions necessary to derive a convergence rate for the PGMM estimator. Section \ref{sec:asy_prop_lin} gives conditions for root-$n$ consistency and asymptotic normality of the debiased estimator for linear functionals. In Section \ref{sec:asy_prop_nonlin} we introduce additional conditions necessary to extend our results to nonlinear functionals. Section \ref{sec:toy_mc} examines the performance of the debiased estimator in a simple Monte Carlo exercise. In Section \ref{sec:np_demand} we introduce the semiparametric demand estimation framework and estimate the mean own-price derivative functional using both simulated and real data. Section \ref{sec:conclusion} gives conclusions and provides possible extensions. All additional details and proofs are left for the Appendix.

\vspace{1em}
\noindent \textsc{Notation:} For a vector $x \in \R^{n}$, let $|x|_{1}$, $||x||$, and $||x||_{\infty}$ denote its $\ell_{1}$-, $\ell_{2}$-, and $\ell_{\infty}$-norms respectively. For an $m \times n$ matrix $A$, we define $||A||_{\infty} = \max_{j,k}|A_{jk}|$. Let $||A||_{\ell_{\infty}} = \max_{i}\sum_{j=1}^{n}|A_{ij}|$ denote the induced $\ell_{\infty}$-norm of A.  For $S \subseteq \{1,\dots,\,n\}$ let $x_S$ be the modification of $x$ that places zeros in all entries of $x$ whose index does not belong to $S$. For a random variable $X$, let $L_{2}(X)$ denote a space of all measurable and square integrable functions.


\section{Flexible estimation under endogeneity} \label{sec:npiv_issues}

We start by briefly discussing the NPIV framework, consequences of ill-posedness for practitioners, and motivating the use of MLIV estimators.

\subsection{Nonparametric IV framework}\label{subsec:npiv}

Consider the nonparametric instrumental variables framework of \cite{newey_powell2003},
\begin{equation*} 
	Y = \gamma_{0}(X) + \varepsilon,\quad \E[\varepsilon|Z] = 0,
\end{equation*}
where $Y$ is an explanatory variable, $X$ is a vector of potentially endogenous regressors, $Z$ is a vector of instruments, and $\varepsilon$ is an error term. Suppose that $\gamma_{0}$ is identified and the completeness condition holds, i.e.\ for all measurable real functions $\delta$ with finite expectation,
\begin{equation*}
	\mathbb{E}[\delta(X)|Z] = 0 \Rightarrow \delta(X) = 0.
\end{equation*}
Intuitively, this condition implies that there is enough variation in the instruments to explain the variation in the endogenous covariates. In the linear model, completeness reduces to the usual rank condition.

The unknown function $\gamma_{0}$ solves the integral equation
\begin{equation} \label{eq:cond_exp}
	\mathbb{E}[Y|Z] = \int \gamma_{0}(x)f(x|z)dx,
\end{equation}
where $f$ denotes the conditional pdf of $X$ given $Z$. Solving for $\gamma$ directly is an ill-posed problem as it involves inverting linear compact operators \citep[see e.g.,][]{kress1989}. Ill-posedness implies that the solution to \eqref{eq:cond_exp} is not continuous in $\E[Y|Z]$ and $f(x|z)$: one cannot construct an estimator of $\gamma$ by simply plugging in consistent estimators of $\E[Y|Z]$ and $f(x|z)$ and solving for $\gamma$.

A standard solution is regularization, which means constructing an estimator of $\gamma_{0}$ so that ill-posedness does not affect consistency. In essence, regularization avoids estimation of higher-order terms that drive up variance. Traditional approaches include replacing $\gamma$ with a finite-dimensional approximation \citep{kress1989}, Tikhonov regularization \citep[see e.g.,][]{hall_horowitz2005,CFR2007regNPIV}, and functional form restrictions \citep{chetverikov2017npiv}.

Ill-posedness slows convergence rates of NPIV estimators relative to standard nonparametric regression. To quantify this, let $T:L_{2}(X)\mapsto L_{2}(Z)$ denote the conditional expectation operator,
\begin{equation*}
	T\gamma = \E[\gamma(X)|Z],
\end{equation*}
and define the measure of ill-posedness as
\begin{equation*}
	\tau = \sup_{\gamma\in\Gamma}\frac{||\gamma - \gamma_{0}||}{||T(\gamma - \gamma_{0})||},
\end{equation*}
where $\Gamma \subseteq L_{2}(X)$ and $||T(\gamma - \gamma_{0})||=\sqrt{\E\{\E[\gamma - \gamma_{0}|Z]\}^{2}}$ is the projected mean square norm. Assuming $\tau$ is bounded,
\begin{equation*}
	||\gamma - \gamma_{0}|| \leq \tau||T(\gamma - \gamma_{0})||,
\end{equation*}
so convergence in the mean square norm is always slower than in the projected norm. Importantly, fast rates in the projected norm are achievable even when mean square rates are slow, since the projected norm sidesteps ill-posedness \citep[see e.g.,][]{blundell2007,ChenPouzo2012,dikkala2020minimax}.

\subsection{Practical concerns and MLIV estimators}\label{subsec:practical_mliv}

Standard NPIV methods provide flexible approaches to nonparametric estimation under endogeneity. Nevertheless, ill-posedness poses several challenges. From the practitioner's standpoint, it limits what can be learned about $\gamma_{0}$: \cite{horowitz2011applied} notes that only low-order approximation terms can be estimated with desirable precision, reflecting a fundamental characteristic of the estimation problem rather than a limitation of any particular method. Using a Gaussian example, \cite{newey2013nonparametric} illustrates the connection between ill-posedness and instrument strength, showing that stronger instruments yield more precise estimates of higher-order terms.

The curse of dimensionality, which affects all nonparametric estimators, becomes more acute in the NPIV context. In severely ill-posed cases, convergence rates may be only polynomial in $\log n$ rather than $n$ \citep[see][]{blundell2007,darolles2011npiv,ChenPouzo2012}. As a result, variance of NPIV estimators can be substantially higher than that of standard nonparametric regression estimators, even in moderate dimensions. \cite{bakhitov2025mliv} provides Monte Carlo evidence demonstrating that variance of series NPIV estimators grows rapidly with the number of endogenous regressors, while the problem is much less severe under exogeneity.

One promising solution is to leverage machine learning algorithms designed specifically for the IV setting. Standard ML estimators, such as Lasso, boosting, and neural networks, fit the conditional expectation $\E[Y|X]$ rather than the structural function $\gamma$, and thus fail to capture causal relationships when endogeneity is present. However, recent literature has developed MLIV estimators that combine sophisticated regularization with the IV moment conditions to solve the ill-posed problem while maintaining functional form flexibility. Examples include DeepIV \citep{hartford2017deepIV,li2024regdeepiv}, Kernel IV \citep{singh2019kiv}, Dual IV \citep{muandet2020dual}, DeepGMM \citep{bennett2019deepGMM}, Double Lasso \citep{gold2020}, minimax estimators \citep{dikkala2020minimax,bennett2023minimax}, iterative estimators \citep{xu2020learning,bakhitov2021boostIV}, and the two-stage ML estimator \citep{bruns-smith2025tsml}. We refer the reader to \cite{bakhitov2025mliv} for a detailed comparison of some of these methods and empirical evidence on their finite-sample performance.

In practice, the structural function itself is rarely the object of interest; rather, practitioners seek economically meaningful quantities such as average partial effects. Consider, for example, a demand estimation problem. The demand level itself does not bear a lot of economic meaning while objects like partial effects of demand shifters, consumer surplus, price and income elasticities or diversion ratios are potential objects of interest. These are functionals of the structural function. The goal of this paper is to provide valid inference on such functionals when $\gamma$ is estimated using MLIV methods.

\section{Learning functionals of MLIV estimators} \label{sec:mliv_functionals}

\subsection{Functionals of interest and economic examples}

This paper focuses on estimation and inference on functionals of a flexible (i.e. nonparametric) structural function $\gamma_{0}$ in presence of endogenous regressors, i.e. within the framework of the nonparametric instrumental variables model. Let $W_{i} \equiv (Y_{i}, X_{i}, Z_{i})$ be a data observation. Let $m(W, \,\gamma)$ denote a functional of $\gamma$ that depends on an observation $W$. We consider parameters of interest of the form
\begin{equation*}
	\theta_{0} = \E[m(W,\,\gamma_{0})].
\end{equation*}
For expositional convenience, in this Section we will focus on functionals that depend linearly on $\gamma$. In Section \ref{sec:asy_prop_nonlin} we extend our results to nonlinear functionals. The object of interest $\theta_{0}$ is an expectation of some functional $m(W,\,\gamma_{0})$ over the data distribution. Hence, we are interested in mean effects, which restricts a set of possible functionals of interest, such as, for example, a simple evaluation functional $\theta_{0} = \gamma_{0}(\bar{X})$, where $\bar{X} \in \text{supp}(X)$. However, our framework is still general enough and covers a wide range of economically important objects.

Below, we give several examples of the types of objects under consideration, including both linear and nonlinear functionals.

\begin{exmp}\label{exmp:weighted_avg_deriv}
\textsc{Weighted average derivative.} \\
In this example, $X$ is a vector of continuous endogenous regressors and
\begin{equation*}
	\theta_{0} = \E\left[\omega(X)\frac{\partial \gamma_{0}(X)}{\partial X_1}\right],
\end{equation*}
which is a weighted average derivative of $\gamma_{0}$ with respect to $X_1$ with known weight $\omega(X)$ as in \cite{ai_chen2007estimation}. Here $m(W,\,\gamma) = \omega(X)\partial\gamma(X)/\partial X_1$, which is linear in $\gamma$. When $\omega(X)=1$, $\theta_{0}$ becomes an average partial effect of $X_1$ on $\gamma_{0}(X)$.
\end{exmp}

\begin{exmp}\label{exmp:avg_policy_effect}
\textsc{Average policy effect.} \\
The object of interest here is the average effect of changing the covariates according to some transformation $x \mapsto g(x)$,
\begin{equation*}
	\theta_{0} = \E[\gamma_{0}(g(X)) - \gamma_{0}(X)],
\end{equation*}
where $m(W,\,\gamma) = \gamma_{0}(g(X)) - \gamma_{0}(X)$ is a linear functional. Thus, $\theta_{0}$ measures the average policy effect of a counterfactual change of covariate values. 
\end{exmp}

\begin{exmp}\label{exmp:cs_dwl}
\textsc{Average consumer surplus (CS) and deadweight loss (DWL).} \\
This example is based on \cite{hausman_newey1995} and its adaptation to the NPIV setting by \cite{chen_christensen2018optimal}. Here, $X = (P,\,I,\,X_{2})$, where $P$ is  product price, which is potentially endogenous, $I$ is consumer income, and $X_{2}$ includes additional covariates.  Let $S(p^{0},\,\iota,\,x_{2})$ denote the exact CS from a price change from $p^{0}$ to $p^{1}$ at income level $\iota$ and covariate values $x_{2}$. Then $S(p^{0},\,\iota,\,x_{2})$ is a solution to
\begin{equation*}
	\frac{\partial S(p(u),\,\iota,\,x_{2})}{\partial u} = -\gamma_{0}(p(u),\,\iota - S(p(u),\,\iota,\,x_{2}),\,x_{2})\frac{\partial p(u)}{\partial u},\quad S(p(1),\,\iota,\,x_{2}) = 0,
\end{equation*}
where $p:[0,\,1] \mapsto \R$ is a twice continuously differentiable price path with $p(0) = p^{0}$ and $p(1) = p^{1}$. Let $D(p^{0},\,\iota,\,x_{2})$ denote the corresponding DWL functional given by
\begin{equation*}
	D(p^{0},\,\iota,\,x_{2}) = S(p^{0},\,\iota,\,x_{2}) - \left(p^{1} - p^{0}\right)\gamma_{0}(p^{1},\,\iota,\,x_{2}).
\end{equation*}
The objects of interest are 
\begin{align*}
	\theta_{0}^{CS} & = \E[\omega(I,\,X_2)S(p(u),\,I,\,X_{2})], \\
	\theta_{0}^{DWL} & = \E[\omega(I,\,X_2)D(p(u),\,I,\,X_{2})] = \theta_{0}^{CS} - \E[\omega(I,\,X_2)(p^{1} - p^{0})\gamma(p^{1},\,I,\,X_{2})],
\end{align*}
where $\omega$ is a weighting function that does not depend on the price level. Unless demand is independent of income, the exact CS and DWL are typically nonlinear functionals of $\gamma_{0}$.
\end{exmp}


\subsection{Orthogonal moment condition}

Suppose that we are given $\hat\gamma$, an MLIV estimator of $\gamma_{0}$. A natural approach to estimate $\theta_{0}$ is to simply plug-in $\hat\gamma$ into $m$ and replace the expectation with the sample average, 
\begin{equation*}
	\hat\theta^{\text{plug-in}} = \frac{1}{n}\sum_{i=1}^{n}m(W,\,\hat\gamma).
\end{equation*}
That said, the plug-in estimator will not be root-$n$ consistent if the first-order bias does not vanish at root-$n$ rate, which is the case when $\hat\gamma$ involves regularization and/or model selection \citep{chernozhukov2022locally}. In the NPIV model, regularization is essential to dealing with ill-posedness rendering all NPIV/MLIV estimators regularized estimators. 

The plug-in estimator suffers from the first-order bias because the moment condition the moment condition defining $\theta_{0}$ is not orthogonal to local perturbations of $\gamma$ around $\gamma_{0}$. Namely, let $\delta$ be a local perturbation around $\gamma_{0}$, then the Gateaux derivative in the direction $\delta$ is
\begin{equation*}
	\left.\frac{\partial}{\partial \tau}\E[m(W,\,\gamma_{0} + \tau\delta)]\right|_{\tau=0} = \E[m(W,\,\delta)] \neq 0.
\end{equation*}
Thus, obtaining an orthogonal moment condition is a crucial step for establishing our results.

We consider functionals $m(W,\,\gamma)$ such that there exists a function $\alpha_{0}(Z)$ with $\E[\alpha^{2}_{0}(Z)] < \infty$ and
\begin{equation}
	\E[m(W,\,\gamma)] = \E[\alpha_{0}(Z)\gamma(X)] \text{ for all } \gamma \text{ with } \E[\gamma^{2}(X)] < \infty. 
\end{equation}
As discussed in \cite{ichimura2022influence}, if there exists $v(X)$ with $\E[v^{2}(X)] < \infty$ and $\E[m(w,\,\gamma)] = \E[v(X)\gamma(X)]$, then the existence of $\alpha_{0}(Z)$ requires $v(X) = \E[\alpha_{0}(Z)|X]$. As pointed out in \cite{severini_tripathi2012}, this is a necessary condition for root-$n$ estimability of $\theta_{0}$. Moreover, by the Riesz representation theorem, the existence of such $\alpha_{0}(Z)$ is equivalent to $\E[m(W,\,\gamma)]$ being a mean square continuous functional of $\gamma$. Henceforth, we refer to $\alpha_{0}(Z)$ as a Riesz representer. \cite{newey1994} shows that mean square continuity of $\E[m(W,\,\gamma)]$ is equivalent to the semiparametric efficiency bound of $\theta_{0}$ being finite. Thus, our approach focuses on regular functionals. Similar uses of the Riesz representation theorem can be found in \cite{ai_chen2007estimation}, \cite{ackerberg2014asymptotic}, \cite{hirshberg2020debiased}, and CNS among others.

\cite{ichimura2022influence} establish the form of the orthogonal moment function for NPIV estimators
\begin{equation} \label{eq:if_npiv}
	\psi(W, \theta, \gamma, \alpha) = m(W,\gamma) - \theta + \alpha(Z)[Y - \gamma(X)],
\end{equation}
where $\alpha(Z)[Y - \gamma(X)]$ is the influence function. Note that the moment function in \eqref{eq:if_npiv} is Neyman-orthogonal to local perturbations $(\delta,\,\beta)$ of $(\gamma_{0},\,\alpha_{0})$ such that
\begin{align*}
\left.\frac{\partial}{\partial\tau}\E[\psi(W, \theta, \gamma_{0} + \tau\delta, \alpha_{0} + \tau\beta)]\right|_{\tau=0} & = \E[m(W,\,\delta)] - \E[\alpha_{0}(Z)\delta(X)]  \\ &+ \E[(Y - \gamma_{0}(X))\beta(Z)]  \\ & = 0,
\end{align*}
where the first two terms cancel out by the Riesz representation theorem and the last term is zero by the exogeneity condition. This property makes the orthogonal moment condition an excellent basis for constructing a debiased estimator of $\theta_{0}$ in the NPIV setting where estimators are typically regularized. Similar uses of the Neyman-orthogonal moment condition can be found in \cite{chen2023ann_npiv} for NPIV sieve estimators and in \cite{gautier2011high} for the high-dimensional linear IV regression.

Moreover, the exogeneity condition and iterated expectations imply
\begin{equation*}
\E[\alpha(Z)(Y - \gamma_{0}(X))] = \E[\alpha(Z)\E[Y - \gamma_{0}(X)|Z]] = 0
\end{equation*}
for any $\alpha(Z)$, meaning that the expectation of the influence function is zero regardless of $\alpha$. This implies
\begin{align*}
\mathbb{E}[\psi(W,\,\theta_{0},\,\gamma_{0},\,\alpha)] & = \mathbb{E}[m(W,\,\gamma_{0})] - \theta_{0} + \mathbb{E}\left[\alpha(Z)[Y - \gamma_{0}(X)]\right]  = 0,
\end{align*}
which allows us to use \eqref{eq:if_npiv} to estimate $\theta_{0}$. The debiased estimator $\hat\theta$ can be constructed by plugging $\hat\gamma$ and $\hat\alpha$ into the moment function $\psi(W, \theta, \gamma, \alpha)$ in place of $\gamma$ and $\alpha$ and solving for $\hat{\theta}$ from setting the sample moment $\psi(W, \theta, \hat{\gamma}, \hat{\alpha})$ to zero.

Note that the debiased estimator $\hat\theta$ requires an estimator of $\alpha_{0}$. Typically in the NPIV setting, the form of $\alpha_{0}$ is very complicated to derive or even unknown. Consider the weighted average derivative example from above. The RR is a solution to the following integral equation
\begin{equation*}
	\E[\alpha_{0}(Z)|X] = - \frac{\partial\{f_{0}(X)\omega(X)\}/\partial X_{1}}{f_{0}(X)},
\end{equation*}
where $f_{0}(X)$ is the marginal pdf of $X$. As a result, it is desirable to have a flexible approach for automatic estimation of the RR. The next subsection describes how to construct such an estimator.


\subsection{Estimation of the Riesz representer} \label{subsec:construction}

\cite{chernozhukov2022locally} show that we can exploit the orthogonality of the debiased moment function $\psi(W, \theta, \gamma, \alpha)$ to estimate $\alpha_{0}$. The Gateaux derivative of $\psi(W, \theta, \gamma, \alpha)$ in the direction $\delta$ is
\begin{align}\label{eq:pop_mom_rr_est}
\mathbb{E}[\psi_{\gamma}(W, \theta_{0}, \delta, \alpha_{0})] & =  \left.\frac{\partial}{\partial\tau}\mathbb{E}[\psi(W, \theta_{0}, \gamma_{0} + \tau \delta, \alpha_{0})]\right\vert_{\tau=0} \nonumber \\
& = \left.\frac{\partial}{\partial \tau}\E\left[m(W, \gamma_{0}  + \tau \delta) -\theta_{0} + \alpha_{0}(Z)[Y - \gamma_{0}(X) - \tau \delta(X)]\right]\right\vert_{\tau=0} \nonumber \\
& = \E[m(W, \delta) - \alpha_{0}(Z)\delta(X)] = 0,
\end{align}
where the last equality comes from $m(W,\gamma)$ being linear in $\gamma$. This can be thought of as a population moment condition for $\alpha_{0}$.

We assume that the Riesz representer estimator takes the form $\hat\alpha = b(Z)'\hat\rho$, where $b(Z)$ is a $p$-dimensional dictionary of basis functions with $p$ being possibly much larger than $n$. Let $d(X)$ be a $q$-dimensional dictionary of basis functions that represent deviations from $\gamma_{0}$. Using $d(X)$, we can construct a vector of moment conditions to estimate $\rho$. Let $d_{j}(X)$ be an element of $d(X)$, then we can form a sample moment condition corresponding to the population moment condition \eqref{eq:pop_mom_rr_est} by replacing the expectation with a sample average and $\alpha_{0}(Z)$ with $b(Z)'\rho$ to obtain
\begin{equation} \label{eq:gen_mom_rr_est}
	\hat{\psi}_{j}(W, \rho) = \frac{1}{n}\sum_{i =1}^{n}\left\{m(W_{i},\,d_{j}) - d_{j}(X_{i})b(Z_{i})'\rho\right\} = 0, \quad j=1,\dots,\,q.
\end{equation}
Note that we require $q\geq p$ to ensure identification and estimability of $\rho$. 

To allow for a high-dimensional $\alpha$ specification, we follow \cite{caner_kock2018} and use the penalized GMM (PGMM) framework. Let $\hat{\psi}(\rho) = (\hat{\psi}_{1}(W,\,\rho), \,\dots,\,\hat{\psi}_{q}(W,\,\rho))'$ where $\hat{\psi}_{j}(W,\,\rho)$ is defined in \eqref{eq:gen_mom_rr_est}. Then a solution to the PGMM problem takes the form
\begin{equation} \label{eq:pgmm_rr}
	\hat{\rho}_{L} = \argmin_{\rho \in \R^{p}} \hat{\psi}(\rho)'\hat{\Omega}_{q}\hat{\psi}(\rho) + 2\lambda_{n}|\rho|_{1},
\end{equation} 
where $\hat{\Omega}_q = \hat{\Omega}/q$, $\hat\Omega$ is a $q \times q$ positive semi-definite matrix, and $2\lambda_{n}|\rho|_{1}$ is a penalty term. This framework allows for $q \geq p > n$ and can be viewed as a Lasso extension of the standard GMM.

Let $\hat{G} = \frac{1}{n}\sum_{i =1}^{n} d(X_{i})b'(Z_{i})$ and $\hat{M} = \frac{1}{n}\sum_{i =1}^{n} m(W_{i},\,d)$ be unbiased estimators of $G = \E[d(X)b'(Z)]$ and $M=\E[m(W,\,d)]$, respectively. Then we can rewrite \eqref{eq:pgmm_rr} in matrix form as
\begin{equation} \label{eq:pgmm_rr_mat}
	\hat{\rho}_{L} = \argmin_{\rho \in \R^{p}} (\hat M - \hat G\rho)'\hat\Omega_{q}(\hat M - \hat G\rho) + 2\lambda_{n}|\rho|_{1}.
\end{equation} 
The estimator $\hat\rho_{L}$ can be interpreted as a minimum distance version of the high-dimensional GMM estimator of \cite{caner_kock2018}. Implementation details can be found in Appendix \ref{app:algo_details}.

\subsubsection{Choice of the weight matrix}
The natural choice is to use the optimal weight matrix $\hat\Omega^{\text{opt}} = \left(\frac{1}{n}\sum_{i=1}^{n}\hat\psi_i(\tilde{\rho})\hat\psi_i(\tilde{\rho})'\right)^{-1}$, where $\tilde\rho$ is the preliminary PGMM estimate based on the identity weight matrix. However, in high-dimensional scenarios $(q > n)$ this is not an appropriate choice as $\hat\Omega^{\text{opt}}$ will be rank-deficient. Instead, we propose to use a diagonal weight matrix of the form
\begin{equation} \label{eq:pgmm_diag_wm}
	\hat\Omega^{d} = \text{diag}\left( \hat\sigma^{-2}_{1}(\tilde\rho),\ldots,\hat\sigma^{-2}_{q}(\tilde\rho) \right),
\end{equation}
where $\hat\sigma^{2}_{j}(\tilde\rho) = \frac{1}{n}\Sigma_{i=1}^{n}\psi_{j}^{2}(W,\,\tilde\rho)$ is the sample variance of the $j$-th moment evaluated at the preliminary PGMM estimate. The only downside of $\hat\Omega^{d}$ is a slight efficiency loss when moments are correlated, but it still allows to handle conditional heteroskedasticity.


\subsection{Informal preview of estimation and inference results}

The estimation procedure can be summarized in the following pseudo-algorithm:
\begin{enumerate}
	\item We use cross-fitting to avoid (i) potentially severe finite sample bias due to the double use of data and (ii)  regularity conditions based on $\hat{\gamma}$ and $\hat{\alpha}$ being in Donsker class, which ML estimators are usually not. Assuming the data $\{W\}_{i=1}^{n}$ is $i.i.d.$, let $I_{\ell}$, $\ell = 1,\dots,\,L$, be a partition of the observation index set $\{1,\dots,\,n\}$ into $L$ distinct subsets of about equal size. Let $n_{\ell}$ denote the number of observations in fold $\ell$.
	\item For each data fold $\ell=1,\dots,\,L$, we obtain estimates $\hat\gamma_{\ell}$ and $\hat\alpha_{\ell}$ that are constructed from the observations not in $I_{\ell}$. In particular, the RR estimate is of the form $\hat\alpha_{\ell} = b(Z)'\hat\rho_{\ell}$, where
	\begin{equation*}
		\hat{\rho}_{\ell} = \argmin_{\rho \in \R^{p}} (\hat M_{\ell} - \hat G_{\ell}\rho)'\hat\Omega_{q}(\hat M_{\ell} - \hat G_{\ell}\rho) + 2\lambda_{n}|\rho|_{1},
	\end{equation*}
	with $\hat{G}_{\ell} = \frac{1}{n - n_{\ell}}\sum_{i\not\in I_{\ell}} d(X_{i})b'(Z_{i})$ and $\hat{M}_{\ell} = \frac{1}{n-n_{\ell}}\sum_{i\not\in I_{\ell}} m(W_{i},\,d)$.
	\item We construct the estimator $\hat{\theta}$ by setting the sample average of $\psi(W, \theta, \hat\gamma_{\ell}, \hat{\alpha}_{\ell})$ to zero and solving for $\theta$. This estimator $\hat\theta$ and the associated asymptotic variance estimator $\hat V$ have the following explicit forms
	\begin{align} \label{eq:dml}
		\hat{\theta} & = \frac{1}{n}\sum_{\ell=1}^{L}\sum_{i \in I_{\ell}}\{m(W_{i}, \hat{\gamma}_{\ell}) + \hat{\alpha}_{\ell}(Z_{i})[Y_{i} - \hat{\gamma}_{\ell}(X_{i})]\} \\
		\hat V & = \frac{1}{n}\sum_{\ell=1}^{L}\sum_{i \in I_{\ell}}\hat{\psi}_{i\ell}^{2},\; \hat{\psi}_{i\ell} = m(W_i, \hat{\gamma}_{\ell}) - \hat{\theta} + \hat{\alpha}_{\ell}(Z_i)[Y_i - \hat{\gamma}_{\ell}(X_i)] \notag .
	\end{align}
\end{enumerate}

Next, we informally discuss the key conditions behind the asymptotic normality result. Since $\hat\theta$ is constructed by plugging-in $\hat\gamma$ and $\hat\alpha$ in the orthogonal moment condition, asymptotic properties of $\hat\theta$ depend on the asymptotic behavior of $\hat\gamma$ and $\hat\alpha$. First, to allow for a wide range of MLIV estimators, we assume that $\hat\gamma$ satisfies some projected mean square convergence rate condition as an estimator of $\gamma_{0}$. Specifically, we require
\begin{equation*}
	||T(\hat\gamma - \gamma_{0})|| = O_p(\kappa^{\gamma}_{n}),
\end{equation*}
where $\kappa_{n}^{\gamma}$ can be slower than root-$n$ rate\footnote{The result also holds for the standard mean square rate condition, i.e. $||\hat\gamma - \gamma_{0}|| = O_p(\kappa_{n}^{\gamma})$, however, for NPIV/MLIV estimators this rate is slower due to ill-posedness.}. As pointed out in Section \ref{subsec:npiv}, it is possible to obtain a fast rate under the projected mean square norm. Hence, it is a weak high-level assumption that can be satisfied by a variety of MLIV estimators such as Double Lasso \citep{gold2020}, Kernel IV \citep{singh2019kiv}, and a series of estimators constructed using the minimax framework of \cite{dikkala2020minimax} among others. 

The second condition is the mean square convergence rate of $\hat\alpha$. For the ease of exposition, assume that $\hat\alpha$ satisfies the following mean square convergence rate condition,
\begin{equation*}
	||\hat\alpha - \alpha_{0}|| = O_p(\kappa^{\alpha}_{n}).
\end{equation*}
We derive an exact expression for $\kappa_{n}^{\alpha}$ in Section \ref{sec:pgmm_prop}. 

Finally, under quite standard regularity conditions  asymptotic normality can be established provided that
\begin{equation*}
	\sqrt{n}||\hat\alpha - \alpha_{0}||\;||T(\hat\gamma - \gamma_{0})|| \xrightarrow{p} 0, 
\end{equation*}
which is satisfied when $\sqrt{n}\kappa_{n}^{\gamma}\kappa_{n}^{\alpha} \rightarrow 0$. Hence, there is a trade-off between the convergence rates of $\hat\gamma$ and $\hat\alpha$. It is possible to allow for a slower convergence rate of $\hat\gamma$ at the expense of a faster convergence rate of $\hat\alpha$ and vice versa.


\section{Properties of the PGMM estimator}\label{sec:pgmm_prop}

In this section we provide the mean square convergence rate for the PGMM estimator $\hat{\alpha}$ which is necessary for the asymptotic analysis of $\hat\theta$. We start by introducing some conditions.

\begin{assumption} \label{ass:weight_matrix}
There exists a sequence of non-random matrices $\Omega$ such that 
\begin{enumerate}
	\item[(i)] $||\hat{\Omega} - \Omega||_{\ell_\infty} = O_p(\varepsilon_{n}^{\Omega})$, where $\varepsilon_{n}^{\Omega} = \sqrt{\frac{\log q}{n}}$,
	\item[(ii)] $||\Omega||_{\ell_\infty} \leq C < \infty$ for some constant $C$.
\end{enumerate}
\end{assumption}
The first part of Assumption \ref{ass:weight_matrix} is pretty standard and requires a consistent estimate of the weight matrix. The rate $O_p(\sqrt{\log q / n})$ is the natural high-dimensional extension of the classic $O_p(n^{-1/2})$ rate \citep[see e.g.,][]{bickel2008covariance}. The second part of the assumption, as discussed in \cite{caner_kock2018}, might be restrictive as it requires a high-dimensional matrix to be uniformly bounded in $\ell_{\infty}$- norm, but for the notational convenience we keep it. Assumption \ref{ass:weight_matrix} is trivially satisfied by the identity weight matrix. We also demonstrate that the main result still follows though under the diagonal weight matrix and relaxed assumptions.

Note that the convergence rate of the PGMM estimator defined in \eqref{eq:pgmm_rr_mat} depends on the convergence rates of $\hat\Omega$, $\hat G$, and $\hat M$. Assumption \ref{ass:weight_matrix} ensures that $\hat\Omega$ is consistent. To obtain a convergence rate for $\hat G$, we impose the following condition.
\begin{assumption} \label{ass:boundedness}
There are constants $C_{b}$ and $C_{d}$ such that with probability approaching one,
\begin{equation*}
	\max_{1 \leq j \leq p}|b_{j}(Z)| \leq C_{b} \text{ and } \max_{1 \leq j \leq q}|d_{j}(X)| \leq C_{d}.
\end{equation*}
\end{assumption}
This condition implies  
\begin{equation*}
	||\hat G - G||_{\infty} = O_p(\varepsilon^{G}_{n}),\text{ where } \varepsilon^{G}_{n}=\sqrt{\frac{\log q}{n}}.
\end{equation*}
Unlike the standard Lasso, the second moment matrix convergence rate depends on the number of moments, i.e. the number of elements in $d(X)$, rather than the number of elements in $b(Z)$. 
\begin{assumption} \label{ass:M_conv}
There is $\varepsilon_{n}^{M}$ such that 
\begin{equation*}
	||\hat M - M||_{\infty} = O_{p}(\varepsilon_{n}^{M}) \text{ and } \varepsilon_{n}^{M} = o(1).
\end{equation*}
\end{assumption}

Next, we impose a sparse approximation condition for $\alpha_{0}$.
\begin{assumption} \label{ass:sparse_approx}
There exist $C > 1$ and $\bar{\rho}$ with $\bar{s}$ non-zero elements such that
\begin{equation*}
	||\alpha_{0} - b'\bar{\rho}||^{2} \leq C\bar{s}\varepsilon^{2}_{n},
\end{equation*}
where $\varepsilon_{n} = \max\{\varepsilon^{G}_{n},\, \varepsilon^{M}_{n}\}$.
\end{assumption} 
Intuitively, this assumption controls the squared approximation error from using the linear combination $b'\bar\rho$ to approximate $\alpha_{0}$. Note that Assumption \ref{ass:sparse_approx} does not necessarily require $\alpha_{0}$ to be equal to the linear combination of $\bar{s}$ terms, it states that there exists a sparse $\bar\rho$ with $\bar{s}$ non-zero elements such that the approximation error is bounded by $C\bar{s}\varepsilon^{2}_{n}$. In other words, Assumption \ref{ass:sparse_approx} is general enough to accommodate both exact and  approximate sparsity of $\alpha_{0}$. Approximate sparsity allows for a large number of potential regressors (possibly much larger than the sample size) when relatively few important regressors give a good approximation but the identity of those few is not known, which is different from a standard series approximation where typically the first $\bar{s}$ regressors are assumed to achieve a good approximation \citep{bradic2021minimax}. Thus, very sparse approximations allow to keep $\bar{s}$ relatively small which results in faster convergence rates. For a more detailed discussion of approximation bias conditions we refer the reader to CNS. 

Let $S = \{1,\dots,\,p\}$, $S_{\rho}$ be a subset of $S$ with $\rho_{j} \neq 0$, and $S_{\rho}^{c}$ be the complement of $S_{\rho}$ in $S$. Let $\rho_{L}$ be the population coefficients, i.e.
\begin{equation*}
	\rho_{L} = \argmin_{\rho\in\R^{p}} (M - G\rho)'\Omega_{q}(M - G\rho) + 2\varepsilon_{n}|\rho|_{1}.
\end{equation*}
The PGMM estimator $\hat\rho_{L}$ estimates the population coefficients $\rho_{L}$, which in turn might be different from the approximation coefficients $\bar\rho$. The following condition is essential to derive the oracle inequality for $\hat\rho_{L}$, and hence, the convergence rate for $\hat\alpha_{L}=b'\hat\rho_{L}$.

\begin{assumption} \label{ass:re_pgmm}
Let $G'\Omega_{q} G$ have its largest eigenvalue uniformly bounded in $n$ and
\begin{equation*}
\phi^{2}(s) = \inf\left\{\frac{\delta'G'\Omega_{q} G\delta}{||\delta_{S_{\rho}}||^{2}}:\delta\in\R^{p} \backslash \{0\},\; |\delta_{S_{\rho}^{c}}|_{1} \leq 3|\delta_{S_{\rho}}|_{1},\, |S_{\rho}| \leq s\right\} \geq c > 0.
\end{equation*}
\end{assumption}

Assumption \ref{ass:re_pgmm} is the modified population restricted eigenvalue condition as in \cite{caner_kock2018}. To accommodate for the PGMM estimator the condition is imposed on $G'\Omega_{q} G$ rather than $\E[b(Z)b'(Z)]$ as in the classic restricted eigenvalue condition of \cite{bickel2009}. Showing that its empirical counterpart is bounded uniformly away from zero will be used to put a bound on the estimation error of $\hat\alpha_{L}$.

\begin{assumption} \label{ass:m_bound}
There is $C > 0$ such that with probability approaching one,
\begin{equation*}
	\max_{1\leq j\leq q}|m(W,\,d_{j})| \leq C.
\end{equation*}
\end{assumption}
This condition is needed to put a bound on $||M||_{\infty}$ which is necessary to establish the oracle inequality for $\hat\rho_{L}$, and hence, the convergence rate for $\hat\alpha_{L}$. Moreover, note that by Assumption \ref{ass:m_bound}, $\varepsilon_{n} = \varepsilon^{M}_{n} = \varepsilon^{G}_{n} = \sqrt{\log q/n}$. This simplifies the analysis, but is not necessary for establishing the results below. Also, let $|\bar\rho|_{1} \leq \bar{A} < \infty$. We can allow for the norm to grow with $n$ at a certain rate, however, it does not change the main results, hence, for simplicity we put a bound on $|\bar\rho|_{1}$.

\begin{theorem}\label{thm:rr_bound}
Suppose Assumptions \ref{ass:weight_matrix}--\ref{ass:m_bound} are satisfied. Let $\varepsilon_{n} = o(\lambda_{n})$ and $\bar{s}\lambda_n = o(1)$, then
\begin{equation*}
	||\hat\alpha_{L} - \alpha_{0}|| = O_p(\kappa_{n}^{\alpha}), \text{ where } \kappa_{n}^{\alpha} = \bar{s}\lambda_{n}.
\end{equation*}
\end{theorem}

The presence of endogeneity results in a slower rate of convergence for the RR estimator compared to the exogenous counterpart in CNS. The MD Lasso estimator of CNS converges at the $\sqrt{\bar{s}}\lambda_{n}$ rate, while the PGMM estimator is slower by a factor of $\sqrt{\bar{s}}$. Note that the convergence rate only depends on the number of approximation elements $\bar{s}$, but is independent of the number of relevant moments. 

\begin{theorem}\label{thm:rr_bound_diag_wm}
Suppose Assumptions \ref{ass:boundedness}--\ref{ass:m_bound} are satisfied. Assume that populations moments are non-degenerate and have non-zero variance. Let $\varepsilon_n = o(\lambda_n)$ and $\bar{s}^{2}\lambda_n = o(1)$. Then, under the diagonal weight matrix \eqref{eq:pgmm_diag_wm},
\begin{equation*}
	||\hat\alpha_{L} - \alpha_{0}|| = O_p(\kappa^{\alpha}_n), \text{ where } \kappa^{\alpha}_n = \bar{s}\lambda_n.
\end{equation*}
\end{theorem}

We relax the convergence rate condition in Assumption \ref{ass:weight_matrix} by allowing for a diagonal weight matrix. As a result, we have to increase the penalty to ensure $\bar{s}^{2}\lambda_n \rightarrow 0$. Otherwise, the main result remains the same.


\section{Inference for linear functionals} \label{sec:asy_prop_lin}

In this Section, we provide conditions ensuring root-$n$ consistency and asymptotic normality of the debiased estimator $\hat\theta$. Under the specified conditions, we can do inference in a standard way. First, we focus on linear functionals and then provide additional conditions to extend the results to nonlinear functionals in Section \ref{sec:asy_prop_nonlin}.

We impose the following conditions.
\begin{assumption}\label{ass:asy_primitives}
$\alpha_{0}(z)$ and $\E\left[[y - \gamma_{0}(x)]^{2}|z\right]$ are bounded and $\E\left[m(w,\,\gamma_{0})^{2}\right] < \infty$.
\end{assumption}

\begin{assumption}\label{ass:mse_consist_gamma}
$\int[m(w,\,\hat\gamma) - m(w,\,\gamma_{0})]^{2}F_{0}(dw) \xrightarrow{p} 0$ and $||\hat\gamma - \gamma_{0}|| \xrightarrow{p} 0$.
\end{assumption}

\begin{assumption}\label{ass:conv_rate_gamma_lin}
$||T(\hat\gamma - \gamma_{0})|| = O_p(\kappa_{n}^{\gamma})$ and $\kappa_{n}^{\gamma} = o(1)$.
\end{assumption}

Assumption \ref{ass:asy_primitives} is purely technical, and we maintain it for simplicity. Assumption \ref{ass:mse_consist_gamma} allows for estimators $\hat\gamma$ that are mean square consistent. Assumption \ref{ass:conv_rate_gamma_lin} requires $\hat\gamma$ to converge to $\gamma_{0}$ in the projected norm at a rate equal to $\kappa_{n}^{\gamma}$ which is typically slower than root-$n$. Note that this condition is weaker than convergence in standard mean square norm (see Section \ref{subsec:npiv}). This specification is general enough and allows for various MLIV estimators.

\begin{assumption}\label{ass:conv_rates_lin}
$\sqrt{n}\kappa_{n}^{\alpha}\kappa_{n}^{\gamma} = o(1)$.	
\end{assumption}
This condition is sufficient to guarantee $\sqrt{n}||\hat\alpha_{L}-\alpha_{0}||\,||T(\hat\gamma - \gamma_{0})|| \xrightarrow{p} 0$, leading to asymptotic normality of $\hat\theta$. 

\begin{theorem}\label{thm:asy_norm_lin}
Suppose Assumptions \ref{ass:weight_matrix}--\ref{ass:conv_rates_lin} are satisfied. Let $\varepsilon_n = o(\lambda_n)$ and $\bar{s}\lambda_n = o(1)$. Then for $\psi_{0}(w) = m(w,\,\gamma_{0}) - \theta_{0} + \alpha_{0}(z)[y - \gamma_{0}(x)]$,
\begin{equation*}
	\sqrt{n}(\hat\theta - \theta_0) \xrightarrow{d} \cN(0,\,V) \text{ and } \hat V \xrightarrow{p} V = \E\left[\psi_{0}^{2}(w)\right].
\end{equation*}
\end{theorem}

\begin{corollary}\label{cor:asy_norm_lin_wm}
Suppose Assumptions \ref{ass:boundedness}--\ref{ass:conv_rates_lin} are satisfied. Assume that populations moments are non-degenerate and have non-zero variance, and the PGMM estimator uses the diagonal weight matrix \eqref{eq:pgmm_diag_wm}. Let $\varepsilon_n = o(\lambda_n)$ and $\bar{s}^{2}\lambda_n = o(1)$. Then for $\psi_{0}(w) = m(w,\,\gamma_{0}) - \theta_{0} + \alpha_{0}(z)[y - \gamma_{0}(x)]$,
\begin{equation*}
	\sqrt{n}(\hat\theta - \theta_0) \xrightarrow{d} \cN(0,\,V) \text{ and } \hat V \xrightarrow{p} V = \E\left[\psi_{0}^{2}(w)\right].
\end{equation*}
\end{corollary}


\section{Nonlinear functionals} \label{sec:asy_prop_nonlin}

It is possible to extend the results from Section \ref{sec:asy_prop_lin} to allow for estimation of $\theta_{0} = \E[m(W,\,\gamma_{0})]$ for nonlinear $m(W,\,\gamma)$. The estimator is similar to the linear case except we estimate the RR of the linearization of $m(W,\,\gamma)$ leading to a different $\hat M$ needed. In this section, we show how to construct such an estimator and provide additional conditions that are sufficient for valid asymptotic inference for nonlinear functionals. As we mentioned in the introduction, due to nonlinearity of $m(W,\,\gamma)$, we have to impose restrictions on the convergence rate of $\hat\gamma$ in terms of the standard mean square norm, not the projected norm as in the linear case. We provide more details below.

To account for nonlinearity of $m(W,\,\gamma)$ in $\gamma$, we assume linearity of the Gateaux derivative of a nonlinear functional. To be more precise, let $\zeta$ be a deviation from $\gamma$. We assume that $m(W,\,\gamma)$ is Gateaux differentiable with the derivative $D(W,\,\gamma,\,\zeta)$, meaning that
\begin{equation*}
	D(W,\,\gamma,\,\zeta) = \left.\frac{d}{d\tau}m(W,\, \gamma + \tau \zeta)\right\vert_{\tau=0}
\end{equation*}
for a scalar $\tau$, and that $D(W,\,\gamma,\,\zeta)$ is linear in $\zeta$. Moreover, assume that $\alpha_{0}(Z)$ satisfies
\begin{equation} \label{eq:rr_exog}
	\mathbb{E}[D(W,\,\gamma_{0},\,\zeta)] = \mathbb{E}[\alpha_{0}(Z)\zeta(X)], \; \text{for all} \; \zeta(X) \; \text{with} \; \mathbb{E}[\zeta^{2}(X)] < \infty.
\end{equation}
In other words, Equation \eqref{eq:rr_exog} implies that $D(W,\,\gamma,\,\zeta)$ is a mean-square continuous functional of $\zeta$, which corresponds to Assumption 3 of \cite{ichimura2022influence}, meaning that $\alpha_{0}(Z)$ is a Riesz representer of the Gateaux derivative of $m(W,\gamma)$ with respect to $\gamma$ evaluated at $\gamma=\gamma_{0}$. Thus, by the Riesz representation theorem, for $D(W, \gamma_{0}, d) = (D(W, \gamma_{0}, d_{1}),\dots,\,D(W, \gamma_{0}, d_{q}))'$,
\begin{equation*}
	M = \mathbb{E}[D(W, \gamma_{0}, d)] = \mathbb{E}[\alpha_{0}(Z)d(X)].
\end{equation*} 

We can construct an estimator $\hat\theta$ exactly like in Equation \eqref{eq:dml} except we need a different estimator of $\alpha_{0}(Z)$ based on \eqref{eq:rr_exog}. Despite $\gamma$ enters $m(W,\,\gamma)$ nonlinearly, the estimator will still have zero first-order bias and be root-$n$ consistent and asymptotically normal under suficient regularity conditions. See \cite{newey1994}, \cite{ichimura2022influence}, and \cite{chernozhukov2022locally} for more details.

An estimator $\hat{\alpha}_{\ell}$ can be constructed exactly as described in Section \ref{subsec:construction} except being based on a different $\hat M_{\ell}$, where it is convenient to bring back the $\ell$ subscript. Let $\hat{\gamma}_{\ell,\ell'}$ be based on observations not in either $I_{\ell}$ or $I_{\ell'}$, then the unbiased estimator $\hat{M}_{\ell}$ is given by
\begin{align*}
	\hat{M}_{\ell} & = (\hat{M}_{\ell1},\,\dots,\,\hat{M}_{\ell q.})' \\
	\hat{M}_{\ell j} & = \frac{1}{n - n_{\ell}} \sum_{\ell' \neq \ell} \sum_{i \in I_{\ell'}}D(W_{i}, \hat{\gamma}_{\ell,\ell'}, d_{j}),
\end{align*}
where $\hat{M}_{\ell j}$ is the Gateaux derivative of the moment function with respect to $\gamma$ in the direction of the $j^{\text{th}}$ dictionary function. This estimator uses further sample splitting where $\hat{M}$ is constructed by averaging over observations that are not used in $\hat{\gamma}_{\ell, \ell'}$. This additional sample splitting allows $\hat{M}_{\ell}$ to depend on an estimator of $\gamma$ as required when $m(W,\,\gamma)$ is nonlinear in $\gamma$. 

To establish the convergence rate for $\hat M_{\ell}$, we impose the following condition.
\begin{assumption}\label{ass:M_bound_nonlin}
There exist $C$, $\varepsilon > 0$ such that for any $\gamma$ with $||\gamma - \gamma_{0}|| \leq \varepsilon$: 
\begin{enumerate}
	\item [(i)] $\max_{1\leq j \leq q}|D(W,\,\gamma,\,d_{j})| \leq C$;
	\item [(ii)] $\sup_{1 \leq j \leq q}\left|\E[D(W,\,\gamma,\,d_{j}) - D(W,\,\gamma_{0},\,d_{j})]\right| \leq C||\gamma - \gamma_{0}||$.
\end{enumerate}	
\end{assumption}

\begin{lemma}\label{lem:M_bound_nonlin}
Suppose that $||\hat{\gamma}_{\ell,\,\ell'} - \gamma_{0}||=O_p(\kappa_{n}^{\gamma})$ for $\ell,\,\ell' = 1,\dots,\,L$, and Assumption \ref{ass:M_bound_nonlin} is satisfied, then
\begin{equation*}
||\hat{M}_{\ell} - M_{\ell}||_{\infty} = O_p(\kappa_{n}^{\gamma}).
\end{equation*}	
\end{lemma}

As CNS point out, the presence of the initial estimator $\hat\gamma_{\ell,\ell'}$ in $\hat M_{\ell}$ makes the convergence rate of $||\hat{M}_{\ell} - M_{\ell}||_{\infty}$ slower, $\kappa_{n}^{\gamma}$ instead of $\sqrt{\log q/n}$. Thus, $\varepsilon_{n} = \varepsilon^{M}_{n} = \kappa_{n}^{\gamma}$, which requires $\lambda_{n}$ to converge to zero slightly slower than $\kappa_{n}^{\gamma}$. Though it affects the convergence rate of $\hat\alpha$, the main results of Section~\ref{sec:pgmm_prop} still hold (see Appendix \ref{app:pgmm_nonlin}). 

\begin{assumption}\label{ass:M_frechet}
There exist $C$, $\varepsilon > 0$ such that for any $\gamma$ with $||\gamma - \gamma_{0}|| \leq \varepsilon$,
\begin{equation*}
	\left| \E[m(W,\,\gamma) - m(W,\,\gamma_{0}) - D(W,\,\gamma_{0},\,\gamma-\gamma_{0})] \right| \leq C||\gamma-\gamma_{0}||^{2}.
\end{equation*}
\end{assumption}
This condition controls the size of the linearization remainder in a linearization using the Gateaux derivative. It implies that $\E[m(W,\,\gamma)]$ is Frechet-differentiable in $||\gamma - \gamma_{0}||$ at $\gamma_{0}$ with derivative $\E[D(W,\,\gamma_{0},\,\gamma-\gamma_{0})]$.

\begin{assumption}\label{ass:conv_rate_gamma_nonlin}
$||\hat\gamma - \gamma_{0}|| = O_p(\kappa_{n}^{\gamma})$ and $n^{1/4}||\hat\gamma - \gamma_{0}|| \xrightarrow{p} 0$.
\end{assumption}

It is a standard assumption to accommodate for nonlinearity of $m(W,\,\gamma)$. This might be a very tight restriction to satisfy given overall slow convergence rates of NPIV estimators, especially in the severely ill-posed case. However, as discussed in CNS, it is not known whether it is possible to weaken the $n^{-1/4}$ condition for nonlinear functionals, which goes back to \cite{newey1994}.

\begin{theorem}\label{thm:asy_norm_nonlin}
Suppose Assumptions \ref{ass:weight_matrix}--\ref{ass:boundedness}, \ref{ass:sparse_approx}--\ref{ass:re_pgmm}, \ref{ass:asy_primitives}, \ref{ass:conv_rates_lin}, and \ref{ass:M_bound_nonlin}--\ref{ass:conv_rate_gamma_nonlin} are satisfied. Let $\kappa_{n}^{\gamma} = o(\lambda_n)$ and $\bar{s}\lambda_{n} = o(1)$. Then for $\psi_{0}(w) = m(w,\,\gamma_{0}) - \theta_{0} + \alpha_{0}(z)[y - \gamma_{0}(x)]$,
\begin{equation*}
	\sqrt{n}(\hat\theta - \theta_0) \xrightarrow{d} \cN(0,\,V) \text{ and }  \hat V \xrightarrow{p} V = \E\left[\psi_{0}^{2}(w)\right].
\end{equation*}
\end{theorem}

\begin{corollary}\label{cor:asy_norm_nonlin_wm}
Suppose Assumptions \ref{ass:boundedness}, \ref{ass:sparse_approx}--\ref{ass:re_pgmm}, \ref{ass:asy_primitives}, \ref{ass:conv_rates_lin}, and \ref{ass:M_bound_nonlin}--\ref{ass:conv_rate_gamma_nonlin} are satisfied. Assume that populations moments are non-degenerate and have non-zero variance, and the PGMM estimator uses the diagonal weight matrix \eqref{eq:pgmm_diag_wm}. Let $\kappa_n^{\gamma} = o(\lambda_n)$ and $\bar{s}^{2}\lambda_n = o(1)$. Then for $\psi_{0}(w) = m(w,\,\gamma_{0}) - \theta_{0} + \alpha_{0}(z)[y - \gamma_{0}(x)]$,
\begin{equation*}
	\sqrt{n}(\hat\theta - \theta_0) \xrightarrow{d} \cN(0,\,V) \text{ and } \hat V \xrightarrow{p} V = \E\left[\psi_{0}^{2}(w)\right].
\end{equation*}
\end{corollary}


\section{Finite Sample Performance}\label{sec:toy_mc}

This section presents Monte Carlo evidence illustrating finite sample performance of our automatic debiasing procedure. We focus on the average derivative functional and compare three approaches: (i) the plug-in estimator, (ii) debiased ML using the analytical Riesz representer, and (iii) our adaptive debiased ML using the PGMM-estimated Riesz representer.

Our design builds on \cite{newey_powell2003}, \cite{santos2012inference}, and \cite{ChenPouzo2015QLR}, modified to allow for multiple regressors and instruments. We generate i.i.d.\ draws 
\begin{equation*}
	\begin{pmatrix}
		X_{ij} \\ Z_{ij} \\ u_{ij}
	\end{pmatrix}
	\sim \cN
	\begin{pmatrix}
		\begin{bmatrix}
			0 \\ 0 \\ 0
		\end{bmatrix}
		, \;
		\begin{bmatrix}
			1 & 0.8 & 0.5 \\
			0.8 & 1 & 0 \\
			0.5 & 0 & 1
		\end{bmatrix}
	\end{pmatrix},
	\quad j=1,\dots,\,k,
\end{equation*}
where $k$ denotes the number of regressors and instruments. The response variable is
\begin{equation*}
	Y_{i} = \gamma(X_{i}) + v_{i}, \quad v_{i} = \sum_{j=1}^{k}u_{ij}.
\end{equation*}
The composite error structure implies that endogeneity weakens as $k$ increases, since each regressor $X_j$ is correlated with only one component of $v_i$. We therefore restrict attention to $k \leq 10$.

Following \cite{bakhitov2025mliv}, we specify the structural function as
\begin{equation*}
	\gamma(X) = X_{1} + \exp\{-0.5X_{-1}'X_{-1}\},
\end{equation*}
where $X_{-1} = (X_2, \ldots, X_k)'$. This specification ensures that the parameter of interest,
\begin{equation*}
	\theta_0 = \E\left[\frac{\partial \gamma(X)}{\partial X_1}\right] = 1,
\end{equation*}
is constant across simulations.

We compare three estimators. The plug-in (PI) estimator averages the estimated derivative without bias correction. The debiased ML (DML) estimator employs the analytical Riesz representer, specifically, the identity score estimator of \cite{chen2023ann_npiv}. The automatic debiased ML (ADML) estimator uses our PGMM framework to learn the Riesz representer directly from the orthogonality conditions.

For all methods, we construct dictionaries $b(Z)$ and $d(X)$ using cubic polynomials with interactions, yielding $p = q$ basis functions. The structural function $\gamma$ is estimated using the Double Lasso estimator of \cite{gold2020}. We report results from $1000$ replications for $k \in \{2, 5, 10\}$ and $n \in \{100, 500, 1000, 10000\}$, using five-fold cross-fitting throughout. Table \ref{tab:mc_wad} reports absolute bias, median standard error, and coverage of a nominal 95\% confidence interval.

\begin{table}[htbp]
\centering
\begin{threeparttable}
\caption{Monte Carlo Results: Average Derivative}
\label{tab:mc_wad}
\small
\begin{tabular}{ll| ccc| ccc| ccc}
\toprule
 & & \multicolumn{3}{c}{$|\text{Bias}|$} & \multicolumn{3}{c}{Median SE} & \multicolumn{3}{c}{Coverage} \\
\cmidrule(lr){3-5} \cmidrule(lr){6-8} \cmidrule(lr){9-11}
$k$ & $n$ & PI & DML & ADML & PI & DML & ADML & PI & DML & ADML \\
\midrule
2 & 100 & 0.043 & 1.020 & 0.178 & 0.093 & 0.506 & 0.198 & 64.3\% & 95.7\% & 93.5\% \\
  & 500 & 0.053 & 0.016 & 0.021 & 0.014 & 0.069 & 0.063 & 29.4\% & 95.6\% & 94.6\% \\
  & 1000 & 0.038 & 0.000 & 0.001 & 0.006 & 0.044 & 0.041 & 16.8\% & 95.9\% & 95.0\% \\
  & 10000 & 0.011 & 0.002 & 0.001 & 0.001 & 0.013 & 0.013 & 4.6\% & 94.2\% & 94.4\% \\
\hline
5 & 100 & 0.074 & 0.045 & 0.039 & 0.057 & 2.957 & 0.162 & 53.6\% & 95.0\% & 92.3\% \\
  & 500 & 0.085 & 0.036 & 0.015 & 0.022 & 0.084 & 0.064 & 24.9\% & 96.7\% & 94.6\% \\
  & 1000 & 0.071 & 0.010 & 0.006 & 0.011 & 0.049 & 0.043 & 13.1\% & 96.3\% & 93.8\% \\
  & 10000 & 0.022 & 0.000 & 0.000 & 0.001 & 0.013 & 0.013 & 1.5\% & 96.5\% & 95.8\% \\
\hline
10 & 100 & 0.070 & 0.242 & 0.001 & 0.055 & 39.717 & 0.165 & 54.4\% & 94.5\% & 90.0\% \\
   & 500 & 0.013 & 0.052 & 0.002 & 0.013 & 0.129 & 0.056 & 41.4\% & 93.3\% & 93.5\% \\
   & 1000 & 0.046 & 0.009 & 0.002 & 0.010 & 0.056 & 0.040 & 21.6\% & 95.9\% & 93.6\% \\
   & 10000 & 0.033 & 0.000 & 0.001 & 0.001 & 0.013 & 0.013 & 0.4\% & 94.6\% & 93.3\% \\
\bottomrule
\end{tabular}
\begin{tablenotes}
\footnotesize
\item PI = plug-in estimator; DML = debiased estimator with analytical Riesz representer; ADML = automatic debiased estimator with PGMM-estimated Riesz representer. Results based on 1000 replications, and $\gamma$ is estimated using Double Lasso.
\end{tablenotes}
\end{threeparttable}
\end{table}

The plug-in estimator exhibits substantial bias across all specifications. More strikingly, the coverage probability deteriorates as the sample size increases: for $k=2$, coverage falls from $64.3\%$ at $n=100$ to $4.6\%$ at $n=10000$; for $k=10$, it drops from $54.4\%$ to $0.4\%$. This reflects the well-known phenomenon that the squared bias of plug-in estimators shrinks slower than the variance \citep{chernozhukov2022locally}, causing the standard errors to increasingly understate the true uncertainty.

The analytical debiasing procedure performs well in large samples but exhibits instability when $n$ is small. At $n=100$, the absolute bias reaches $1.020$ for $k=2$ and $0.242$ for $k=10$, accompanied by severely inflated standard errors (the median SE reaches $2.957$ for $k=5$ and $39.717$ for $k=10$). This occurs because the identity score estimator requires basis function expansions for intermediate estimation steps\footnote{See Section 4.2 in \cite{chen2023ann_npiv}.}, which become numerically unstable when the effective sample size is small relative to the dimension. Despite the elevated bias and variance, DML maintains close-to-nominal coverage ($93$--$97\%$) as the inflated standard errors appropriately capture the increased uncertainty, consistent with the slight conservatism documented in \cite{chen2023ann_npiv}.

The automatic debiasing procedure demonstrates substantially greater stability in small samples. At $n=100$, ADML achieves bias of $0.178$, $0.039$, and $0.001$ for $k=2$, $5$, and $10$, with moderate standard errors ($0.198$, $0.162$, $0.165$) that avoid the inflation seen in the analytical approach. ADML correctly quantifies its uncertainty, yielding coverage probabilities close to the nominal level. The procedure exhibits slight undercoverage ($90$--$93.5\%$ at $n=100$), contrasting with DML's conservatism and reflecting different bias-variance tradeoffs. As sample size increases, both procedures converge to similar performance, with ADML delivering comparable or lower bias throughout.


\section{Application to nonparametric demand estimation}\label{sec:np_demand}

\subsection{Model and estimation framework}\label{subsec:semiparam_demand}

In this Section, we introduce a new framework for demand estimation that combines the nonparametric identification arguments of \cite{berry2014identification} with the dimensionality reduction techniques of \cite{gandhi_houde2019}, making it applicable to real data sets with more than two products unlike \cite{compiani2022market} whose approach suffers from the curse of dimensionality.

We follow \cite{berry2014identification} and present a general model of demand first, later on we will impose additional restrictions on the form of the indirect utility function. In market $t$, $t=1,\dots,\,T$, there is a continuum of consumers choosing from a set of products $\mathcal{J} = \{0,\,1,\dots,\,J\}$ which includes the outside option. The choice set in market $t$ is characterized by a set of product characteristics $\chi_{t}$ partitioned as follows:
\begin{equation*}
	\chi_{t} \equiv (x_{t}, p_{t}, \xi_{t}),
\end{equation*}
where $x_{t} \equiv (x_{1t}, \dots, x_{Jt})$ is a vector of exogenous observable characteristics (e.g. exogenous product characteristics or market-level income), $p_{t} \equiv (p_{1t}, \dots, p_{Jt})$ are observable endogenous characteristics (typically, market prices) and $\xi_{t} \equiv (\xi_{1t}, \dots, \xi_{Jt})$ represent unobservables potentially correlated with $p_{t}$ (e.g. unobserved product quality). Let $\mathcal{X}$ denote the support of $\chi_{t}$. Then the structural demand system is given by
\begin{equation*}
	\sigma: \mathcal{X} \mapsto \Delta^{J},
\end{equation*}
where $\Delta^{J}$ is a unit $J$-simplex. The function $\sigma$ gives, for every market $t$, the vector $s_{t}$ of shares for the $J$ goods.

Following \cite{berry2014identification}, we partition the exogenous characteristics as $x_{t} = \left(x_{t}^{(1)}, x_{t}^{(2)} \right)$, where $x_{t}^{(1)} \equiv \left(x_{1t}^{(1)}, \dots, x_{Jt}^{(1)}\right)$, $x_{jt} \in \mathbb{R}$ for $j \in \mathcal{J} \backslash \{0\}$, and define the linear indices
\begin{equation*}
	\delta_{jt} = x_{jt}^{(1)}\beta_{j} + \xi_{jt}, \quad j \in \mathcal{J} \backslash \{0\},
\end{equation*}
and let $\delta_{t} \equiv (\delta_{1t}, \dots, \delta_{Jt})$. Without loss of generality, we can normalize $\beta_{j} = 1$ for all $j$ as unobserved characteristics $\xi_{jt}$ have no natural scale (see \cite{berry2014identification} for more details). Given the definition of the demand system, for every market $t$,
\begin{equation*}
	\sigma(\chi_{t}) = \sigma\left(\delta_{t}, p_{t}, x_{t}^{(2)}\right).
\end{equation*}

Following \cite{berry2013connected} and \cite{berry2014identification}, under the connected substitutes assumption, there exists at most one vector $\delta_{t}$ such that $s_{t} = \sigma\left(\delta_{t}, p_{t}, x_{t}^{(2)}\right)$, meaning that for every inside good we can write
\begin{equation} \label{eq:inv_demand}
	\delta_{jt} = \sigma_{j}^{-1}\left(s_{t}, p_{t}, x_{t}^{(2)} \right), \quad j \in \mathcal{J} \backslash \{0\}.
\end{equation}
We can rewrite \eqref{eq:inv_demand} in a more convenient form to get the following estimation equation 
\begin{equation} \label{eq:est_eq1}
	x_{jt}^{(1)} = \sigma_{j}^{-1}\left(s_{t}, p_{t}, x_{t}^{(2)} \right) - \xi_{jt}.
\end{equation}
Note that in \eqref{eq:est_eq1} the inverse demand is indexed by $j$, meaning that we have to estimate $J$ inverse demand functions, which is exactly why the approach of \cite{compiani2022market} suffers from the curse of dimensionality.

To circumvent this problem, we exploit the symmetry properties of the demand system. Under the linear in characteristics utility specification, \cite{gandhi_houde2019} show that the inverse demand function admits a symmetric representation. Specifically, define the state vector for product $j$ in market $t$ as
\begin{equation*}
	\omega_{jt} \equiv \left(\{s_{kt}, \Delta_{jkt}\}_{k \neq j}\right),
\end{equation*}
where $\Delta_{jkt} = \tilde{x}_{jt} - \tilde{x}_{kt}$ denotes characteristic differences with $\tilde{x}_{jt} \equiv \left(p_{jt}, x_{jt}^{(2)}\right)$, and index $k$ runs over all products except $j$, including the outside good $k = 0$. Note that $\omega_{jt}$ contains (i) the shares of all rival products (from which the own share $s_{jt}$ can be recovered), and (ii) the differences in characteristics with respect to all alternatives, including the outside good, so that $\Delta_{j0t} = (p_{jt}, x_{jt}^{(2)})$ captures own-product levels.

Building on Proposition 1 of \cite{gandhi_houde2019}, we establish the following result.
\begin{theorem}\label{thm:symmetric_inverse}
Let Assumptions 1 and 2 in \cite{berry2014identification} hold. Then under the linear in characteristics utility model, there exists a function $\gamma: \mathcal{W} \to \mathbb{R}$, symmetric in rival products, such that
\begin{equation} \label{eq:est_eq2}
	\log\left(\frac{s_{jt}}{s_{0t}}\right) = x_{jt}^{(1)} + \gamma(\omega_{jt}) + \xi_{jt},
\end{equation}
for all inside goods $j \in \mathcal{J} \backslash \{0\}$.
\end{theorem}

The key insight underlying Theorem~\ref{thm:symmetric_inverse} is that the standard outside good normalization ($\delta_{0t} = 0$) pins down any market-specific constant that arises in the symmetric representation. Since the state vector $\omega_{jt}$ contains the difference $\Delta_{j0t} = (p_{jt}, x_{jt}^{(2)})$ with respect to the outside good, it encodes sufficient information to recover the outside good's state $\omega_{0t}$, thereby absorbing the market effect into $\gamma$.

Equation~\eqref{eq:est_eq2} has several important features. First, the function $\gamma$ is \emph{not indexed by $j$}, the same function applies to all products, with only the inputs $\omega_{jt}$ varying across products. This allows us to pool observations across products and markets for estimation. Second, $\gamma$ is symmetric in rival products: permuting the labels of competing goods $k \neq j$ leaves $\gamma(\omega_{jt})$ unchanged. This symmetry dramatically reduces the effective dimensionality of the estimation problem. Third, under the multinomial logit model, we have $\gamma(\omega_{jt}) = \beta_p p_{jt} + x_{jt}^{(2)\prime}\beta_x$, so that $\gamma$ can be seen as a generalization of the linear functional form in the logit model.\footnote{This estimation equation is inherently connected to the Independence of Irrelevant Alternatives (IIA) property. If we specify $\gamma(\omega_{jt}) = \beta_p p_{jt} + x_{jt}^{(2)\prime}\beta_x + g(\{\Delta_{jkt}\}_{j \neq k})$, we will get the IIA-test form Section 3.3 of \cite{gandhi_houde2019}, who propose testing $g = 0$ as a diagnostic for IIA violations and instrument relevance.} Thus, \eqref{eq:est_eq2} is a semiparametric demand model in the spirit of \cite{lu2023semi}, though derived from different primitives. While \cite{lu2023semi} exploit the inclusive value structure to summarize cross-product information in a market-level aggregate, our approach retains product-level rival information in $\omega_{jt}$ and leverages symmetry for dimensionality reduction. Additionally, they treat the nonparametric component as having a random true value (induced by the dependence of the inclusive value on unobserved quality), whereas in our framework $\gamma$ is a fixed function evaluated at endogenous arguments, leading to a standard NPIV estimation problem.

Let $y_{jt} \equiv \log(s_{jt} / s_{0t}) - x_{jt}^{(1)}$, then we can rewrite equation \eqref{eq:est_eq2} in a more convenient form
\begin{equation} \label{eq:demand_main_eq}
	y_{jt} = \gamma(\omega_{jt}) + \xi_{jt}.
\end{equation}
Equation \eqref{eq:demand_main_eq} is the main structural equation where $\gamma$ is a complex nonparametric function characterizing the relationship between the inverse demand and product attributes and shares. Dimensionality of the input vector $\omega_{jt}$ depends on both the dimensionality of the characteristics space and the number of products in the market, thus, $\omega_{jt}$ is potentially high-dimensional. This will always be the case if we want to augment standard datasets with unstructured data such as product reviews, package images, etc. Since both the market shares $s_t$ and prices $p_t$ depend on the unobservable characteristics $\xi_t$, $\mathbb{E}[\xi_{jt}|\omega_{jt}] \neq 0$, and hence, $\omega_{jt}$ is endogenous.

In order to estimate $\gamma$, we need to construct a vector of instruments $z_{jt}$. \cite{BerryLevinsohnPakes1995} argue that the vector of product characteristics $x_{jt}$ is exogenous with respect to the structural error term $\xi_{jt}$, i.e. $\E[\xi_{jt}|x_{jt}] = 0$. This exogeneity condition can be used to construct demand side instruments $z_{jt}$. Instrument construction is a well-known problem in demand estimation, since it can lead to weak identification and distorted inference. We refer the reader to \cite{reynaert_verboven2014} and \cite{gandhi_houde2019} for a more detailed discussion.

To construct demand side instruments, we follow \cite{gandhi_houde2019} and use the transformed characteristics space $z_{jt} = (\{\Delta^{x}_{jkt}\}_{j \neq k})$, where $\Delta^{x}_{jkt} = x_{jt} - x_{kt}$, such that $\E[\xi_{jt}|z_{jt}] = 0$. Note that since $\omega_{jt}$ includes $z_{jt}$, it enforces strong correlation between endogenous inputs and instruments. If data permit, one can augment the instrument space with supply side instruments, such as cost shifters. Let $c_{jt}$ be a cost shifter for product $j$ in market $t$, then the instrument space becomes $z_{jt} = \left(\left\{\Delta^{x}_{jkt},\,\Delta^{c}_{jkt}\right\}_{j\neq k}\right)$, where $\Delta^{c}_{jkt} = c_{jt} - c_{kt}$.

\subsection{Price elasticity functional}\label{subsec:elasticity}

One of the main primitives in demand estimation is substitution patterns, which allow the researcher to          
  investigate the responsiveness of consumer choices to changes in market structure and, thus, understand the      
  nature of competition between firms. A conventional measure of substitution is price elasticities, reliable estimation and inference on which is therefore a
  first-order concern for applied researchers. Traditional approaches rely on parametric demand models such as logit,
  nested logit, BLP \citep{BerryLevinsohnPakes1995} that impose strong functional form restrictions on the
  mapping from structural parameters to elasticities. The semiparametric framework of Section \ref{subsec:semiparam_demand}
  removes these restrictions, but introduces a new challenge: the price elasticity is a \emph{nonlinear}
  functional of $\gamma$, and therefore requires the machinery developed in Section~\ref{sec:asy_prop_nonlin}. For
  expositional clarity, we focus on own-price elasticities below; the analysis extends to cross-price elasticities
  naturally.

The own-price elasticity for product $j$ in market $t$ is defined as the percentage change in market share in response to a percentage change in own price,
\begin{equation*}
	\varepsilon_{jj}(\gamma) = \frac{p_{jt}}{s_{jt}}\frac{\partial s_{jt}}{\partial p_{jt}}.
\end{equation*}
By the implicit function theorem applied to equation~\eqref{eq:demand_main_eq}, the share derivative matrix satisfies 
\begin{equation*}
	\nabla_{p_t}s_t = (L - \Gamma^s(\gamma))^{-1}\Gamma^p(\gamma),
\end{equation*}
where $L \in \R^{J \times J}$ is the log-share Jacobian with $L_{jk} = \partial\log(s_{jt}/s_{0t})/\partial s_{kt}$, $\Gamma^p \in \R^{J \times J}$ collects the price derivatives $\Gamma^p_{jk} = \partial\gamma(\omega_{jt})/\partial p_{kt}$, and $\Gamma^s \in \R^{J \times J}$ collects the share derivatives $\Gamma^s_{jk} = \partial\gamma(\omega_{jt})/\partial s_{kt}$. Setting $A(\gamma) \equiv L - \Gamma^s(\gamma)$, the own-price elasticity becomes
\begin{equation*}
	\varepsilon_{jj}(\gamma) = \frac{p_{jt}}{s_{jt}}\left[A^{-1}(\gamma)\Gamma^p(\gamma)\right]_{jj}.
\end{equation*}
The matrix $A$ captures the net effect of shares on the inverse demand system: $L$ accounts for the mechanical relationship between log-share ratios and share levels, while $\Gamma^s$ reflects how $\gamma$ responds to share movements. The interplay between $A^{-1}$ and $\Gamma^p$ encodes the full substitution pattern implied by the model, including indirect effects through equilibrium share responses. See Appendix~\ref{app:price_el} for a detailed derivation.

Under the logit model, $\gamma(\omega_{jt}) = \beta_p p_{jt} + x_{jt}^{(2)\prime}\beta_x$, so $\Gamma^s = 0$ and $\Gamma^p = \beta_p I$. In this case $A = L$ and the elasticity simplifies to $\varepsilon_{jj} = \beta_p\,p_{jt}(1 - s_{jt})$, recovering the familiar closed-form logit formula. More generally, departures from the logit model introduce nonzero share derivatives $\Gamma^s \neq 0$, which generate richer and more flexible substitution patterns through the matrix inverse.

The functional of interest is the average own-price elasticity,
\begin{equation}\label{eq:avg_el}
	\theta_{jj} = \E[\varepsilon_{jj}(\gamma)] = \E\left[\frac{p_{jt}}{s_{jt}}\left[A^{-1}(\gamma)\Gamma^p(\gamma)\right]_{jj}\right],
\end{equation}
where the expectation is taken over the distribution of markets. Since $\varepsilon_{jj}(\gamma)$ depends nonlinearly on $\gamma$ through the matrix inverse $A^{-1}(\gamma)$, the average own-price elasticity \eqref{eq:avg_el} is a nonlinear functional of $\gamma$. To construct a debiased estimator with valid inference, we appeal to the framework of Section~\ref{sec:asy_prop_nonlin}.

The key step is to verify that the Gateaux derivative of $\varepsilon_{jj}(\gamma)$ with respect to $\gamma$ is linear in the perturbation direction, so that the Riesz representer is well-defined and the PGMM estimator of Section~\ref{subsec:construction} can be applied. Proposition~\ref{prop:gateaux} in Appendix~\ref{app:price_el} establishes that for a perturbation $\zeta$\footnote{For notational convenience, we write $D_\gamma \varepsilon_{jj}[\zeta]$ for the Gateaux derivative $D(W, \gamma, \zeta)$ specialized to the elasticity functional.},
\begin{equation}\label{eq:gateaux_main}
	D_\gamma \varepsilon_{jj}[\zeta] = \frac{p_{jt}}{s_{jt}}\left[\left(A^{-1}Z^p\right)_{jj} + \left(A^{-1}Z^s A^{-1}\Gamma^p\right)_{jj}\right],
\end{equation}
where $Z^p_{jk} = \partial\zeta(\omega_{jt})/\partial p_{kt}$ and $Z^s_{jk} = \partial\zeta(\omega_{jt})/\partial s_{kt}$ are the price and share derivative matrices of the perturbation $\zeta$. Expression \eqref{eq:gateaux_main} consists of two terms: the first captures how a perturbation of $\gamma$ affects the price derivative matrix $\Gamma^p$, while the second captures how it affects the share derivative matrix $\Gamma^s$ and, consequently, the matrix inverse $A^{-1}$. Proposition~\ref{prop:linearity} shows that $D_\gamma\varepsilon_{jj}[\zeta]$ is linear in $\zeta$ since the matrices $A$ and $\Gamma^p$ depend on the current function $\gamma$ at which the derivative is evaluated, not on the perturbation direction $\zeta$, while $Z^p$ and $Z^s$ are linear in $\zeta$ by the linearity of differentiation. Therefore, Theorem~\ref{thm:asy_norm_nonlin} and Corollary~\ref{cor:asy_norm_nonlin_wm} are applicable, ensuring that the debiased estimator is asymptotically normal under the conditions stated therein.

A critical distinction from linear functionals is the role of \emph{double cross-fitting}. For a linear functional, the target vector $\hat{M}_\ell$ in the PGMM problem does not depend on an estimate of $\gamma$ and can be computed directly from the basis functions. For the elasticity functional, however, the Gateaux derivative \eqref{eq:gateaux_main} involves $\Gamma^p(\gamma)$, $\Gamma^s(\gamma)$, and $A^{-1}(\gamma)$, all of which depend on $\gamma$. Consequently, $\hat{M}_\ell$ must be constructed using an estimate $\hat\gamma$, and additional sample splitting is required to prevent the estimation error in $\hat\gamma$ from contaminating the Riesz representer estimate (see Appendix \ref{app:el_estimation} for details).

The asymptotic theory of Sections~\ref{sec:pgmm_prop}--\ref{sec:asy_prop_nonlin} applies with the market as the unit of observation which is a natural asymptotic regime for scanner data applications \citep[see e.g.,][]{nevo2001measuring}. Markets are assumed to be independent, so cross-fitting partitions markets rather than individual product-level observations. Let $W_{jt} \equiv (y_{jt},\,\omega_{jt},\,z_{jt})$ be a data tuple, and let $\cT_\ell$, $\ell = 1, \ldots, L$, be a partition of the observation index set $\{1, \ldots, T\}$ into $L$ distinct subsets of approximately equal size. Then for test fold $\ell$, the target vector is constructed as
\begin{equation}\label{eq:M_double_cf}
	\hat{M}_{\ell j} = \frac{1}{T - T_\ell}\sum_{\ell' \neq \ell}\sum_{t \in \cT_{\ell'}} D\!\left(W_{t},\,\hat\gamma_{\ell,\ell'},\,d_j\right),
\end{equation}
where $\hat\gamma_{\ell,\ell'}$ is trained on all markets outside both folds $\ell$ and $\ell'$, and $d_j$ denotes the $j$-th basis function. This double cross-fitting ensures that each summand in \eqref{eq:M_double_cf} evaluates the Gateaux derivative at markets from fold $\ell'$ using a $\hat\gamma_{\ell,\ell'}$ that was not fitted on those markets. At the same time, the entire vector $\hat{M}_\ell$ uses no data from the test fold $\ell$, preserving the independence required for valid inference on the Riesz representer $\hat\alpha_{jj,\ell}$. As discussed in Section~\ref{sec:asy_prop_nonlin}, the presence of $\hat\gamma_{\ell,\ell'}$ in $\hat{M}_\ell$ slows its convergence rate from $\sqrt{\log q / T}$ to $\kappa_T^\gamma$, the convergence rate of the MLIV estimator in the standard mean square norm, which in turn requires a slightly stronger rate condition on $\hat\gamma$ (Assumption~\ref{ass:conv_rate_gamma_nonlin}).

Given the estimates $\hat\gamma_{\ell}$ and $\hat\alpha_{jj,\ell}$ for each fold $\ell = 1,\dots,\,L$, the debiased estimator of the average own-price elasticity takes the form
\begin{equation}\label{eq:el_dml}
	\hat\theta_{jj} = \frac{1}{T}\sum_{\ell=1}^{L}\sum_{t \in \cT_\ell}\left\{\varepsilon_{jj}\!\left(W_t,\,\hat\gamma_\ell\right) + \hat\alpha_{jj,\ell}(z_{jt})\left[y_{jt} - \hat\gamma_\ell(\omega_{jt})\right]\right\},
\end{equation}
where $\varepsilon_{jj}\!\left(W_t,\,\hat\gamma_\ell\right)$ is the plug-in own-price elasticity for product $j$ in market $t$ and $\hat\alpha_{jj,\ell}(z_{jt})$ is the PGMM estimate of the Riesz representer of the Gateaux derivative of $\varepsilon_{jj}$. The variance can be consistently estimated by
\begin{equation*}
	\hat{V}_{jj} = \frac{1}{T}\sum_{\ell=1}^{L}\sum_{t \in \cT_\ell}\hat\psi_\ell^2(W_t), \quad \hat\psi_\ell(W_t) = \varepsilon_{jj}\!\left(W_t,\,\hat\gamma_\ell\right) - \hat\theta_{jj} + \hat\alpha_{jj,\ell}(z_{jt})[y_{jt} - \hat\gamma_\ell(\omega_{jt})].
\end{equation*}
The full estimation procedure is described in Appendix~\ref{app:el_estimation}.

\subsection{Simulated data experiments} \label{subsec:demand_mc}

We now evaluate the finite sample performance of the debiased estimator for the average own-price elasticity functional $\theta_{jj} = \E[\varepsilon_{jj}(\gamma)]$. This is a nonlinear functional of $\gamma$, requiring the double cross-fitting procedure and the stronger convergence conditions described in Section~\ref{sec:asy_prop_nonlin}. Unlike the average derivative studied in Section~\ref{sec:toy_mc}, the elasticity depends on $\gamma$ through the matrix inverse $(L - \Gamma^s)^{-1}$, making both the functional evaluation and the Riesz representer estimation more challenging.

We simulate data from a simple logit model. The mean valuation in the logit model is given by $\delta_{jt} = \beta_{p} p_{jt} + x_{jt}'\beta_{x} + \xi_{jt}$. Product shares can be calculated using the following formulae, for $j=1,\dots,\,J$ and $t=1,\dots,\,T$,
\begin{equation*}
	s_{jt} = \frac{\exp(\delta_{jt})}{1 + \sum_{j=1}^{J}\exp(\delta_{jt})} \quad \text{and} \quad
	s_{0t} = \frac{1}{1 + \sum_{j=1}^{J}\exp(\delta_{jt})}.
\end{equation*}

We set the total number of product characteristics besides the price to be equal to 4, i.e. $x^{(1)}_{jt}$ is a scalar and $x^{(2)}_{jt}$ is a three-dimensional vector. We draw all the observed product characteristics from $\cU(0,1)$, while the unobserved characteristics, $\xi_{jt}$, are distributed as $\cN(1,\,0.15^2)$ for all $j$ and $t$. The price is
\begin{equation*}
	p_{jt} = \frac{1}{2}\left|1 + x^{(1)}_{jt} + \sum_{k=1}^{3}x_{k,jt}^{(2)} + \xi_{jt} + c_{jt} + e_{jt}\right|,
\end{equation*}
where $c_{jt} \sim \cU(0,1)$ is a cost shifter, and $e_{jt} \sim \cU(0, 0.1)$ is idiosyncratic noise. The price coefficient is $\beta_{p} = -2$ and the coefficients on product characteristics are $\beta_{x} = (1,\,-0.5,\,0.5,\,1)'$.

To estimate $\gamma$, we choose KIV over the Double Lasso estimator used in Section~\ref{sec:toy_mc} due to its substantially lower computational cost in the nonlinear functional setting which requires double cross-fitting. We construct dictionaries $d(\omega_{jt})$ and $b(z_{jt})$ using quadratic polynomials with and without interactions, respectively. Under the specified DGP, $\omega_{jt} = \left(\{s_{kt}, \Delta_{jkt}\}_{j \neq k}\right)$ and $z_{jt} = \left(\left\{\Delta^{x}_{jkt},\,\Delta^{c}_{jkt}\right\}_{j\neq k}\right)$, and hence, $dim(\omega_{jt}) = dim(z_{jt})$ and $p>q$. We run 500 replications for $J \in \{2, 5\}$\footnote{Note that the estimation problem quickly becomes high-dimensional. Under $J=2$, $dim(\omega_{jt}) = 15$, which grows up to $30$ when $J=5$.} products and $T \in \{100, 200, 400, 800\}$ markets. We use five-fold cross-fitting, $L = 5$. Without loss of generality, we focus on the elasticity of the first product ($j = 1$).

Under logit, the true own-price elasticity for product $j$ in market $t$ has the closed-form expression $\varepsilon_{jj,t} = \beta_p\, p_{jt}(1 - s_{jt})$. The population parameter of interest is $\theta_0 = \E[\varepsilon_{jj,t}]$, which we approximate using a large pre-simulation draw of $T = 100{,}000$ markets. The resulting values are approximately $-4.22$ for $J = 2$ and $-4.28$ for $J = 5$.

Table~\ref{tab:mc_elasticity} reports absolute bias, median estimated standard error, and empirical coverage of nominal 95\% confidence intervals for the plug-in and debiased estimators. Unlike the average derivative application in Section \ref{sec:toy_mc}, the analytical form of the Riesz representer for the elasticity functional is intractable, hence, we omit the analytical debiasing approach from the comparison.

\begin{table}[htbp]
\centering
\begin{threeparttable}
\caption{Monte Carlo Results: Average Own-Price Elasticity}
\label{tab:mc_elasticity}
\begin{tabular}{cc| cc| cc| cc| cc}
\toprule
 & & \multicolumn{2}{c}{$|\text{Bias}|$} & \multicolumn{2}{c}{Median SE} & \multicolumn{2}{c}{Coverage} \\
\cmidrule(lr){3-4} \cmidrule(lr){5-6} \cmidrule(lr){7-8} 
$J$ & $T$ & PI & ADML & PI & ADML & PI & ADML \\
\midrule
2 & 100 & 0.023 & 0.004 & 0.068 & 0.396 & 33.1\% & 94.8\% \\
  & 200 & 0.009 & 0.020 & 0.045 & 0.203 & 36.6\% & 92.2\% \\
  & 400 & 0.030 & 0.017 & 0.030 & 0.129 & 31.5\% & 92.2\% \\
  & 800 & 0.023 & 0.023 & 0.020 & 0.091 & 37.6\% & 91.2\% \\
\midrule
5 & 100 & 0.027 & 0.039 & 0.073 & 0.555 & 36.3\% & 91.4\% \\
  & 200 & 0.009 & 0.001 & 0.043 & 0.199 & 43.2\% & 94.0\% \\
  & 400 & 0.005 & 0.003 & 0.032 & 0.116 & 47.5\% & 95.0\% \\
  & 800 & 0.005 & 0.005 & 0.023 & 0.077 & 54.8\% & 94.2\% \\
\bottomrule
\end{tabular}\begin{tablenotes}
\footnotesize
\item PI = plug-in estimator; ADML = automatic debiased estimator with PGMM-estimated Riesz representer. The inverse demand function $\gamma$ is estimated using KIV. Results averaged over 500 Monte Carlo replications. 
\end{tablenotes}
\end{threeparttable}
\end{table}

The plug-in estimator severely undercovers across all specifications. While its absolute bias is small, its estimated standard errors dramatically understate the true sampling variability. In contrast, the debiased estimator achieves near-nominal coverage throughout, ranging from 91.2\% to 95.0\%. These results align with and extend the evidence from Section~\ref{sec:toy_mc} on linear functionals. The key qualitative finding that plug-in inference fails due to standard error underestimation rather than excessive bias carries over from average derivatives to own-price elasticities. 

Nevertheless, the nonlinear setting introduces additional challenges causing the coverage for the debiased estimator to be  slightly lower than in the linear case, particularly at small $T$ with $J = 2$. This behavior is consistent with the additional estimation noise introduced by double cross-fitting and the PGMM estimation of the $\gamma$-dependent Riesz representer. Despite these additional difficulties, the automatic debiasing procedure reliably delivers valid inference for the nonlinear elasticity functional across all sample sizes considered.


\subsection{Real data example}\label{subsec:iri}

We apply our estimator to retail scanner data from the IRI Academic Database \citep{bronnenberg2008IRI}. It records unit sales by UPC code, store, and week for a sample of supermarkets over 2001--2012, together with product characteristics. We focus on one year span of 2003 and top ten most sold products which account for approximately $69\%$ of all sales. To exploit variation in product attributes, we do not aggregate to the brand level; instead, products are defined by the combination of a brand name and a set of observable characteristics, e.g. Coke Classic or Caffeine-free Diet Pepsi.

Carbonated beverages are sold in different packages and package sizes. We restrict our attention to cans and define a product unit as a 12 oz can. Hence, we construct market shares based on the total amount of cans sold. Prices are defined as the ratio of total revenue to total number of units sold. We aggregate the data to geographic region-month level resulting in an unbalanced panel with 5,993 observations at the product-region-month level.

Data on product characteristics includes beverage flavor, sugar, caffeine and calorie levels. All characteristics are represented by categorical variables, thus, for computational reasons we aggregate product attributes in larger groups (see Appendix \ref{app:data} for more details). After aggregation and dropping collinear characteristics, we are left with five product attributes we use for estimation. We have \emph{Caffeine} and \emph{Sugar} dummy variables indicating whether a product contains caffeine and sugar, respectively. The remaining three variables represent different flavor categories: \emph{Cola}, \emph{Lemonade}, and \emph{Pepper}, with the base category representing other flavors.

We begin with a logit model as a parametric benchmark. In the absence of external cost-shifter data, we use BLP-style instruments (sums of rivals' attributes) to identify the price coefficient. Given the binary nature of product characteristics, most rival sums are nearly collinear with own attributes, leaving only two instruments with independent variation: sums of rivals' sugar and caffeine levels. One popular approach to achieve stronger price coefficient identification without external data is to use Hausman instruments \citep{hausman1996}, however, in our case they capture endogenous demand shifts rather than track down exogenous cost shocks\footnote{See Appendix \ref{app:hausman} for more details.}, which is a common critique of this approach \citep[see e.g.,][]{nevo2001measuring}. We therefore use BLP instruments throughout.

Before estimating the semiparametric model, we present evidence that the data reject the logit specification. Following \citet{gandhi_houde2019}, we conduct an Independence of Irrelevant Alternatives (IIA) test using local differentiation instruments.\footnote{See Appendix~\ref{app:iia_test} for details on the test construction.} The test strongly rejects IIA ($F = 68.5$, $p < 0.0001$), indicating that the data contain richer substitution patterns than logit can rationalize. This motivates the use of a semiparametric approach that does not impose parametric restrictions on demand.

\begin{table}[htbp]
\centering
\begin{threeparttable}
\caption{Mean Own-Price Elasticity Estimates}
\label{tab:elasticity_comparison}
\begin{tabular}{lccc}
\toprule
Product & Logit & PI & ADML \\
\midrule
Coke Classic & $-2.739$ $(0.016)$ & $-3.456$ $(0.021)$ & $-3.436$ $(0.022)$ \\
Pepsi Classic & $-2.762$ $(0.016)$ & $-3.468$ $(0.024)$ & $-3.472$ $(0.024)$ \\
Diet Coke & $-2.973$ $(0.014)$ & $-3.765$ $(0.018)$ & $-3.448$ $(0.089)$ \\
Diet Pepsi & $-3.019$ $(0.015)$ & $-3.838$ $(0.022)$ & $-3.965$ $(0.025)$ \\
Caffeine-free Diet Coke & $-3.237$ $(0.017)$ & $-3.889$ $(0.018)$ & $-3.767$ $(0.021)$ \\
Caffeine-free Diet Pepsi & $-3.216$ $(0.016)$ & $-3.937$ $(0.020)$ & $-4.108$ $(0.021)$ \\
Dr. Pepper & $-3.300$ $(0.021)$ & $-3.543$ $(0.019)$ & $-3.543$ $(0.020)$ \\
Mountain Dew Classic & $-3.276$ $(0.019)$ & $-3.927$ $(0.032)$ & $-3.927$ $(0.032)$ \\
Sprite & $-3.171$ $(0.016)$ & $-3.702$ $(0.017)$ & $-3.600$ $(0.022)$ \\
Mountain Dew Other & $-3.250$ $(0.020)$ & $-3.992$ $(0.023)$ & $-3.992$ $(0.023)$\\ 
\midrule
Mean & $-3.094$ & $-3.752$ & $-3.726$ \\
\bottomrule
\end{tabular}
\begin{tablenotes}[flushleft]
\footnotesize
\item Logit = Logit estimated with BLP instruments; PI = plug-in estimator; ADML = automatic debiased estimator with PGMM-estimated Riesz representer. Standard errors in parentheses. The inverse demand function $\gamma$ is estimated using KIV.
\end{tablenotes}
\end{threeparttable}
\end{table}

Similarly to Section \ref{subsec:demand_mc}, we estimate $\gamma$ using KIV and construct debiased estimates using five-fold cross-fitting. Since the dimensionality of $\omega_{jt}$ and $z_{jt}$ grows with the number of products, we exploit the symmetry of the inverse demand function to construct symmetric basis functions\footnote{Using 
    symmetric dictionaries implicitly restricts the Riesz representer $\alpha_0(z_{jt})$ to be a symmetric function of rival characteristics. This restriction is justified by the structure of the moment conditions that define $\alpha_0$. First, the own-price elasticity $\varepsilon_{jj}$ is invariant to permutations of rivals, so the Gateaux derivative $D_\gamma\varepsilon_{jj}[d_k]$ inherits this symmetry whenever $\gamma$ and the basis function $d_k$ are symmetric (Proposition~\ref{prop:gateaux}). Then, since the instruments $z_{jt}$ are also constructed symmetrically, the defining moment condition $\E[D_\gamma\varepsilon_{jj}[d_k]] = \E[\alpha_0(z_{jt})\,d_k(\omega_{jt})]$ is symmetric, and hence $\alpha_0$ must be as well.} for dictionaries $d(\omega_{jt})$ and $b(z_{jt})$ (see Appendix \ref{app:em_bfs} for details). It allows us to reduce the dimensionality of the PGMM problem leading to substantial computational gains. 

Table~\ref{tab:elasticity_comparison} reports own-price elasticity estimates for each of the ten products under three approaches: the logit benchmark, the plug-in KIV estimator (PI), and the debiased estimator (ADML). 

First, the semiparametric estimators yield substantially larger elasticities in absolute value than the logit. The mean own-price elasticity increases from $-3.09$ under the logit to $-3.75$ (plug-in) and $-3.73$ (ADML). This is consistent with the IIA rejection: by constraining all cross-price elasticities to be proportional to market shares, the logit restricts substitution patterns in a way that attenuates own-price elasticities. While the magnitudes differ, the qualitative ranking of products is stable across methods. More popular products like Coke Classic and Pepsi Classic have smaller elasticities in absolute value, while niche products such as Caffeine-Free Diet Pepsi and Mountain Dew have larger elasticities. This implies that consumers of mainstream products are less sensitive to their price changes compared to more niche products, which is quite intuitive.

Second, the debiasing corrections vary substantially across products, both in magnitude and direction. For several products (Dr.\ Pepper, Mountain Dew Classic, and Mountain Dew Other) the two estimates are virtually identical, suggesting that the plug-in bias is negligible for these products. For all Coke products and Sprite, the debiased estimates are less elastic (positive debiasing term) compared to the plug-in ones, while all Pepsi products exhibit more elastic own-price elasticities compared to the plug-in ones (negative debiasing term). 

Third, both plug-in and debiased estimates have similar standard errors, with the exception of Diet Coke that also has the largest correction. We observe that debiasing plays a big role in this application resulting in half the products (Diet Coke, Caffeine-free Diet Coke, Sprite, Diet Pepsi, and Caffeine-free Diet Pepsi) having debiased estimates differ from plug-in estimates by much more than the associated standard errors. These results are also consistent with the CNS demand application.


\section{Conclusion} \label{sec:conclusion}

In this paper, we have given an automatic method of debiasing functionals of machine learners under endogeneity. We have shown how to use a PGMM minimum distance estimator to perform debiasing using only the form of the object of interest, without knowing the form of the bias correction term. We allow for a wide range of MLIV estimators that satisfy certain convergence rate conditions. We have shown root-$n$ consistency and asymptotic normality and given a consistent asymptotic variance estimator for both linear and nonlinear functionals. For linear functionals we require MLIV estimators to converge fast enough in the projected mean square norm, while for nonlinear functionals we require fast enough convergence in the standard mean square norm, which is a more stringent requirement due to ill-posedness. Relaxing the convergence rate condition for nonlinear functionals as well as extending the approach to irregular functionals are promising directions for future research.

Finally, we have applied our debiasing procedure to estimate the mean own-price elasticity functional in the nonparametric demand for differentiated products framework. We have obtained evidence from both simulated and real scanner data that plug-in estimates are biased and have smaller variance compared to the debiased estimates, which reflects the bias-variance trade-off occurring due to regularization. Using scanner data, we have demonstrated that debiasing corrections vary substantially across products, both in magnitude and direction. In particular,  half the products exhibit profound debiasing effects highlighting the importance of debiasing in our application.


\bibliography{bib_admliv}




\newpage
\appendix

\setcounter{page}{1}
\begin{center}
{\large \textsc{Supplementary Material for}}\\[1em]
{\Large ``Penalized GMM Framework for Inference on Functionals of
Nonparametric Instrumental Variable Estimators''}\\[1em]
{\normalsize Edvard Bakhitov}\\[0.5em]
{\normalsize \today}
\end{center}
\vspace{1em}
\addtocontents{toc}{\protect\setcounter{tocdepth}{2}}
\tableofcontents
\newpage

\numberwithin{equation}{section}
\numberwithin{theorem}{section}
\numberwithin{lemma}{section}
\numberwithin{assumption}{section}
\numberwithin{proposition}{section}
\numberwithin{corollary}{section}
\numberwithin{algorithm}{section}
\numberwithin{table}{section}


\section{Analytical solution to the unpenalized GMM problem}

In this Section, we provide additional intuition behind the PGMM estimator of the RR. To do so, we focus on the standard low-dimensional GMM problem without adding the penalty term, i.e. $\hat\rho$ is a solution to
\begin{equation}\label{eq:gmm_rr}
\min_{\rho \in \R^{p}}\; (\hat{M} - \hat{G}\rho)'\hat{\Omega}_{q}(\hat{M} - \hat{G}\rho).	
\end{equation}
Given the form of the debiased moment function \eqref{eq:if_npiv} and the linear approximation for the RR, the orthogonal moment condition \eqref{eq:gen_mom_rr_est} will always be linear in $\rho$, meaning that the GMM criterion in \eqref{eq:gmm_rr} is globally convex and has a unique global minimizer.

For the ease of exposition, we drop the cross-fitting notation and assume that we are interested in a linear functional $\theta = \mathbb{E}[m(W,\,\gamma)]$. Then the moment condition takes the form
\begin{equation*}
\hat{\psi}_{j}(W, \rho) = \frac{1}{n}\sum_{i=1}^{n}\left\{m(W_{i},\,d_{j}) - d_{j}(X_{i})b(Z_{i})'\rho\right\}, \quad j=1,\dots,\,q.
\end{equation*}
Let $m(W,\,d) = (m(W,\,d_{1}),\,\dots,\,m(W,\,d_{q}))$. Taking the first-order condition of the GMM criterion gives
\begin{equation} \label{eq:gmm_foc}
	\frac{\partial \hat{\psi}(\hat\rho)}{\partial \rho'}\hat{\Omega}_{q}\left\{\frac{1}{n}\sum_{i=1}^{n} m(W_{i},\,d) - \frac{1}{n}\sum_{i=1}^{n} d(X_{i})b(Z_{i})'\hat\rho \right\}= 0.
\end{equation}
We can rewrite \eqref{eq:gmm_foc} in matrix form as
\begin{equation*}
	-\hat{G}'\hat{\Omega}_{q}\hat{M} + \hat{G}'\hat{\Omega}_{q}\hat{G}\hat\rho = 0,
\end{equation*}
which immediately gives a closed-form solution for $\hat\rho$,
\begin{equation} \label{eq:gmm_rr_sol}
	\hat{\rho} = (\hat{G}'\hat{\Omega}_{q}\hat{G})^{-1}\hat{G}'\hat{\Omega}_{q}\hat{M}.
\end{equation}

Note that the form of the GMM solution in \eqref{eq:gmm_rr_sol} resembles the GMM solution to the classical linear IV problem, but with endogenous regressors and instruments being switched. \cite{ichimura2022influence} point out that $\alpha(Z)$ is the solution of a ``reverse'' structural equation involving an expectation conditional on the endogenous variables $X$ rather than the instruments $Z$. If we set $\hat{\Omega} = \left(\frac{1}{n}\sum_{i=1}^{n} d(X_{i})d(X_{i})'\right)^{-1}$, we will get the exact solution to the ``reverse'' NPIV problem.


\section{Penalized GMM Implementation Details}\label{app:algo_details}

\subsection{Standard Coordinate-Wise Descent}

Recall, in matrix form the PGMM problem is given by
\begin{equation} \label{eq:pgmm_objective}
	\min_{\rho\in\R^{p}}\; (\hat{M} - \hat{G}\rho)'\hat{\Omega}_{q}(\hat{M} - \hat{G}\rho) + 2\lambda_{n}|\rho|_{1}.
\end{equation}
Note that the objective in \eqref{eq:pgmm_objective} is a generalized version of the Lasso objective. Thus, we can generalize the coordinate descent approach for Lasso to the PGMM objective that we use in this paper. We follow CNS and use a coordinate-wise descent algorithm with the soft-thresholding update. 

We denote the $j^{\text{th}}$ element of a generic vector $v$ by $v_{j}$ and let $e_{j}$ be a $p \times 1$ unit vector with $1$ in the $j^{\text{th}}$ coordinate and zeros elsewhere. Let $H = \hat{G}'\hat{\Omega}_{q}\hat{G}$ and $c = \hat{G}'\hat{\Omega}_{q}\hat{M}$ denote the Hessian and linear term of the quadratic component of \eqref{eq:pgmm_objective}, respectively. The soft-thresholding operator is defined as
\begin{equation*}
	\mathcal{S}_{\tau}(z) = \text{sign}(z) \cdot \max(|z| - \tau, 0) = 
	\begin{cases}
		z + \tau & \text{if } z < -\tau \\
		0 & \text{if } z \in [-\tau, \tau] \\
		z - \tau & \text{if } z > \tau
	\end{cases}
\end{equation*}

The basic coordinate-wise descent algorithm cycles through all $p$ coordinates in each iteration. The justification for Algorithm \ref{algo:pgmm_coordinate} is similar to the one of CNS. It follows from the fact that the GMM objective \eqref{eq:pgmm_objective} is of the form of eq. 21 of \cite{friedman2007}, hence, the coordinate descent converges to the minimizer of the objective \citep{tseng2001convergence}.

\begin{algorithm}[!h] 
	\caption{Coordinate-wise descent algorithm for PGMM} 
	\begin{algorithmic}[1]
		\Require Data matrices $\hat{G} \in \R^{q \times p}$, $\hat{M} \in \R^{q}$, weight matrix $\hat{\Omega}_{q} \in \R^{q \times q}$, penalty $\lambda_{n} > 0$, tolerance $\epsilon > 0$
		\Ensure $\hat{\rho}$ solving \eqref{eq:pgmm_objective}
		\State \textbf{Initialize:} $\rho \gets 0_{p}$
		\State \textbf{Precompute:} $H \gets \hat{G}'\hat{\Omega}_{q}\hat{G}$, $c \gets \hat{G}'\hat{\Omega}_{q}\hat{M}$
		\Repeat
			\For {$j = 1, \ldots, p$}
				\State Compute partial residual: $A_{j} \gets c_{j} - \sum_{k \neq j} H_{jk}\rho_{k}$
				\State Extract diagonal element: $B_{j} \gets H_{jj}$
				\State Store previous value: $\rho_{j}^{\text{old}} \gets \rho_{j}$
				\State Update via soft-thresholding: $\rho_{j} \gets \mathcal{S}_{\lambda_{n}/B_{j}}(A_{j}/B_{j})$
			\EndFor
		\Until{$\max_{j} |\rho_{j} - \rho_{j}^{\text{old}}| < \epsilon$}
		\State \Return $\rho$
	\end{algorithmic} 
	\label{algo:pgmm_coordinate}
\end{algorithm}

\subsection{Active Set Strategy}\label{subsec:active_set}

While Algorithm \ref{algo:pgmm_coordinate} enjoys strong convergence guarantees, its computational burden becomes prohibitive when the dimension $p$ is large, as each iteration requires $O(pq)$ operations. In high-dimensional sparse estimation problems, however, the solution $\hat{\rho}$ typically exhibits substantial sparsity, with the number of nonzero elements $|\mathcal{A}^{*}|$ satisfying $|\mathcal{A}^{*}| \ll p$. The active set strategy, introduced by \cite{friedman2007} and refined in the \texttt{glmnet} implementation of \cite{friedman2010}, exploits this sparsity structure to achieve substantial computational gains without sacrificing statistical accuracy.

The theoretical foundation of the active set strategy rests on the Karush-Kuhn-Tucker (KKT) optimality conditions for problem \eqref{eq:pgmm_objective}. A vector $\rho^{*}$ is optimal if and only if for all $j = 1, \ldots, p$:
\begin{align}
	r_{j} &= \lambda_{n} \text{sign}(\rho_{j}^{*}) & \text{if } \rho_{j}^{*} \neq 0, \nonumber \\
	|r_{j}| &\leq \lambda_{n} & \text{if } \rho_{j}^{*} = 0, \label{eq:kkt_zero}
\end{align}
where $r = c - H\rho^{*}$ denotes the gradient of the smooth component of the objective evaluated at $\rho^{*}$.

The key insight underlying the active set strategy is that once a coordinate enters the zero region (i.e., $|r_{j}| \leq \lambda_{n}$), it tends to remain inactive for many subsequent iterations. This observation motivates a two-stage iterative scheme: an inner loop that optimizes exclusively over the current active set $\mathcal{A}$ of nonzero coordinates, and an outer loop that verifies the KKT conditions \eqref{eq:kkt_zero} for all inactive coordinates and augments $\mathcal{A}$ with any violators. Algorithm \ref{algo:pgmm_active_set} formalizes this procedure.

\begin{algorithm}[!h] 
	\caption{PGMM with active set strategy} 
	\begin{algorithmic}[1]
		\Require Data matrices $\hat{G} \in \R^{q \times p}$, $\hat{M} \in \R^{q}$, weight matrix $\hat{\Omega}_{q} \in \R^{q \times q}$, penalty $\lambda_{n} > 0$, tolerance $\epsilon > 0$
		\Ensure $\hat{\rho}$ solving \eqref{eq:pgmm_objective}
		\State \textbf{Initialize:} $\rho \gets 0_{p}$, $\mathcal{A} \gets \emptyset$
		\State \textbf{Precompute:} $H \gets \hat{G}'\hat{\Omega}_{q}\hat{G}$, $c \gets \hat{G}'\hat{\Omega}_{q}\hat{M}$
		\State \textbf{Full pass:} Execute one cycle of Lines 4--9 of Algorithm \ref{algo:pgmm_coordinate} over all $j = 1, \ldots, p$
		\State \textbf{Initialize active set:} $\mathcal{A} \gets \{j : \rho_{j} \neq 0\}$
		\Repeat
			\Repeat \Comment{Inner loop over active set}
				\For {$j \in \mathcal{A}$}
					\State Compute partial residual: $A_{j} \gets c_{j} - \sum_{k \neq j} H_{jk}\rho_{k}$
					\State Extract diagonal element: $B_{j} \gets H_{jj}$
					\State Store previous value: $\rho_{j}^{\text{old}} \gets \rho_{j}$
					\State Update via soft-thresholding: $\rho_{j} \gets \mathcal{S}_{\lambda_{n}/B_{j}}(A_{j}/B_{j})$
				\EndFor
			\Until{$\max_{j \in \mathcal{A}} |\rho_{j} - \rho_{j}^{\text{old}}| < \epsilon$}
			\State Compute gradient: $r \gets c - H\rho$
			\State Identify KKT violators: $\mathcal{V} \gets \{j \notin \mathcal{A} : |r_{j}| > \lambda_{n}\}$
			\If {$\mathcal{V} \neq \emptyset$}
				\State Augment active set: $\mathcal{A} \gets \mathcal{A} \cup \mathcal{V}$
			\EndIf
		\Until{$\mathcal{V} = \emptyset$}
		\State \Return $\rho$
	\end{algorithmic} 
	\label{algo:pgmm_active_set}
\end{algorithm}

The convergence of Algorithm \ref{algo:pgmm_active_set} follows from the finite termination property of the active set identification and the convergence of coordinate descent on the restricted subproblem. Since the active set can only grow (never shrink) during the outer iterations and is bounded above by $\{1, \ldots, p\}$, the outer loop must terminate in at most $p$ iterations. In practice, termination occurs much sooner, typically within $K \leq 10$ outer iterations for problems with moderate sparsity.


\subsection{Initialization parameters}

For the first-stage PGMM, two approaches are available. The simplest option initializes at $\rho = 0_p$, which is adequate when the regularization parameter $\lambda_n$ is sufficiently large. Alternatively, following CNS, one may obtain an initial estimate by solving a lower-dimensional unpenalized problem \eqref{eq:gmm_rr} using a reduced dictionary $b^{\text{low}} \subset b$ of dimension $p^{\text{low}} \ll p$. The closed-form solution
\begin{equation*}
	\tilde{\rho}^{\text{low}} = \left(\hat{G}^{\text{low}'}\hat{G}^{\text{low}}\right)^{-1}\hat{G}^{\text{low}'}\hat{M}
\end{equation*}
then provides an initial estimate that is extended to the full parameter space by padding with zeros. This latter approach yields a warm start that can reduce the number of iterations required for convergence. 

For the second-stage PGMM, the natural initialization is the first-stage estimate $\tilde{\rho}$, which typically lies close to the final solution and thus substantially accelerates convergence.


\subsection{Computational Complexity}

We now provide a detailed analysis of the computational complexity of Algorithm \ref{algo:pgmm_active_set} relative to the standard coordinate descent procedure in Algorithm \ref{algo:pgmm_coordinate}.

First, consider the standard algorithm. Each complete cycle through all $p$ coordinates requires $O(p)$ soft-thresholding operations, each involving a partial residual computation of cost $O(p)$ when implemented via the formula $A_j = c_j - \sum_{k \neq j} H_{jk} \rho_k$. The matrices $H = G'\Omega G$ and $c = G'\Omega M$ are computed once at initialization at a cost of $O(pq^2 + p^2 q)$ and $O(q^2 + pq)$, respectively. Denoting by $T_{\text{full}}$ the number of complete cycles required for convergence, the total computational cost is
\begin{equation}\label{eq:complexity_standard}
	\mathcal{C}_{\text{CD}} = O(pq^{2} + p^{2}q + T_{\text{full}} \cdot p^{2}).
\end{equation}

For the active set algorithm, let $K$ denote the number of outer iterations and let $T_{\mathcal{A}}^{(k)}$ denote the number of inner iterations at outer iteration $k$, with $|\mathcal{A}^{(k)}|$ the cardinality of the active set. The inner loop at iteration $k$ requires $O(T_{\mathcal{A}}^{(k)} \cdot |\mathcal{A}^{(k)}|^{2})$ operations when residuals are maintained incrementally. The KKT verification step requires computing the full gradient $r = c - H\rho$, which can be accomplished in $O(p \cdot |\mathcal{A}^{(k)}|)$ operations by exploiting the sparsity of $\rho$. Summing over all outer iterations yields
\begin{equation}\label{eq:complexity_active}
	\mathcal{C}_{\text{AS}} = O\left(pq^{2} + p^{2}q + \sum_{k=1}^{K} \left(T_{\mathcal{A}}^{(k)} \cdot |\mathcal{A}^{(k)}|^{2} + p \cdot |\mathcal{A}^{(k)}|\right)\right).
\end{equation}

Under the assumption that the final active set satisfies $|\mathcal{A}^{*}| \ll p$ and that the active set grows monotonically to $\mathcal{A}^{*}$ over the $K$ outer iterations, the dominant term in \eqref{eq:complexity_active} is $O(K \cdot p \cdot |\mathcal{A}^{*}|)$ from the KKT checks. Comparing with \eqref{eq:complexity_standard}, the speedup factor is approximately
\begin{equation*}
	\frac{\mathcal{C}_{\text{CD}}}{\mathcal{C}_{\text{AS}}} \approx \frac{T_{\text{full}} \cdot p}{K \cdot |\mathcal{A}^{*}|}.
\end{equation*}
For typical high-dimensional sparse problems where $|\mathcal{A}^{*}|/p \in [0.01, 0.1]$ and $K \leq 10$, this ratio yields speedups on the order of $10\times$ or greater, consistent with our empirical findings.

Note that in the worst case, when the solution is dense with $|\mathcal{A}^{*}| = O(p)$, the active set strategy has no computational advantage. However, such scenarios are atypical in the high-dimensional sparse inference settings for which PGMM is designed. Thus, we always recommend to use Algorithm \ref{algo:pgmm_active_set} in practice.


\subsection{Adaptive PGMM}

The statistical performance of the PGMM estimator can be further enhanced by incorporating adaptive penalty loadings in the spirit of \cite{zou2006adaptive}. This modification transforms the optimization problem \eqref{eq:pgmm_objective} into the adaptive PGMM (A-PGMM) problem
\begin{equation} \label{eq:apgmm_objective}
	\min_{\rho \in \R^{p}}\; (\hat{M} - \hat{G}\rho)'\hat{\Omega}_{q}(\hat{M} - \hat{G}\rho) + 2\lambda_{n}\sum_{j=1}^{p}\hat{w}_{j}|\rho_{j}|,
\end{equation}
where $\hat{w}=(\hat{w}_{1},\dots,\,\hat{w}_{p})$ is a vector of data-dependent weights with $\hat{w}_{j} = 1/|\tilde{\rho}_{j}|$, and $\tilde{\rho}$ is a preliminary consistent estimate, typically obtained from the first stage PGMM fit.

Algorithms \ref{algo:pgmm_coordinate} and \ref{algo:pgmm_active_set} extend naturally to the adaptive setting. The required modifications are twofold: (i) in the soft-thresholding update, the penalty parameter $\lambda_{n}$ is replaced by the coordinate-specific penalty $\hat{w}_{j}\lambda_{n}$; (ii) in the KKT verification step of Algorithm \ref{algo:pgmm_active_set}, the threshold for coordinate $j$ becomes $\hat{w}_{j}\lambda_{n}$. 


\subsection{Penalty Parameter Selection}\label{subsec:penalty_tuning}

The regularization parameter $\lambda_n$ in the PGMM objective \eqref{eq:pgmm_objective} governs the bias-variance tradeoff and must be chosen appropriately. The theoretically optimal penalty takes the form\footnote{Note that nonlinear functionals require a larger penalty (see Section \ref{sec:asy_prop_nonlin}). Typically, a penalty proportional to $n^{-1/4}$ should suffice in practice.}
\begin{equation}\label{eq:lambda_theory}
	\lambda^{*} = c_1 \sqrt{\frac{\log q}{n}},
\end{equation}
where $c_1 > 0$ is a tuning parameter.

Following CNS, we distinguish between the penalty applied to the intercept term and that applied to the remaining basis functions. This distinction is motivated by the observation that the intercept, corresponding to the first element of the dictionary $b(Z)$, should not be heavily penalized, as it captures the average level of the Riesz representer. The penalized objective becomes
\begin{equation}\label{eq:pgmm_split_penalty}
	\min_{\rho \in \mathbb{R}^p} \; (\hat{M} - \hat{G}\rho)'\hat{\Omega}_q(\hat{M} - \hat{G}\rho) + 2c_2\lambda^{*}|\rho_1| +  2\lambda^{*}\sum_{j=2}^{p} |\rho_j|,
\end{equation}
where $c_2$ is set to a small value (e.g., $c_2 = 0.1$) to impose minimal penalization on the intercept, while $c_1$ is the main tuning parameter.

\subsubsection{Cross-Validation Procedure}

In \eqref{eq:pgmm_split_penalty}, the choice of the tuning parameters $(c_1,c_2)$ is left to the analyst. Since $c_1$ is of utmost importance as it scales the penalty term $\lambda^{*}$, we adopt a $K$-fold cross-validation procedure in the spirit of \cite{caner_kock2018} to empirically select the penalty multiplier $c_1$. Let $\{I_1, \ldots, I_K\}$ denote a partition of the observation indices $\{1, \ldots, n\}$ into $K$ disjoint subsets of approximately equal size $n/K$. For a candidate value $c_1 \in \mathcal{C}$, where $\mathcal{C}$ is a finite grid of candidate values, define the leave-fold-out estimator
\begin{equation*}
	\hat{\rho}_{-k}(c_1) = \argmin_{\rho \in \mathbb{R}^p} \left\{ (\hat{M}_{-k} - \hat{G}_{-k}\rho)'\hat{\Omega}^{-k}_q(\hat{M}_{-k} - \hat{G}_{-k}\rho) + 2c_2\lambda(c_1)|\rho_1| + 2\lambda(c_1)\sum_{j=2}^{p} |\rho_j| \right\},
\end{equation*}
where $\hat{G}_{-k}$, $\hat{M}_{-k}$ and $\hat{\Omega}^{-k}_q$ are computed using observations not in fold $I_k$, and $\lambda(c_1) = c_1 \sqrt{\log q / (n - n_{k})}$.

The cross-validation criterion evaluates the out-of-sample GMM objective on the held-out fold
\begin{equation*}
	\text{CV}(c_1) = \sum_{k=1}^{K} \left( \hat{M}_k - \hat{G}_k \hat{\rho}_{-k}(c_1) \right)' \hat{\Omega}^{k}_q \left( \hat{M}_k - \hat{G}_k \hat{\rho}_{-k}(c_1) \right),
\end{equation*}
where $\hat{G}_k$, $\hat{M}_k$ and $\hat{\Omega}^{k}_q$ are computed using only observations in fold $I_k$. The optimal penalty multiplier is then selected as
\begin{equation*}
	\hat{c}_1 = \argmin_{c_1 \in \mathcal{C}} \text{CV}(c_1).
\end{equation*}
We leave establishing theoretical guarantees of the cross-validated PGMM procedure for future work.


\subsection{Numerical performance}

We evaluate the numerical performance of the PGMM algorithm \ref{algo:pgmm_active_set} in two scenarios: (i) exogenous high-dimensional linear regression, (ii) high-dimensional linear IV regression. 

\subsubsection{HD linear regression}

We borrow the set-up from CNS and compare the performance of PGMM and A-PGMM algorithms with the MD Lasso estimator of CNS as well as with the built-in \texttt{scikit-learn} implementations of the coordinate descent (CD) and least-angle regression (LARS) algorithms\footnote{We use \texttt{LassoCV} and \texttt{LassoLarsCV} commands to run CD and LARS algorithms, respectively.}.

In this design, the data generating process is 
\begin{equation*}
	Y = X'\beta_{0} + \varepsilon,
\end{equation*}
where $X = (1, X_{1},\dots,\,X_{100})'$, $X_{j} \sim \cN(0,\,1)$ and i.i.d., and $\varepsilon \sim \cN(0,\,1)$. The true value of the regression coefficient is $\beta_{0} = (1,\,1,\,1,\,0,\,0,\dots)$ and $\text{dim}(\beta_{0}) = 101$. We can recover $\beta_{0}$ by using the functional $m(w,\,h) = yh(x)$ in the PGMM and MD Lasso formulations\footnote{Alternatively, we could simply use the standard GMM moment $g(w,\,h) = (y - h(x))x$ for the linear regression to implement PGMM.}. 

\begin{table}[htbp]
\centering
\caption{PGMM algorithm performance: High-Dimensional Linear Regression}
\label{tab:CWD_exog}
\begin{tabular}{l|rr|rr|rr}
\toprule
 & \multicolumn{2}{c}{$n=100$} & \multicolumn{2}{c}{$n=1000$} & \multicolumn{2}{c}{$n=10000$} \\
 \cmidrule(lr){2-3} \cmidrule(lr){4-5} \cmidrule(lr){6-7}
Method & MSE & $R^2$ & MSE & $R^2$ & MSE & $R^2$ \\
\midrule
CD & 0.1945 & 0.7752 & 0.0124 & 0.6706 & 0.0014 & 0.6663 \\
LARS & 0.1649 & 0.6985 & 0.0124 & 0.6706 & 0.0014 & 0.6663 \\
MDLasso & 0.1470 & 0.6609 & 0.0118 & 0.6693 & 0.0014 & 0.6662 \\
MDLasso-CV & 0.1710 & 0.6912 & 0.0127 & 0.6705 & 0.0014 & 0.6663 \\
PGMM & 0.1943 & 0.7394 & 0.0127 & 0.6695 & 0.0015 & 0.6660 \\
PGMM-CV & 0.2179 & 0.7327 & 0.0136 & 0.6726 & 0.0015 & 0.6664 \\
A-PGMM & 0.0824 & 0.6494 & 0.0123 & 0.6662 & 0.0014 & 0.6658 \\
A-PGMM-CV & 0.0918 & 0.6561 & 0.0052 & 0.6686 & 0.0006 & 0.6661 \\
\bottomrule
\end{tabular}
\end{table}

Table \ref{tab:CWD_exog} reports the mean squared error $\text{MSE} = |\hat{\beta} - \beta_{0}|_{2}^{2}$ and in-sample $R^2$ across $100$ Monte Carlo replications for three sample sizes, $n \in \{100, 1000, 10000\}$. Several patterns emerge from the results. First, the adaptive variants (A-PGMM and A-PGMM-CV) consistently achieve the lowest MSE across all sample sizes, confirming the theoretical advantages of adaptive penalization. At $n = 100$, A-PGMM attains an MSE of $0.0824$, which is less than half that of the next-best competitor (MDLasso at $0.1470$).

Second, the performance of the cross-validation procedure depends critically on sample size. In the small-sample regime ($n = 100$), the CV variants exhibit marginally higher MSE than their non-CV counterparts, for instance, PGMM-CV yields $0.2179$ versus $0.1943$ for PGMM, and A-PGMM-CV yields $0.0918$ versus $0.0824$ for A-PGMM. This degradation reflects the inherent instability of cross-validation when the effective sample size per fold is small, as the variance in fold-specific estimates inflates the CV criterion. Notably, the default penalty is included in the CV search grid, so the CV procedure can recover the default choice. The observed differences arise from CV occasionally selecting suboptimal values due to noise in the criterion.

As the sample size increases, this pattern reverses. At $n = 1000$, A-PGMM-CV achieves an MSE of $0.0052$, substantially outperforming A-PGMM at $0.0123$. This improvement becomes even more pronounced at $n = 10000$, where A-PGMM-CV attains an MSE of $0.0006$ compared to $0.0014$ for A-PGMM. In large samples, the CV criterion stabilizes and reliably identifies the optimal penalty, yielding substantial gains over the fixed theoretical choice. The $R^2$ values remain comparable across methods within each sample size, indicating that the MSE differences primarily reflect bias-variance tradeoffs in coefficient estimation rather than predictive performance.

\subsubsection{HD linear IV regression}

We follow the exponential design of \cite{belloni2012ecta}. The DGP is 
\begin{align*}
	Y & = X'\beta_{0} + \varepsilon  \\
	X & = \Pi Z + v,
\end{align*}
where $\beta_{0} = (1,\,1,\,1,\,0,\,0,\dots)$ and $\text{dim}(\beta_{0}) = 101$, $X = (1, X_{1},\dots,\,X_{100})'$, $Z = (Z_{1},\dots,\,Z_{150}) \sim \cN(0,\,\Sigma_{Z})$ is a $150 \times 1$ vector with $\mathbb{E}[Z^{2}_{j}] = 1$ and $\text{Corr}(Z_{h},\,Z_{j}) = 0.5^{|h - j|}$. We set the first stage coefficients $\Pi = (1,\,0.7,\,0.7^{2},\dots,\,0.7^{149})$. The structure of the error terms is the following: $\varepsilon \sim \cN(0,\,1)$ and $v|\varepsilon \sim \cN(r \varepsilon, I - r^{2})$ so that the unconditional covariance matrix of the endogenous variables is the identity. We set $r = 0.5$.

We compare the performance of PGMM and A-PGMM algorithms to the Double Lasso estimator of \cite{gold2020}. Table \ref{tab:CWD_endog} demonstrates MSE and $R^{2}$ of the considered implementations based on $100$ simulations. Results are similar to the exogenous case, demonstrating the validity of both PGMM and A-PGMM algorithms.

\begin{table}[htbp]
\centering
\caption{PGMM algorithm performance: High-Dimensional Linear IV Regression}
\label{tab:CWD_endog}
\begin{tabular}{l|rr|rr|rr}
\toprule
 & \multicolumn{2}{c}{$n=100$} & \multicolumn{2}{c}{$n=1000$} & \multicolumn{2}{c}{$n=10000$} \\
 \cmidrule(lr){2-3} \cmidrule(lr){4-5} \cmidrule(lr){6-7}
Method & MSE & $R^2$ & MSE & $R^2$ & MSE & $R^2$ \\
\midrule
DLasso & 0.1755 & 0.9676 & 0.0859 & 0.9524 & 0.1019 & 0.9506 \\
PGMM & 0.2901 & 0.9546 & 0.1419 & 0.9499 & 0.1264 & 0.9496 \\
PGMM-CV & 0.3268 & 0.9571 & 0.1437 & 0.9498 & 0.1307 & 0.9495 \\
A-PGMM & 0.1154 & 0.9507 & 0.0194 & 0.9539 & 0.0465 & 0.9525 \\
A-PGMM-CV & 0.1251 & 0.9508 & 0.0289 & 0.9537 & 0.0457 & 0.9525 \\
\bottomrule
\end{tabular}
\end{table}


\section{Proofs of results}\label{app:proofs}

In this section, we present proofs of the main theoretical results of the paper along with auxiliary lemmas and their corresponding proofs. Note that under Assumption \ref{ass:boundedness}, the rates for $\hat{G}$ and $\hat{\Omega}$ coincide, i.e., $\varepsilon_{n}^{\Omega} = \varepsilon_{n}^{G} = \sqrt{\log q / n}$. Hence, to ease notation, we use just $\varepsilon_{n}^{G}$ throughout the proofs.

\subsection{Properties of the PGMM estimator}

\begin{lemma}\label{lem:g_rate}
If Assumption \ref{ass:boundedness} is satisfied, then
\begin{equation*}
||\hat G - G||_{\infty} = O_{p}(\varepsilon_{n}^{G}), \quad \varepsilon_{n}^{G} = \sqrt{\frac{\log q}{n}}.
\end{equation*}
\end{lemma}

\begin{proof}
The proof is similar to the proof of Lemma A10 of \cite{chernozhukov2022adml}. Define 
\begin{equation*}
	T_{ijk} = d_{j}(X_{i})b_{k}(Z_{i}) - \E[d_{j}(X_{i})b_{k}(Z_{i})],\; U_{jk} = \frac{1}{n}\sum_{i=1}^{n}T_{ijk}.
\end{equation*}
For any constant $C$,
\begin{align*}
	\P(||\hat G - G||_{\infty} \geq C\varepsilon_{n}^{G}) & \leq \sum_{j=1}^{q}\sum_{k=1}^{p}\P(|U_{ijk}| \geq C\varepsilon_{n}^{G}) \\ 
	& \leq pq\max_{j,k}\P(|U_{ijk}| \geq C\varepsilon_{n}^{G}) \\ 
	& \leq q^{2}\max_{j,k}\P(|U_{ijk}| \geq C\varepsilon_{n}^{G}),
\end{align*}
where the last inequality follows from $q \geq p$. Note that $\E[T_{ijk}] = 0$ and by Assumption \ref{ass:boundedness},
\begin{equation*}
	|T_{ijk}| \leq |d_{j}(X_{i})|\cdot|b_{k}(Z_{i})| + \E[|d_{j}(X_{i})|\cdot|b_{k}(Z_{i})|] \leq 2C_{b}C_{d}. 
\end{equation*}
Since $T_{ijk}$ is a bounded random variable, it is sub-Gaussian. Let $||T_{ijk}||_{\Psi_{2}}$ denote the sub-Gaussian norm. Define $K = 2C_{b}C_{d} / \log2 \geq ||T_{ijk}||_{\Psi_{2}}$. By Hoeffding's inequality (see Theorem 2.6.2 in \cite{vershynin2018high}), there is a constant $c$ such that
\begin{align*}
	q^{2}\max_{j,k}\P(|U_{ijk}| \geq C\varepsilon_{n}^{G}) & \leq 2q^{2}\exp\left(-\frac{c(nC\varepsilon_{n}^{G})^{2}}{nK^{2}}\right) \\
	& = 2q^{2}\exp\left(-\frac{cC^{2}\log q}{K^{2}}\right) \\
	& \leq 2\exp\left(\log q\left[2 - \frac{cC^{2}}{K^{2}}\right]\right) \rightarrow 0
\end{align*}
for any $C > K\sqrt{2/c}$. Thus, for large enough $C$, $\P(|\hat G - G|_{\infty} \geq C\varepsilon_{n}^{G}) \rightarrow 0$, which completes the proof.
\end{proof}

\begin{lemma} \label{lem:optimization}
For any $q \times 1$ vector $\hat M$, $q \times p$ matrix $\hat G$, $q \times q$ matrix $\hat \Omega$, and $\lambda > 0$, if 
\begin{equation*}
	\rho^{*} = \argmin_{\rho \in \R^{p}}\left\{(\hat M - \hat G \rho)'\hat{\Omega}_{q}(\hat M - \hat G \rho) + 2\lambda|\rho|_{1}\right\},
\end{equation*}
then
\begin{equation*}
||\hat{G}'\hat{\Omega}_{q}(\hat{M} - \hat{G}\rho^{*})||_{\infty} \leq \lambda.
\end{equation*}
\end{lemma}

\begin{proof}
Since the objective function is convex in $\rho$, a necessary condition for minimization is that zero belongs to the sub-differential of the objective function, i.e.
\begin{equation*}
	0 \in -\hat G'\hat{\Omega}_{q}(\hat M - \hat G \rho^{*}) + \lambda([-1,\, 1],\dots,\,[-1,\,1])'.
\end{equation*}
Thus, for $j = 1,\dots,\,p$ we have
\begin{equation*}
	-e_{j}'\hat G'\hat{\Omega}_{q}(\hat M - \hat G \rho^{*}) + \lambda \geq 0, \quad -e_{j}'\hat G'\hat{\Omega}_{q}(\hat M - \hat G \rho^{*}) - \lambda \leq 0,
\end{equation*}
where $e_{j}$ is the $j^{\text{th}}$ unit vector. Combining two inequalities above yields
\begin{equation*}
|e_{j}'\hat{G}'\hat{\Omega}_{q}(\hat{M} - \hat{G}\rho^{*})| \leq \lambda,
\end{equation*}
which completes the proof as the inequality holds for every $j$.
\end{proof}

Following \cite{bradic2021minimax}, by Assumption \ref{ass:sparse_approx} we can define $S_{\bar\rho} \subset S$ as indices of a sparse approximation with $|S_{\bar\rho}| = \bar s$, where $|A|$ denotes the cardinality of set $A$, and coefficients $\bar\rho = (\bar\rho_{1},\,\dots,\,\bar\rho_{p})'$, with $\bar\rho_{j} = 0$ for $j\not\in S_{\bar\rho}$ such that
\begin{equation*}
	||\rho_{L} - \bar\rho||^{2} \leq C\bar{s}\varepsilon_{n}^{2}.
\end{equation*} 
Also define $\rho_{\star}$ as 
\begin{equation} \label{eq:rho_star_problem}
	\rho_{\star} = \argmin_{v\in\R^{p}} (\rho_{L} - v)'G'\Omega_{q}G(\rho_{L} - v) + 2\varepsilon_{n}\sum_{j\in S_{\bar\rho}^{c}}|v_{j}|.
\end{equation}
Moreover, we assume that $|\rho_{\star}|_{1} = O(1)$.

\begin{lemma} \label{lem:foc_rho_star}
	$||G'\Omega_{q}G(\rho_{\star} - \rho_{L})||_{\infty} \leq \varepsilon_{n}$.
\end{lemma}

\begin{proof}
	Follows directly from the proof of Lemma \ref{lem:optimization}.
\end{proof}

\begin{lemma} \label{lem:rho_star_aux1}
	$(\rho_{L} - \rho_{\star})'G'\Omega_{q}G(\rho_{L} - \rho_{\star}) \leq C\bar{s}\varepsilon_{n}^{2}$.
\end{lemma}

\begin{proof}
	By the definition of $\rho_{\star}$ and the fact that the largest eigenvalue of $G'\Omega_{q}G$ is bounded, we have
	\begin{align*}
		(\rho_{L} - \rho_{\star})'G'\Omega_{q}G(\rho_{L} - \rho_{\star}) + 2\varepsilon_{n}\sum_{j\in S_{\bar\rho}^{c}}|\rho_{\star,j}| & \leq (\rho_{L} - \bar\rho)'G'\Omega_{q}G(\rho_{L} - \bar\rho) + 2\varepsilon_{n}\sum_{j\in S_{\bar\rho}^{c}}|\bar\rho_{j}| \\
		& = (\rho_{L} - \bar\rho)'G'\Omega_{q}G(\rho_{L} - \bar\rho) \\ 
		& \leq C||\rho_{L} - \bar\rho||^{2} \leq C\bar{s}\varepsilon_{n}^{2}.
	\end{align*}
\end{proof}

\begin{lemma} \label{lem:rho_star_aux2}
	Let $S_{\rho_{\star}}$ be the vector of indices of nonzero elements of $\rho_{\star}$. Then, $s_{\star} \equiv |S_{\rho_{\star}}| \leq C\bar{s}$.
\end{lemma}

\begin{proof}
	For all $j \in S_{\rho_{\star}} \backslash S_{\bar\rho}$ the first order conditions to equation \eqref{eq:rho_star_problem} imply $|e_{j}'G'\Omega_{q}G(\rho_{\star} - \rho_{L})| = \varepsilon_{n}$. Therefore, it follows that
	\begin{equation*}
		\sum_{j\in S_{\rho_{\star}} \backslash S_{\bar\rho}} \left(e_{j}'G'\Omega_{q}G(\rho_{\star} - \rho_{L})\right)^{2} = \varepsilon_{n}^{2}|S_{\rho_{\star}} \backslash S_{\bar\rho}|.
	\end{equation*}
	Moreover, using Lemma \ref{lem:rho_star_aux1} and the fact that the largest eigenvalue of $G'\Omega_{q}G$ is bounded, we get
	\begin{align*}
		\sum_{j\in S_{\rho_{\star}} \backslash S_{\bar\rho}} \left(e_{j}'G'\Omega_{q}G(\rho_{\star} - \rho_{L})\right)^{2} & \leq \sum_{j=1}^{p} \left(e_{j}'G'\Omega_{q}G(\rho_{\star} - \rho_{L})\right)^{2} \\ 
		& = (\rho_{\star} - \rho_{L})'G'\Omega_{q}G\left(\sum_{j=1}^{p}e_{j}e_{j}'\right)G'\Omega_{q}G(\rho_{\star} - \rho_{L}) \\
		& = (\rho_{\star} - \rho_{L})(G'\Omega_{q}G)^{2}(\rho_{\star} - \rho_{L}) \\
		& \leq \lambda_{max}(G'\Omega_{q}G)\{(\rho_{\star} - \rho_{L})G'\Omega_{q}G(\rho_{\star} - \rho_{L})\} \leq C\bar{s}\varepsilon_{n}^{2}.
	\end{align*}
	Combining the results above, we obtain
	\begin{equation*}
		\varepsilon_{n}^{2}|S_{\rho_{\star}} \backslash S_{\bar\rho}| \leq C\bar{s}\varepsilon_{n}^{2}.
	\end{equation*}
	Dividing both sides by $\varepsilon_{n}^{2}$ gives $|S_{\rho_{\star}} \backslash S_{\bar\rho}| \leq C\bar{s}$. As a result,
	\begin{equation*}
		s_{\star} = |S_{\bar\rho}| + |S_{\rho_{\star}} \backslash S_{\bar\rho}| \leq \bar{s} + C\bar{s} \leq C\bar{s}.
	\end{equation*}
\end{proof}

\begin{lemma}\label{lem:rho_star_aux3}
	Let $B = \E[b(Z)b(Z)']$ has its largest eigenvalue bounded uniformly in $n$, then $||\alpha_{0} - b'\rho_{\star}||^{2} \leq C\bar{s}\varepsilon_{n}^{2}$.
\end{lemma}

\begin{proof}
	By the triangle inequality and Assumption \ref{ass:sparse_approx},
	\begin{align*}
		||\alpha_{0} - b'\rho_{\star}||^{2} & \leq ||\alpha_{0} - b'\bar\rho||^{2} + ||b'(\bar\rho - \rho_{L})||^{2} + ||b'(\rho_{L} - \rho_{\star})||^{2} \\
		& \leq C\bar{s}\varepsilon_{n}^{2} + ||b'(\bar\rho - \rho_{L})||^{2} + ||b'(\rho_{L} - \rho_{\star})||^{2}.
	\end{align*}
	Moreover, by the definition of $\bar\rho$ and $\lambda_{max}(B) \leq C$,
	\begin{equation*}
		||b'(\bar\rho - \rho_{L})||^{2} \leq \lambda_{max}(B)||\bar\rho - \rho_{L}||^{2} \leq C\bar{s}\varepsilon_{n}^{2}.
	\end{equation*}
	Also, by Lemma \ref{lem:rho_star_aux1},
	\begin{equation*}
		||b'(\rho_{L} - \rho_{\star})||^{2} \leq \lambda_{max}(B)||\rho_{L} - \rho_{\star}||^{2} \leq C\bar{s}\varepsilon_{n}^{2},
	\end{equation*}
	which completes the proof.
\end{proof}

\begin{lemma} \label{lem:aux1}
If Assumptions \ref{ass:weight_matrix}--\ref{ass:M_conv} and \ref{ass:m_bound} are satisfied, then
\begin{equation*}
	||\hat G'\hat{\Omega}_{q}(\hat M - \hat G\rho_{\star})||_{\infty} = O_p(\varepsilon_{n}).
\end{equation*}
\end{lemma}

\begin{proof}
By the triangle inequality,
\begin{align}
	||\hat G'\hat{\Omega}_{q}(\hat M - \hat G\rho_{\star})||_{\infty} & \leq ||\hat G'\hat{\Omega}_{q} \hat M - G'\Omega_{q} M||_{\infty}  \label{eq1:1} \\
	& + ||G'\Omega_{q} M - G'\Omega_{q} G\rho_{\star}||_{\infty} \label{eq1:2} \\
	& + ||(G'\Omega_{q} G - \hat G'\hat{\Omega}_{q} \hat G)\rho_{\star}||_{\infty}. \label{eq1:3}
\end{align}
Consider the first element \eqref{eq1:1}. Note that by the triangle inequality,
\begin{align}
||\hat G'\hat{\Omega}_{q} \hat M - G'\Omega_{q} M||_{\infty} & \leq ||(\hat G - G)'(\hat{\Omega}_{q} - \Omega_{q})(\hat M - M)||_{\infty} \label{eq2:1} \\
& + ||(\hat G - G)'\Omega_{q}(\hat M - M)||_{\infty} \label{eq2:2} \\
& + ||G'(\hat{\Omega}_{q} - \Omega_{q})(\hat M- M)||_{\infty} \label{eq2:3} \\
& + ||G'\Omega_{q}(\hat M- M)||_{\infty} \label{eq2:4} \\ 
& + ||(\hat G - G)'\Omega_{q}M||_{\infty} \label{eq2:5} \\
& + ||(\hat G - G)'(\hat{\Omega}_{q} - \Omega_{q})M||_{\infty} \label{eq2:6} \\ 
& + ||G'(\hat{\Omega}_{q} - \Omega_{q})M||_{\infty}. \label{eq2:7}
\end{align}
Now we will bound every term on the RHS of the inequality above. To do so, we will use the following matrix norm inequality from \cite{caner_kock2018}. For any $q \times p$ matrix $A$, $p \times q$ matrix $B$, and $q \times q$ matrix $F$ the following inequality holds
\begin{equation}\label{eq:matrix_ineq}
	||BFA||_{\infty} \leq q||B||_{\infty}||F||_{\ell_{\infty}}||A||_{\infty}.
\end{equation}

We can use \eqref{eq:matrix_ineq} to put an upper bound on \eqref{eq2:1},
\begin{align*}
||(\hat G - G)'(\hat{\Omega}_{q} - \Omega_{q})(\hat M - M)||_{\infty} & \leq ||\hat G - G||_{\infty}||\hat{\Omega} - \Omega||_{\ell_{\infty}}||\hat M - M||_{\infty} \\ & = O_p(\varepsilon^{G}_{n})O_p(\varepsilon^{G}_{n})O_p(\varepsilon^{M}_{n}) = O_p(\varepsilon^{3}_{n}).
\end{align*} 
Moreover, notice that Assumptions \ref{ass:boundedness} and \ref{ass:m_bound} imply that $||G||_{\infty} = O(1)$ and $||M||_{\infty} = O(1)$.
Using this fact and \eqref{eq:matrix_ineq}, we can bound the remaining terms \eqref{eq2:2}--\eqref{eq2:7},
\begin{align*}
||(\hat G - G)'\Omega_{q}(\hat M - M)||_{\infty} & \leq ||\hat G - G||_{\infty}||\Omega||_{\ell_{\infty}}||\hat M - M||_{\infty} = O_p(\varepsilon^{G}_{n})O(1)O_p(\varepsilon^{M}_{n}) = O_p(\varepsilon^{2}_{n})\\
||G'(\hat{\Omega}_{q} - \Omega_{q})(\hat M- M)||_{\infty} & \leq ||G||_{\infty}||\hat{\Omega} - \Omega||_{\ell_{\infty}}||\hat M - M||_{\infty} = O(1)O_p(\varepsilon^{G}_{n})O_p(\varepsilon^{M}_{n}) = O_p(\varepsilon^{2}_{n}) \\
||G'\Omega_{q}(\hat M- M)||_{\infty} & \leq ||G||_{\infty}||\Omega||_{\ell_{\infty}}||\hat M - M||_{\infty} = O(1)O_p(\varepsilon^{M}_{n}) = O_p(\varepsilon^{M}_{n}) \\
||(\hat G - G)'\Omega_{q}M||_{\infty} & \leq ||\hat G - G||_{\infty}||\Omega||_{\ell_{\infty}}||M||_{\infty} = O_p(\varepsilon^{G}_{n})O(1) = O_p(\varepsilon^{G}_{n}) \\
||(\hat G - G)'(\hat{\Omega}_{q} - \Omega_{q})M||_{\infty} & \leq ||\hat G - G||_{\infty}||\hat{\Omega} - \Omega||_{\ell_{\infty}}||M||_{\infty} = O_p(\varepsilon^{G}_{n})O_p(\varepsilon^{G}_{n})O(1) = O_p((\varepsilon^{G}_{n})^{2}) \\
||G'(\hat{\Omega}_{q} - \Omega_{q})M||_{\infty} & \leq ||G||_{\infty}||\hat{\Omega} - \Omega||_{\ell_{\infty}}||M||_{\infty} = O(1)O_p(\varepsilon^{G}_{n}) = O_p(\varepsilon^{G}_{n}).
\end{align*}
Collecting all the terms gives the upper bound for \eqref{eq1:1}
\begin{equation*}
||\hat G'\hat{\Omega}_{q} \hat M - G'\Omega_{q} M||_{\infty} = O_p(\varepsilon_{n}).
\end{equation*}

Next, by the triangle and H\"{o}lder's inequalities,
\begin{align*}
	||G'\Omega_{q} M - G'\Omega_{q} G\rho_{\star}||_{\infty} & \leq ||G'\Omega_{q} M - G'\Omega_{q} G\rho_{L}||_{\infty} + ||G'\Omega_{q}G(\rho_{L} - \rho_{\star})||_{\infty}.
\end{align*}
By Lemma \ref{lem:optimization} and the fact that $\rho_{L}$ are the population PGMM coefficients,
\begin{equation*}
	||G'\Omega_{q} M - G'\Omega_{q} G\rho_{L}||_{\infty} \leq \varepsilon_{n}.
\end{equation*}
Moreover, by Lemma \ref{lem:foc_rho_star},
\begin{equation*}
	||G'\Omega_{q}G(\rho_{L} - \rho_{\star})||_{\infty} \leq \varepsilon_{n}.
\end{equation*}
Thus, using the results above,
\begin{equation*}
	||G'\Omega_{q} M - G'\Omega_{q} G\rho_{\star}||_{\infty} = O(\varepsilon_{n}).
\end{equation*}

We are left with putting an upper bound on \eqref{eq1:3}. By H\"{o}lder's inequality,
\begin{equation*}
||(G'\Omega_{q} G - \hat G'\hat{\Omega}_{q} \hat G)\rho_{\star}||_{\infty} \leq ||G'\Omega_{q} G - \hat G'\hat{\Omega}_{q} \hat G||_{\infty}|\rho_{\star}|_{1}.
\end{equation*}
Moreover, by the triangle inequality,
\begin{align*}
||G'\Omega_{q} G - \hat G'\hat{\Omega}_{q} \hat G||_{\infty} &  \leq ||(\hat G - G)'(\hat{\Omega}_{q} - \Omega_{q})(\hat G - G)||_{\infty} \\
& + 2||(\hat G - G)'(\hat{\Omega}_{q} - \Omega_{q})G||_{\infty} \\
& + ||(\hat G - G)'\Omega_{q}(\hat G - G)||_{\infty} \\
& + 2||(\hat G - G)'\Omega_{q}G||_{\infty} \\
& + ||G'(\hat{\Omega}_{q} - \Omega_{q})G||_{\infty}.
\end{align*}
Using \eqref{eq:matrix_ineq}, we can bound all the terms on the RHS of the inequality above,
\begin{align*}
||(\hat G - G)'(\hat{\Omega}_{q} - \Omega_{q})(\hat G - G)||_{\infty} & \leq ||\hat G - G||_{\infty}||\hat{\Omega} - \Omega||_{\ell_{\infty}}||\hat G - G||_{\infty} = O_p((\varepsilon^{G}_{n})^{3}) \\
2||(\hat G - G)'(\hat{\Omega}_{q} - \Omega_{q})G||_{\infty} & \leq 2||\hat G - G||_{\infty}||\hat{\Omega} - \Omega||_{\ell_{\infty}}||G||_{\infty} = O_p((\varepsilon^{G}_{n})^{2})O(1) = O_p((\varepsilon^{G}_{n})^{2}) \\
||(\hat G - G)'\Omega_{q}(\hat G - G)||_{\infty} & \leq ||\hat G - G||_{\infty}||\Omega||_{\ell_{\infty}}||\hat G - G||_{\infty} = O_p((\varepsilon^{G}_{n})^{2})O(1) = O_p((\varepsilon^{G}_{n})^{2}) \\
2||(\hat G - G)'\Omega_{q}G||_{\infty} & \leq 2||\hat G - G||_{\infty}||\Omega||_{\ell_{\infty}}||G||_{\infty} = O_p(\varepsilon^{G}_{n})O(1) = O_p(\varepsilon^{G}_{n}) \\
||G'(\hat{\Omega}_{q} - \Omega_{q})G||_{\infty} & \leq ||G||_{\infty}||\hat{\Omega} - \Omega||_{\ell_{\infty}}||G||_{\infty} = O_p(\varepsilon^{G}_{n})O(1) = O_p(\varepsilon^{G}_{n}).
\end{align*}
Collecting all the terms gives,
\begin{equation*}
||(G'\Omega_{q} G - \hat G'\hat{\Omega}_{q} \hat G)||_{\infty} = O_p(\varepsilon^{G}_{n}).
\end{equation*}
Combining the result above with $|\rho_{\star}|_{1}=O(1)$ yields
\begin{equation*}
||(G'\Omega_{q} G - \hat G'\hat{\Omega}_{q} \hat G)\rho_{\star}||_{\infty} = O_p(\varepsilon^{G}_{n})O(1) = O_p(\varepsilon^{G}_{n}).
\end{equation*}
Collecting all the terms for \eqref{eq1:1}--\eqref{eq1:3} gives us the desired upper bound,
\begin{equation*}
||\hat G'\hat{\Omega}_{q}(\hat M - \hat G\rho_{\star})||_{\infty} = O_p(\varepsilon_{n}) + O(\varepsilon_{n}) + O_p(\varepsilon^{G}_{n}) = O_p(\varepsilon_{n}).
\end{equation*}
\end{proof}

Let $\phi^{2}(s_{\star})$ denote the population restricted eigenvalue from Assumption \ref{ass:re_pgmm} at $s = s_{\star}$,
\begin{equation*}
\phi^{2}(s_{\star}) = \inf\left\{\frac{\delta'G'\Omega_{q} G\delta}{||\delta_{S_{\rho_{\star}}}||^{2}}:\delta\in\R^{p} \backslash \{0\},\; |\delta_{S_{\rho_{\star}}^{c}}|_{1} \leq 3|\delta_{S_{\rho_{\star}}}|_{1},\, |S_{\rho_{\star}}|\leq s_{\star}\right\}.
\end{equation*}
Next, let us introduce an empirical version of the condition above, 
\begin{equation*}
\hat\phi^{2}(s_{\star}) = \inf\left\{\frac{\delta'\hat G'\hat \Omega_{q} \hat G\delta}{||\delta_{S_{\rho_{\star}}}||^{2}}:\delta\in\R^{p} \backslash \{0\},\; |\delta_{S_{\rho_{\star}}^{c}}|_{1} \leq 3|\delta_{S_{\rho_{\star}}}|_{1},\, |S_{\rho_{\star}}|\leq s_{\star}\right\}.
\end{equation*}
In the following Lemma we show that we can bound $\hat\phi^{2}(s_{\star})$ from below, which will be useful in the proof of Theorem \ref{thm:rr_bound}.

\begin{lemma}\label{lem:re_bound}
If Assumptions \ref{ass:boundedness} and \ref{ass:boundedness} are satisfied, then
\begin{equation*}
	\hat\phi^{2}(s_{\star}) \geq \phi^{2}(s_{\star}) - O_p(s_{\star}\varepsilon_{n}^{G}).
\end{equation*}	
\end{lemma}

\begin{proof}
The proof follows the proof of Lemma S3 in \cite{caner_kock2018}. By adding and subtracting $G$ and $\Omega_{q}$ and the reverse triangle inequality,
\begin{align*}
	|\delta'\hat G'\hat \Omega_{q} \hat G\delta| & = |\delta'(\hat G - G + G)'(\hat \Omega_{q} - \Omega_{q} + \Omega_{q})(\hat G - G + G)\delta| \\
	& \geq |\delta'G'\Omega_{q}G\delta| \\ 
	& - |\delta'(\hat G - G)'(\hat \Omega_{q} - \Omega_{q})(\hat G - G)\delta| \\
	& - |\delta'(\hat G - G)'\Omega_{q}(\hat G - G)\delta| \\
	& - |\delta'G'(\hat \Omega_{q} - \Omega_{q})G\delta| \\
	& - 2|\delta'(\hat G - G)'(\hat \Omega_{q} - \Omega_{q})G\delta| \\
	& - 2|\delta'(\hat G - G)'\Omega_{q}G\delta|.
\end{align*} 
The following inequality from \cite{caner_kock2018} will help us bound the expression above. For any $q \times p$ matrix $A$, $p \times q$ matrix $B$, $q \times q$ matrix $F$, and $p \times 1$ vector $x$ the following inequality holds
\begin{equation}\label{eq:matrix_ineq2}
	|x'BFAx| \leq q|x|^{2}_{1}||B||_{\infty}||F||_{\ell_{\infty}}||A||_{\infty}.
\end{equation}
Using \eqref{eq:matrix_ineq2}, we get
\begin{align*}
|\delta'(\hat G - G)'(\hat \Omega_{q} - \Omega_{q})(\hat G - G)\delta| & \leq |\delta|_{1}^{2}||\hat G - G||^{2}_{\infty}||\hat \Omega - \Omega||_{\ell_{\infty}} \\
|\delta'(\hat G - G)'\Omega_{q}(\hat G - G)\delta| & \leq |\delta|_{1}^{2}||\hat G - G||^{2}_{\infty}||\Omega||_{\ell_{\infty}} \\
|\delta'G'(\hat \Omega_{q} - \Omega_{q})G\delta| & \leq |\delta|_{1}^{2}||G||^{2}_{\infty}||\hat \Omega - \Omega||_{\ell_{\infty}} \\
2|\delta'(\hat G - G)'(\hat \Omega_{q} - \Omega_{q})G\delta| & \leq 2|\delta|_{1}^{2}||\hat G - G||_{\infty}||\hat \Omega - \Omega||_{\ell_{\infty}}||G||_{\infty} \\
2|\delta'(\hat G - G)'\Omega_{q}G\delta| & \leq 2|\delta|_{1}^{2}||\hat G - G||_{\infty}||\Omega||_{\ell_{\infty}}||G||_{\infty}.
\end{align*}
Combining the terms gives
\begin{align} \label{eq3}
|\delta'\hat G'\hat \Omega_{q} \hat G\delta| & \geq |\delta'G'\Omega_{q}G\delta| \\
& - |\delta|_{1}^{2}||\hat G - G||^{2}_{\infty}(||\hat \Omega - \Omega||_{\ell_{\infty}} + ||\Omega||_{\infty}) \notag \\
& - |\delta|_{1}^{2}||G||^{2}_{\infty}||\hat \Omega - \Omega||_{\ell_{\infty}} \notag \\
& - 2|\delta|_{1}^{2}||\hat G - G||_{\infty}||G||_{\infty}(||\hat \Omega - \Omega||_{\ell_{\infty}} + ||\Omega||_{\infty}) \notag
\end{align}

Recall, we have the restriction
\begin{equation*}
|\delta_{S_{\rho_{\star}}^{c}}|_{1} \leq 3|\delta_{S_{\rho_{\star}}}|_{1} \leq 3\sqrt{s_{\star}}||\delta_{S_{\rho_{\star}}}||
\end{equation*}
where the second inequality is Cauchy-Schwarz. Adding $|\delta_{S_{\rho_{\star}}}|$ to both sides gives
\begin{equation} \label{eq4}
	|\delta|_{1} \leq 4 \sqrt{s_{\star}}||\delta_{S_{\rho_{\star}}}|| \Rightarrow \frac{|\delta|_{1}^{2}}{||\delta_{S_{\rho_{\star}}}||^{2}} \leq 16s_{\star}.
\end{equation}
Divide \eqref{eq3} by $||\delta_{S_{\rho_{\star}}}||^{2}$ and use \eqref{eq4},
\begin{align*}
\frac{|\delta'\hat G'\hat \Omega_{q} \hat G\delta|}{||\delta_{S_{\rho_{\star}}}||^{2}} & \geq \frac{|\delta'G'\Omega_{q}G\delta|}{||\delta_{S_{\rho_{\star}}}||^{2}} \\
& - 16s_{\star}||\hat G - G||^{2}_{\infty}(||\hat \Omega - \Omega||_{\ell_{\infty}} + ||\Omega||_{\infty})  \\
& - 16s_{\star}||G||^{2}_{\infty}||\hat \Omega - \Omega||_{\ell_{\infty}}  \\
& - 32s_{\star}||\hat G - G||_{\infty}||G||_{\infty}(||\hat \Omega - \Omega||_{\ell_{\infty}} + ||\Omega||_{\infty}).
\end{align*}
Since $\frac{|\delta'G'\Omega_{q}G\delta|}{||\delta_{S_{\rho_{\star}}}||^{2}} \geq \phi^{2}(s_{\star})$ for all $\delta$ satisfying $|\delta_{S_{\rho_{\star}}^{c}}|_{1} \leq 3|\delta_{S_{\rho_{\star}}}|_{1}$, minimizing the LHS of the inequality above over such $\delta$ yields
\begin{equation*}
	\hat\phi^{2}(s_{\star}) \geq \phi^{2}(s_{\star}) - a_n,
\end{equation*}
where 
\begin{align*}
a_n & = 16s_{\star}||\hat G - G||^{2}_{\infty}(||\hat \Omega - \Omega||_{\ell_{\infty}} + ||\Omega||_{\infty})  \\
& + 16s_{\star}||G||^{2}_{\infty}||\hat \Omega - \Omega||_{\ell_{\infty}}  \\
& + 32s_{\star}||\hat G - G||_{\infty}||G||_{\infty}(||\hat \Omega - \Omega||_{\ell_{\infty}} + ||\Omega||_{\infty}).
\end{align*}
Using Assumptions \ref{ass:weight_matrix} and \ref{ass:boundedness} and Lemma \ref{lem:g_rate}, we can put an upper bound on $a_n$ as follows
\begin{align*}
16s_{\star}||\hat G - G||^{2}_{\infty}(||\hat \Omega - \Omega||_{\ell_{\infty}} + ||\Omega||_{\infty}) & = 16s_{\star}O_p((\varepsilon_{n}^{G})^{2})(O_p(\varepsilon^{G}_{n}) + O(1)) = O_p(s_{\star}(\varepsilon_{n}^{G})^{3}) \\
16s_{\star}||G||^{2}_{\infty}||\hat \Omega - \Omega||_{\ell_{\infty}} & =16s_{\star}O(1)O_p(\varepsilon^{G}_{n}) = O_p(s_{\star}\varepsilon^{G}_{n})\\
32s_{\star}||\hat G - G||_{\infty}||G||_{\infty}(||\hat \Omega - \Omega||_{\ell_{\infty}} + ||\Omega||_{\infty}) & = 32s_{\star}O_p(\varepsilon_{n}^{G})O(1)(s_{\star}  + O(1)) = O_p(s_{\star}(\varepsilon_{n}^{G})^{2}).
\end{align*}
Gathering the terms gives
\begin{equation*}
	a_{n} = O_p(s_{\star}\varepsilon_{n}^{G}),
\end{equation*}
which completes the proof.
\end{proof}

\subsubsection{Proof of Theorem \ref{thm:rr_bound}}

As $\hat{\Omega}$ is positive definite, we can write
\begin{equation*}
\hat{\rho}_{L} = \argmin_{\rho \in \R^{p}}\{||\hat{\Omega}^{1/2}_{q}(\hat M - \hat G \rho)||^{2} + 2\lambda_{n}|\rho|_{1}\}.
\end{equation*}
The minimizing property of $\hat \rho_{L}$ implies
\begin{equation} \label{eq:min_prop}
||\hat{\Omega}_{q}^{1/2}(\hat M - \hat G \hat\rho_{L})||^{2} + 2\lambda_{n}|\hat\rho_{L}|_{1} \leq ||\hat{\Omega}_{q}^{1/2}(\hat M - \hat G \rho_{\star})||^{2} + 2\lambda_{n}|\rho_{\star}|_{1}.
\end{equation}
First, observe that
\begin{align*}
||\hat{\Omega}_{q}^{1/2}(\hat M - \hat G \hat\rho_{L})||^{2} - ||\hat{\Omega}_{q}^{1/2}(\hat M - \hat G \rho_{\star})||^{2} & = (\hat M - \hat G \hat\rho_{L})'\hat\Omega_{q}(\hat M - \hat G \hat\rho_{L}) - (\hat M - \hat G \rho_{\star})'\hat\Omega_{q}(\hat M - \hat G \rho_{\star}) \\
& = \hat\rho_{L}\hat G'\hat\Omega_{q}\hat G\hat\rho_{L} - \rho_{\star}\hat G'\hat\Omega_{q}\hat G\rho_{\star} - 2(\hat G'\hat\Omega_{q}\hat M)'(\hat\rho_{L} - \rho_{\star}) \\
& = (\hat\rho_{L} - \rho_{\star})'\hat G'\hat\Omega_{q}\hat G(\hat\rho_{L} - \rho_{\star}) +2\rho_{\star}'\hat G'\hat\Omega_{q}\hat G(\hat\rho_{L} - \rho_{\star})  \\
& - 2(\hat G'\hat\Omega_{q}\hat M)'(\hat\rho_{L} - \rho_{\star}) \\
& = ||\hat\Omega_{q}^{1/2}\hat G(\hat\rho_{L} - \rho_{\star})||^{2} - 2(\hat G'\hat\Omega_{q}\hat M - \hat G'\hat\Omega_{q}\hat G\rho_{\star})'(\hat\rho_{L} - \rho_{\star}).
\end{align*}
Plug the expression above in \eqref{eq:min_prop} to get
\begin{align}
||\hat\Omega_{q}^{1/2}\hat G(\hat\rho_{L} - \rho_{\star})||^{2} + 2\lambda_{n}|\hat\rho_{L}|_{1} & \leq 2(\hat G'\hat\Omega_{q}\hat M - \hat G'\hat\Omega_{q}\hat G\rho_{\star})'(\hat\rho_{L} - \rho_{\star}) + 2\lambda_{n}|\rho_{\star}|_{1} \notag \\
& \leq 2||\hat G'\hat\Omega_{q}\hat M - \hat G'\hat\Omega_{q}\hat G\rho_{\star}||_{\infty}|\hat\rho_{L} - \rho_{\star}|_{1} + 2\lambda_{n}|\rho_{\star}|_{1} \notag \\
& = 2||\hat G'\hat\Omega_{q}(\hat M - \hat G\rho_{\star})||_{\infty}|\hat\rho_{L} - \rho_{\star}|_{1} + 2\lambda_{n}|\rho_{\star}|_{1} \notag \\
& = 2o_p(\lambda_{n})|\hat\rho_{L} - \rho_{\star}|_{1} + 2\lambda_{n}|\rho_{\star}|_{1} \label{eq:aux_rho} ,
\end{align}
where the second inequality is H\"{o}lder and the last equality comes from Lemma \ref{lem:aux1} and the fact that $\varepsilon_{n} = o(\lambda_{n})$. Hence, with probability approaching one,
\begin{equation*}
||\hat\Omega_{q}^{1/2}\hat G(\hat\rho_{L} - \rho_{\star})||^{2} + 2\lambda_{n}|\hat\rho_{L}|_{1} \leq 2\lambda_{n}|\hat\rho_{L} - \rho_{\star}|_{1} + 2\lambda_{n}|\rho_{\star}|_{1}.
\end{equation*}

Next, note that $|\hat\rho_{L}|_{1} = |\hat\rho_{L, S_{\rho_{\star}}}|_{1} + |\hat\rho_{L, S^{c}_{\rho_{\star}}}|_{1}$ and $|\rho_{\star}|_{1} = |\rho_{\star, S_{\rho_{\star}}}|_{1}$ as $|\rho_{\star, S^{c}_{\rho_{\star}}}|_{1} = 0$. Therefore,
\begin{align*}
||\hat\Omega_{q}^{1/2}\hat G(\hat\rho_{L} - \rho_{\star})||^{2} + 2\lambda_{n}|\hat\rho_{L,S^{c}_{\rho_{\star}}}|_{1} & \leq 2\lambda_{n}|\hat\rho_{L} - \rho_{\star}|_{1} + 2\lambda_{n}(|\rho_{\star, S_{\rho_{\star}}}|_{1} - |\hat\rho_{L, S_{\rho_{\star}}}|_{1})\\
& \leq 2\lambda_{n}|\hat\rho_{L} - \rho_{\star}|_{1} + 2\lambda_{n}|\hat\rho_{L, S_{\rho_{\star}}} - \rho_{\star, S_{\rho_{\star}}}|_{1},
\end{align*}
where the second line comes from the reverse triangle inequality. Using that $|\hat\rho_{L} - \rho_{\star}|_{1} = |\hat\rho_{L, S_{\rho_{\star}}} - \rho_{\star, S_{\rho_{\star}}}|_{1} + |\hat\rho_{L, S^{c}_{\rho_{\star}}}|_{1}$ gives
\begin{equation}\label{eq:re_aux}
||\hat\Omega_{q}^{1/2}\hat G(\hat\rho_{L} - \rho_{\star})||^{2} + \lambda_{n}|\hat\rho_{L,S^{c}_{\rho_{\star}}}|_{1} \leq 3\lambda_{n}|\hat\rho_{L, S_{\rho_{\star}}} - \rho_{\star, S_{\rho_{\star}}}|_{1}.
\end{equation}
The inequality in \eqref{eq:re_aux} implies $\lambda_{n}|\hat\rho_{L,S^{c}_{\rho_{\star}}}|_{1} \leq 3\lambda_{n}|\hat\rho_{L, S_{\rho_{\star}}} - \rho_{\star, S_{\rho_{\star}}}|_{1}$ leading to $|\hat\rho_{L,S^{c}_{\rho_{\star}}}|_{1} \leq 3|\hat\rho_{L, S_{\rho_{\star}}} - \rho_{\star, S_{\rho_{\star}}}|_{1}$, meaning that the restricted eigenvalue condition is satisfied. Note that by Cauchy-Schwarz inequality, $|\hat\rho_{L, S_{\rho_{\star}}} - \rho_{\star, S_{\rho_{\star}}}|_{1} \leq \sqrt{s_{\star}}||\hat\rho_{L, S_{\rho_{\star}}} - \rho_{\star, S_{\rho_{\star}}}||$. Using this along with the restricted eigenvalue condition on  \eqref{eq:re_aux} yields
\begin{equation*}
||\hat\Omega_{q}^{1/2}\hat G(\hat\rho_{L} - \rho_{\star})||^{2} + \lambda_{n}|\hat\rho_{L,S^{c}_{\rho_{\star}}}|_{1} \leq 3\lambda_{n}\sqrt{s_{\star}}||\hat\rho_{L, S_{\rho_{\star}}} - \rho_{\star, S_{\rho_{\star}}}|| \leq 3\lambda_{n}\sqrt{s_{\star}}\frac{||\hat\Omega_{q}^{1/2}\hat G'(\hat\rho_{L} - \rho_{\star})||}{\hat\phi(s_{\star})}.
\end{equation*}
Note that by AM-GM inequality, 
\begin{equation*}
||\hat\Omega_{q}^{1/2}\hat G(\hat\rho_{L} - \rho_{\star})||^{2} + \lambda_{n}|\hat\rho_{L,S^{c}_{\rho_{\star}}}|_{1} \leq \frac{1}{2}||\hat\Omega_{q}^{1/2}\hat G'(\hat\rho_{L} - \rho_{\star})||^{2} + \frac{9}{2}\frac{\lambda_{n}^{2}s_{\star}}{\hat\phi^{2}(s_{\star})}.
\end{equation*}
Multiplying both sides by $2$ and collecting terms gives
\begin{equation}\label{eq:aux_re_bound}
||\hat\Omega_{q}^{1/2}\hat G(\hat\rho_{L} - \rho_{\star})||^{2} + 2\lambda_{n}|\hat\rho_{L,S^{c}_{\rho_{\star}}}|_{1} \leq \frac{9\lambda_{n}^{2}s_{\star}}{\hat\phi^{2}(s_{\star})}.
\end{equation}

To get the $\ell_{1}$-error bound, ignore the first term on the LHS of \eqref{eq:re_aux} and add $\lambda_{n}|\hat\rho_{L, S_{\rho_{\star}}} - \rho_{\star, S_{\rho_{\star}}}|_{1}$ to both sides,
\begin{equation*}
\lambda_{n}|\hat\rho_{L} - \rho_{\star}|_{1} \leq 4\lambda_{n}|\hat\rho_{L, S_{\rho_{\star}}} - \rho_{\star, S_{\rho_{\star}}}|_{1}.
\end{equation*}
By Cauchy-Schwarz inequality and the restricted eigenvalue condition,
\begin{equation*}
\lambda_{n}|\hat\rho_{L} - \rho_{\star}|_{1} \leq 4\lambda_{n}\sqrt{s_{\star}}||\hat\rho_{L, S_{\rho_{\star}}} - \rho_{\star, S_{\rho_{\star}}}|| \leq 4\lambda_{n}\sqrt{s_{\star}}\frac{||\hat\Omega_{q}^{1/2}\hat G(\hat\rho_{L} - \rho_{\star})||}{\hat\phi(s_{\star})}.
\end{equation*}
The bound on $||\hat\Omega_{q}^{1/2}\hat G'(\hat\rho_{L} - \rho_{\star})||^{2}$ in \eqref{eq:aux_re_bound} implies
\begin{equation*}
|\hat\rho_{L} - \rho_{\star}|_{1} \leq \frac{12\lambda_{n}s_{\star}}{\hat\phi^{2}(s_{\star})}.
\end{equation*}

Next, by Lemma \ref{lem:re_bound} and $\varepsilon_{n} = o(\lambda_{n})$,
\begin{equation}\label{eq:rho_bound_phi}
|\hat\rho_{L} - \rho_{\star}|_{1} \leq \frac{12\lambda_{n}s_{\star}}{\phi^{2}(s_{\star}) - o_p(s_{\star}\lambda_{n})}.
\end{equation}
By Lemma \ref{lem:rho_star_aux2}, $s_\star = O(\bar{s})$, which together with $\bar{s}\lambda_n = o(1)$ implies $s_\star \lambda_n = o(1)$. Given $\phi^{2}(s_{\star})$ is bounded away from zero by Assumption \ref{ass:re_pgmm}, we get
\begin{equation}\label{eq:rho_bound}
|\hat\rho_{L} - \rho_{\star}|_{1} = O_p(\bar{s}\lambda_{n}).
\end{equation}
Finally, let $\alpha_{\star} = b(Z)'\rho_{\star}$, then by the triangle inequality and Lemma \ref{lem:rho_star_aux3},
\begin{equation*}
||\hat\alpha_{L} - \alpha_{0}||^{2} \leq 2||\hat\alpha_{L} - \alpha_{\star}||^{2} + 2||\alpha_{\star} - \alpha_{0}||^{2} \leq 2||\hat\alpha_{L} - \alpha_{\star}||^{2} + 2C\bar{s}\varepsilon_{n}^{2}.
\end{equation*}
By H\"{o}lder's inequality and \eqref{eq:rho_bound},
\begin{equation*}
||\hat\alpha_{L} - \alpha_{\star}||^{2} = (\hat\rho_{L} - \rho_{\star})'B(\hat\rho_{L} - \rho_{\star}) \leq ||B||_{\infty}|\hat\rho_{L} - \rho_{\star}|^{2}_{1} \leq O_p(\bar{s}^{2}\lambda_{n}^{2}).
\end{equation*}
The conclusion comes from the fact that $\bar{s}^{2}\lambda_{n}^{2} > \bar{s}^{2}\varepsilon_{n}^{2} \geq \bar{s}\varepsilon_{n}^{2}$, where the second inequality is due to $\bar{s}^{2}$ growing faster than $\bar{s}$.

\qed


\subsection{PGMM estimator with the diagonal weight matrix}
In this section, we provide useful lemmas to help us prove Theorem \ref{thm:rr_bound_diag_wm}. 

The diagonal weight matrix is constructed in two steps:

\begin{enumerate}
    \item Obtain a preliminary estimate $\tilde{\rho}$ by solving \eqref{eq:pgmm_rr_mat} with $\Omega = I_q$.
    \item Compute the diagonal weight matrix
    \begin{equation*}
    	\hat\Omega = \text{diag}\left(\hat{\sigma}_1^{-2}(\tilde{\rho}), \ldots, \hat{\sigma}_q^{-2}(\tilde{\rho})\right),
    \end{equation*}
\end{enumerate}
where the sample variance of the $j$-th moment is:
\begin{equation*}
\hat{\sigma}_j^2(\tilde{\rho}) = \frac{1}{n}\sum_{i=1}^n \psi_j^2(W_i, \tilde{\rho})
\end{equation*}

The population counterpart is:
\begin{equation*}
	\Omega = \text{diag}\left(\sigma_1^{-2}(\rho_{\star}), \ldots, \sigma_q^{-2}(\rho_{\star})\right)
\end{equation*}
where $\sigma_j^2(\rho_\star) \equiv \E[\psi_j^2(W, \rho_\star)]$ is the population variance and $\rho_{\star}$ is a population PGMM minimizer for the sparse approximation defined in \eqref{eq:rho_star_problem}.

Next, we make an additional assumption to ensure that population moments are non-degenerate and have non-zero variance.
\begin{assumption}[Non-degenerate Population Moments]\label{ass:nondegen}
	There exists $\underline{\sigma}^2 > 0$ such that
	\begin{equation}
		\min_{1 \leq j \leq q} \sigma_j^2(\rho_{\star}) \geq \underline{\sigma}^2.
	\end{equation}
\end{assumption}


\begin{lemma}\label{lem:bounded_psi}
Under Assumptions~\ref{ass:boundedness} and~\ref{ass:m_bound}, there exists a constant $C_\psi > 0$ such that for any $\rho$ with $|\rho|_1 \leq A$ for some constant $A > 0$, and with probability approaching one,
\begin{equation*}
\max_{1 \leq j \leq q} |\psi_j(W, \rho)| \leq C_\psi,
\end{equation*}
where $C_\psi = C_m + C_d C_b A$.
\end{lemma}

\begin{proof}
By definition of the orthogonal moment
\begin{equation*}
\psi_j(W, \rho) = m(W, d_j) - d_j(X) b(Z)'\rho
\end{equation*}
Taking absolute values and applying the triangle inequality
\begin{equation*}
|\psi_j(W, \rho)| \leq |m(W, d_j)| + |d_j(X)| \cdot |b(Z)'\rho|
\end{equation*}
For the second term, apply H\"older's inequality
\begin{equation*}
|b(Z)'\rho| \leq ||b(Z)||_{\infty} |\rho|_1 \leq C_b A,
\end{equation*}
where we used Assumption~\ref{ass:boundedness} and $|\rho|_1 \leq A$. Therefore,
\begin{equation*}
|\psi_j(W, \rho)| \leq C_m + C_d C_b A
\end{equation*}
Taking the maximum over $j$, with probability approaching one, we get
\begin{equation*}
\max_j |\psi_j(W, \rho)| \leq C_m + C_d C_b A =: C_\psi.
\end{equation*}

\end{proof}


\begin{lemma}\label{lem:var_conv}
Under Assumptions~\ref{ass:boundedness} and \ref{ass:m_bound}, for any fixed $\rho$ with $|\rho|_1 \leq A$:
\begin{equation}
\max_{1 \leq j \leq q} \left|\hat{\sigma}_j^2(\rho) - \sigma_j^2(\rho)\right| = O_p\left(\sqrt{\frac{\log q}{n}}\right)
\end{equation}
\end{lemma}

\begin{proof}
The proof is similar to the proof of Lemma \ref{lem:g_rate}. Let $T_{ij}(\rho) = \psi_j^2(W_i, \rho) - \sigma_j^2(\rho)$. Note that $\E[T_{ij}(\rho)] = 0$ and $\{T_{ij}(\rho)\}_{i=1}^n$ are i.i.d.\ for fixed $j$ and $\rho$. The sample variance satisfies
\begin{equation*}
\hat{\sigma}_j^2(\rho) - \sigma_j^2(\rho) = \frac{1}{n}\sum_{i=1}^n T_{ij}(\rho).
\end{equation*}
By Lemma~\ref{lem:bounded_psi}, $|\psi_j(W, \rho)| \leq C_\psi$ with probability approaching one. Therefore,
\begin{equation*}
|T_{ij}(\rho)| \leq |\psi_j^2(W_i, \rho)| + \sigma_j^2(\rho) \leq C_\psi^2 + C_\psi^2 = 2C_\psi^2.
\end{equation*}
Since $T_{ij}(\rho)$ is bounded, it is sub-Gaussian with parameter $K = 2C_\psi^2/\log 2$. By Hoeffding's inequality for bounded random variables, for any $t > 0$ there is a constant $c$ such that,
\begin{equation*}
\P\left(\left|\frac{1}{n}\sum_{i=1}^n T_{ij}(\rho)\right| \geq t\right) \leq 2\exp\left(-\frac{cnt^2}{K^2}\right).
\end{equation*}
Applying the union bound gives, 
\begin{equation*}
\P\left(\max_{1 \leq j \leq q}\left|\hat{\sigma}_j^2(\rho) - \sigma_j^2(\rho)\right| \geq t\right) \leq 2q\exp\left(-\frac{cnt^2}{K^2}\right).
\end{equation*}
Set $t = C\sqrt{\frac{\log q}{n}}$ for a sufficiently large constant $C > 0$. Then,
\begin{equation*}
2q\exp\left(-\frac{cnt^2}{K^2}\right) = 2q\exp\left(-\frac{cC^2\log q}{K^2}\right) = 2\exp\left(\log q\left[1 - \frac{cC^2}{K^2}\right]\right) \rightarrow 0
\end{equation*}
for any $C > K\sqrt{1/c}$. Therefore, 
\begin{equation*}
	\P\left(\max_{1 \leq j \leq q}\left|\hat{\sigma}_j^2(\rho) - \sigma_j^2(\rho)\right| \geq C\sqrt{\frac{\log q}{n}}\right) \rightarrow 0,
\end{equation*}
which completes the proof.
\end{proof}


\begin{lemma}\label{lem:lipschitz_var}
Under Assumptions~\ref{ass:boundedness} and~\ref{ass:m_bound}, for any $\rho, \rho' \in \R^p$ with $|\rho|_1, |\rho'|_1 \leq A$:
\begin{equation}
\max_{1 \leq j \leq q} |\sigma_j^2(\rho) - \sigma_j^2(\rho')| \leq L |\rho - \rho'|_1,
\end{equation}
where $L = 2C_\psi C_d C_b$ is a Lipschitz constant.
\end{lemma}

\begin{proof}
Since $\psi_j(W, \rho) = m_j(W) - g_j(W)'\rho$ is linear in $\rho$, the population variance is
\begin{align*}
\sigma_j^2(\rho) 	& = \E[\psi_j^2(W, \rho)] = \E\left[(m_j(W) - g_j(W)'\rho)^2\right] \\
					& = \E[m_j^2(W)] - 2\E[m_j(W) g_j(W)]'\rho + \rho'\E[g_j(W) g_j(W)']\rho.
\end{align*}
The gradient with respect to $\rho$ is
\begin{align*}
	\nabla_\rho \sigma_j^2(\rho) 	& = -2\E[m_j(W) g_j(W)] + 2\E[g_j(W) g_j(W)']\rho \\
									& = -2\E[(m_j(W) - g_j(W)'\rho)g_j(W)] \\
									& = -2\E[\psi_j(W, \rho) g_j(W)].
\end{align*}

By the multivariate mean value theorem, there exists $\bar{\rho}$ on the line segment between $\rho$ and $\rho'$ such that:
\begin{equation*}
\sigma_j^2(\rho) - \sigma_j^2(\rho') = \nabla_\rho \sigma_j^2(\bar{\rho}_j)'(\rho - \rho').
\end{equation*}
Therefore, by H\"older's inequality:
\begin{equation*}
|\sigma_j^2(\rho) - \sigma_j^2(\rho')| \leq ||\nabla_\rho \sigma_j^2(\bar{\rho})||_\infty|\rho - \rho'|_1.
\end{equation*}
The $k$-th component of the gradient is
\begin{equation}
\left[\nabla_\rho \sigma_j^2(\bar{\rho})\right]_k = -2\E[\psi_j(W, \bar{\rho}) d_j(X) b_k(Z)].
\end{equation}
Applying the Cauchy-Schwarz inequality gives
\begin{equation}
\left|\E[\psi_j(W, \bar{\rho}_j) d_j(X) b_k(Z)]\right| \leq \sqrt{\E[\psi_j^2(W, \bar{\rho}_j)]} \cdot \sqrt{\E[d_j^2(X) \cdot b_k^2(Z)]}.
\end{equation}
By Lemma~\ref{lem:bounded_psi}, $\sqrt{\E[\psi_j^2(W, \bar{\rho}_j)]} \leq C_\psi$. By Assumption~\ref{ass:boundedness}, $\sqrt{\E[d_j^2(X) b_k^2(Z)]} \leq C_d C_b$. Therefore,
\begin{equation}
||\nabla_\rho \sigma_j^2(\bar{\rho})||_\infty \leq 2 C_\psi C_d C_b.
\end{equation}
Taking the maximum over $j$ and setting $L = 2C_\psi C_d C_b$ gives the result
\begin{equation}
\max_j |\sigma_j^2(\rho) - \sigma_j^2(\rho')| \leq L |\rho - \rho'|_1.
\end{equation}

\end{proof}


\begin{lemma}\label{lem:var_lower_pop}
Suppose Assumptions~\ref{ass:boundedness}--\ref{ass:m_bound} and \ref{ass:nondegen} hold, and let $\bar{s}\lambda_n \rightarrow 0$. Then with probability approaching one,
\begin{equation*}
\min_{1 \leq j \leq q} \sigma_j^2(\tilde{\rho}) \geq \frac{3\underline{\sigma}^2}{4}.
\end{equation*}
\end{lemma}

\begin{proof}
For each $j$, by the reverse triangle inequality and Lemma~\ref{lem:lipschitz_var},
\begin{equation*}
\sigma_j^2(\tilde{\rho}) \geq \sigma_j^2(\rho_{\star}) - L|\tilde{\rho} - \rho_{\star}|_1.
\end{equation*}
Taking the minimum on both sides and applying Assumption \ref{ass:nondegen} gives
\begin{equation}\label{eq:min_aux1}
\min_{1 \leq j \leq q}\sigma_j^2(\tilde{\rho}) \geq \min_{1 \leq j \leq q}\sigma_j^2(\rho_{\star}) - L|\tilde{\rho} - \rho_{\star}|_1 \geq \underline{\sigma}^2 - L|\tilde{\rho} - \rho_{\star}|_1.
\end{equation}

With the identity weight matrix, Assumption \ref{ass:weight_matrix} is trivially satisfied, hence, by Theorem \ref{thm:rr_bound}, $|\tilde\rho - \rho_{\star}|_1 = O_p(\bar{s}\lambda_n)$. Moreover, $|\tilde\rho - \rho_{\star}|_1 = o_p(1)$ since $\bar{s}\lambda_n \rightarrow 0$. Thus, for any $\epsilon > 0$, with probability approaching one,
\begin{equation}\label{eq:min_aux2}
|\tilde{\rho} - \rho_{\star}|_1 \leq \frac{\epsilon}{L}.
\end{equation}
Substituting \eqref{eq:min_aux2} into \eqref{eq:min_aux1} and taking $\epsilon = \underline{\sigma}^2/4$ gives
\begin{equation*}
\min_{1 \leq j \leq q}\sigma_j^2(\tilde{\rho}) \geq \underline{\sigma}^2 - \frac{\underline{\sigma}^2}{4} = \frac{3\underline{\sigma}^2}{4}
\end{equation*}
with probability approaching one.
\end{proof}


\begin{lemma}\label{lem:var_lower_sample}
Suppose Assumptions~\ref{ass:boundedness}--\ref{ass:m_bound} and \ref{ass:nondegen} hold, and let $\bar{s}\lambda_n \rightarrow 0$. Then with probability approaching one,
\begin{equation*}
\min_{1 \leq j \leq q} \hat{\sigma}_j^2(\tilde{\rho}) \geq \frac{\underline{\sigma}^2}{2}.
\end{equation*}
\end{lemma}

\begin{proof}
By the reverse triangle inequality and Lemma~\ref{lem:var_conv},
\begin{align*}
	\min_{1 \leq j \leq q} \hat{\sigma}_j^2(\tilde{\rho}) & \geq \min_{1 \leq j \leq q} \sigma_j^2(\tilde{\rho}) - \max_{1 \leq j \leq q}|\hat{\sigma}_j^2(\tilde{\rho}) - \sigma_j^2(\tilde{\rho})| \\
	& \geq \min_{1 \leq j \leq q} \sigma_j^2(\tilde{\rho}) - O_p(\bar{s}\lambda_n).
\end{align*}
Combining the result in Lemma~\ref{lem:var_lower_pop} with $\bar{s}\lambda_n \rightarrow 0$ gives
\begin{equation*}
\min_{1 \leq j \leq q} \hat{\sigma}_j^2(\tilde{\rho}) \geq \frac{3\underline{\sigma}^2}{4} - o_p(1) \geq \frac{\underline{\sigma}^2}{2}.
\end{equation*}
with probability approaching one.
\end{proof}


\begin{lemma}\label{lem:wm_rate}
Suppose Assumptions~\ref{ass:boundedness}--\ref{ass:m_bound} and \ref{ass:nondegen} hold. Define the population weight matrix
\begin{equation}
\Omega^{d} = \text{diag}\left(\sigma_1^{-2}(\rho_{\star}), \ldots, \sigma_q^{-2}(\rho_{\star})\right),
\end{equation}
and the sample weight matrix
\begin{equation}
\hat{\Omega}^{d} = \text{diag}\left(\hat{\sigma}_1^{-2}(\tilde{\rho}), \ldots, \hat{\sigma}_q^{-2}(\tilde{\rho}),\right)
\end{equation}
where $\tilde\rho$ is the preliminary PGMM estimator. Let $\bar{s}\lambda_n \rightarrow 0$, then
\begin{enumerate}
    \item $||\Omega^{d}||_{\ell_{\infty}} \leq \underline{\sigma}^{-2} < \infty$
    \item $||\hat\Omega^{d} - \Omega^{d}||_{\ell_{\infty}} = O_p(\bar{s}\lambda_n)$
\end{enumerate}
\end{lemma}

\begin{proof}
The first part of the proof is straightforward. By Assumption \ref{ass:nondegen}, $\sigma_j^2(\rho_{\star}) \geq \underline{\sigma}^2 > 0$ for all $j$. Therefore,
\begin{equation*}
||\Omega^{d}||_{\ell_{\infty}} = \max_{1 \leq j \leq q} \sigma_j^{-2}(\rho_{\star}) \leq \underline{\sigma}^{-2} < \infty.
\end{equation*}

The second part is more involved. Since both $\hat{\Omega}^{d}$ and $\Omega^{d}$ are diagonal,
\begin{equation*}
||\hat\Omega^{d} - \Omega^{d}||_{\ell_{\infty}} = \max_{1 \leq j \leq q} |\hat{\sigma}_j^{-2}(\tilde{\rho}) - \sigma_j^{-2}(\rho_\star)|.
\end{equation*}
For each $j$, we have the following decomposition
\begin{equation*}
\hat{\sigma}_j^{-2}(\tilde{\rho}) - \sigma_j^{-2}(\rho_\star) = \underbrace{\left[\hat{\sigma}_j^{-2}(\tilde{\rho}) - \sigma_j^{-2}(\tilde{\rho})\right]}_{\equiv \Delta_j^{(I)}} + \underbrace{\left[\sigma_j^{-2}(\tilde{\rho}) - \sigma_j^{-2}(\rho_\star)\right]}_{\equiv \Delta_j^{(II)}},	
\end{equation*}
where $\Delta_j^{(I)}$ is the sampling error and $\Delta_j^{(II)}$ is the estimation error.

We proceed by applying the mean value theorem to each term. Define $h:(0,\infty) \to (0,\infty)$ by $h(x) = 1/x$, so $h'(x) = -1/x^2$. We start with $\Delta_j^{(I)}$. There exists $\xi_j$ strictly between $\hat{\sigma}_j^2(\tilde{\rho})$ and $\sigma_j^2(\tilde{\rho})$ such that
\begin{equation*}
\Delta_j^{(I)} = h'(\xi_j) \left(\hat{\sigma}_j^2(\tilde{\rho}) - \sigma_j^2(\tilde{\rho})\right) = -\frac{1}{\xi_j^2} \left(\hat{\sigma}_j^2(\tilde{\rho}) - \sigma_j^2(\tilde{\rho})\right).
\end{equation*}
Since $\xi_j$ lies between $\hat{\sigma}_j^2(\tilde{\rho})$ and $\sigma_j^2(\tilde{\rho})$,
\begin{equation*}
\xi_j \geq \min\left\{\hat{\sigma}_j^2(\tilde{\rho}), \sigma_j^2(\tilde{\rho})\right\}.
\end{equation*}
Note that with probability approaching one, $\sigma_j^2(\tilde{\rho}) \geq 3\underline{\sigma}^2/4$ by Lemma~\ref{lem:var_lower_pop}, and $\hat{\sigma}_j^2(\tilde{\rho}) \geq \underline{\sigma}^2/2$ by Lemma~\ref{lem:var_lower_sample}. Therefore, with probability approaching one, $\xi_j \geq \underline{\sigma}^2/2$, which implies $1/\xi_j^2 \leq 4/\underline{\sigma}^4$. 

Now let's move to $\Delta_j^{(II)}$. There exists $\eta_j$ strictly between $\sigma_j^2(\tilde{\rho})$ and $\sigma_j^2(\rho_\star)$ such that
\begin{equation*}
\Delta_j^{(II)} = h'(\eta_j) \left(\sigma_j^2(\tilde{\rho}) - \sigma_j^2(\rho_\star)\right) = -\frac{1}{\eta_j^2} \left(\sigma_j^2(\tilde{\rho}) - \sigma_j^2(\rho_\star)\right).
\end{equation*}
Since $\eta_j$ lies between $\sigma_j^2(\tilde{\rho})$ and $\sigma_j^2(\rho_L)$,
\begin{equation*}
\eta_j \geq \min\left\{\sigma_j^2(\tilde{\rho}), \sigma_j^2(\rho_\star)\right\} \geq \min\left\{\frac{3\underline{\sigma}^2}{4}, \underline{\sigma}^2\right\} = \frac{3\underline{\sigma}^2}{4}
\end{equation*}
with probability approaching one. Thus, $1/\eta^2_j \leq 16/9\underline{\sigma}^4$.

Next, we bound the variance differences. By Lemma~\ref{lem:var_conv},
\begin{equation*}
R_n^{(I)} \equiv \max_{1 \leq j \leq q} |\hat{\sigma}_j^2(\tilde{\rho}) - \sigma_j^2(\tilde{\rho})| = O_p\left(\sqrt{\frac{\log q}{n}}\right).
\end{equation*}
By Lemma~\ref{lem:lipschitz_var} and Theorem \ref{thm:rr_bound}:
\begin{equation*}
R_n^{(II)} \equiv \max_{1 \leq j \leq q} |\sigma_j^2(\tilde{\rho}) - \sigma_j^2(\rho_\star)| \leq L|\tilde{\rho} - \rho_\star|_1 = O_p(\bar{s}\lambda_n).
\end{equation*}
Taking absolute values and the maximum over $j$
\begin{align*}
\max_{1 \leq j \leq q} |\hat{\sigma}_j^{-2}(\tilde{\rho}) - \sigma_j^{-2}(\rho_\star)| &\leq \max_{1 \leq j \leq q} |\Delta_j^{(I)}| + \max_{1 \leq j \leq q} |\Delta_j^{(II)}| \\
&\leq \max_{1 \leq j \leq q} \frac{1}{\xi_j^2} \cdot R_n^{(I)} + \max_{1 \leq j \leq q} \frac{1}{\eta_j^2} \cdot R_n^{(II)} \\
&\leq \frac{4}{\underline{\sigma}^4} \cdot O_p\left(\sqrt{\frac{\log q}{n}}\right) + \frac{16}{9\underline{\sigma}^4} \cdot O_p(\bar{s}\lambda_n) \\
& = O_p(\bar{s}\lambda_n).
\end{align*}

\end{proof}


\subsubsection{Proof of Theorem \ref{thm:rr_bound_diag_wm}}

The proof follows closely the proof of Theorem \ref{thm:rr_bound}, but differs slightly in a couple places where the weight matrix enters. We will only focus on these difference for brevity as the main logic and results remain  unchanged.

First, recall the third line of Equation \eqref{eq:aux_rho},
\begin{equation*}
	||(\hat\Omega^{d}_{q})^{1/2}\hat{G}'(\hat\rho_L - \rho_\star)||^{2} + 2\lambda_n|\hat\rho_L|_1 \leq 2||\hat{G}'\hat\Omega^{d}_{q}(\hat{M} - \hat{G}\rho_\star)||_{\infty}|\hat\rho_L - \rho_\star|_{1} + 2\lambda_n|\rho_\star|_1.
\end{equation*}
Instead of relying on Lemma \ref{lem:aux1} to bound $||\hat{G}'\hat\Omega^{d}_{q}(\hat{M} - \hat{G}\rho_\star)||_{\infty}$, we will leverage some properties of the diagonal weight matrix. By the triangle and H{\"o}lder's  inequalities,
\begin{align*}
	||\hat{G}'\hat\Omega^{d}_{q}(\hat{M} - \hat{G}\rho_\star)||_{\infty} & \leq ||\hat{G}'\hat\Omega^{d}_{q}\hat{M}||_{\infty} + ||\hat{G}'\hat\Omega^{d}_{q}\hat{G}\rho_\star||_{\infty} \\
	& \leq \underbrace{||\hat{G}'\hat\Omega^{d}_{q}\hat{M}||_{\infty}}_{=:A} + \underbrace{||\hat{G}'\hat\Omega^{d}_{q}\hat{G}||_{\infty}}_{=:B}|\rho_\star|_1.
\end{align*}
Using \eqref{eq:matrix_ineq}, we can bound both terms as follows,
\begin{align*}
	A & \leq ||\hat{G}||_{\infty} ||\hat\Omega^d||_{\ell_\infty} ||\hat{M}||_{\infty} = ||\hat{G}||_{\infty}\left[\frac{1}{\min_{1 \leq j \leq q}\hat\sigma^{2}_{j}(\tilde\rho)}\right] ||\hat{M}||_{\infty} \\
	B & \leq ||\hat{G}||^{2}_{\infty} ||\hat\Omega^d||_{\ell_\infty} = ||\hat{G}||^{2}_{\infty} \left[\frac{1}{\min_{1 \leq j \leq q}\hat\sigma^{2}_{j}(\tilde\rho)}\right].
\end{align*}
Next, note that by Assumptions \ref{ass:boundedness}, \ref{ass:M_conv} and \ref{ass:m_bound},
\begin{align*}
	||\hat{G}||_{\infty} & \leq ||\hat{G} - G||_{\infty} + ||G||_{\infty} = O_p(\varepsilon^{G}_n) \\ 
	||\hat{M}||_{\infty} & \leq ||\hat{M} - M||_{\infty} + ||M||_{\infty} = O_p(\varepsilon^{M}_n).
\end{align*}
Together with Lemma \ref{lem:var_lower_sample}, we get
\begin{align*}
	A & = O_p(\varepsilon^{G}_n) O(1) O_p(\varepsilon^{M}_n) = O_p(\varepsilon^{2}_{n}) \\
	B & = O_p((\varepsilon^{G}_n)^2) O(1) = O_p((\varepsilon^{G}_n)^2).
\end{align*}
As a result,
\begin{equation*}
	||\hat{G}'\hat\Omega^{d}_{q}(\hat{M} - \hat{G}\rho_\star)||_{\infty} = O_p(\varepsilon^{2}_{n}).
\end{equation*}
Since $\varepsilon_n = o(\lambda_n)$, the final line of Equation \eqref{eq:aux_rho} remains unchanged, i.e.,
\begin{equation*}
	||(\hat\Omega^{d}_{q})^{1/2}\hat{G}'(\hat\rho_L - \rho_\star)||^{2} + 2\lambda_n|\hat\rho_L|_1 \leq 2o_p(\lambda_n)|\hat\rho_L - \rho_\star|_1 + 2\lambda_n|\rho_\star|_1.
\end{equation*}

The second spot where the diagonal matrix enters is the restricted eigenvalue condition. By Lemma \ref{lem:wm_rate}, $||\hat\Omega^{d} - \Omega^{d}||_{\ell_\infty} = O_p(\bar{s}\lambda_n)$. Using this rate in the proof of Lemma \ref{lem:re_bound} gives
\begin{equation*}
	\hat\phi^{2}(s_\star) \geq \phi^{2}(s_\star) - O_p(\bar{s}^{2}\lambda_n).
\end{equation*}
Thus, Equation \eqref{eq:rho_bound_phi} becomes
\begin{equation*}
|\hat\rho_{L} - \rho_{\star}|_{1} \leq \frac{12\lambda_{n}s_{\star}}{\phi^{2}(s_{\star}) - O_p(\bar{s}^{2}\lambda_n)}.
\end{equation*}
Since $\bar{s}^{2}\lambda_n = o(1)$ and $\phi^{2}(s_{\star})$ is bounded away from zero by Assumption \ref{ass:re_pgmm}, with probability approaching one,
\begin{equation*}
|\hat\rho_{L} - \rho_{\star}|_{1} = O_p(\bar{s}\lambda_n).
\end{equation*}
The rest of the proof is unchanged.
\qed


\subsection{PGMM for nonlinear functionals}\label{app:pgmm_nonlin}

In this section, we demonstrate that the results in Theorems~\ref{thm:rr_bound} and \ref{thm:rr_bound_diag_wm} hold when the penalty is larger than $\kappa_n^{\gamma}$ to account for convergence of $\hat\gamma$ in $\hat{M}_{\ell}$. 

\begin{corollary}\label{cor:rr_bound_nonlin}
Suppose Assumptions \ref{ass:weight_matrix}--\ref{ass:boundedness}, \ref{ass:sparse_approx}--\ref{ass:re_pgmm}, and \ref{ass:M_bound_nonlin} hold. Let $\kappa^{\gamma}_{n} = o(\lambda_{n})$ and $\bar{s}\lambda_n = o(1)$, then
\begin{equation*}
	||\hat\alpha_{L} - \alpha_{0}|| = O_p(\kappa_{n}^{\alpha}), \text{ where } \kappa_{n}^{\alpha} = \bar{s}\lambda_{n}.
\end{equation*}
\end{corollary}

\begin{proof}
The proof follows directly from the proof of Theorem \ref{thm:rr_bound}, except we can no longer directly rely on Lemma \ref{lem:aux1} to bound $||\hat{G}'\hat\Omega^{d}_{q}(\hat{M} - \hat{G}\rho_\star)||_{\infty}$. By Lemma \ref{lem:M_bound_nonlin}, $||\hat{M} - M||_{\infty} = O_p(\kappa_{n}^{\gamma})$. Thus, following the arguments in Lemma \ref{lem:aux1}, the bound becomes
\begin{equation*}
	||\hat{G}'\hat\Omega^{d}_{q}(\hat{M} - \hat{G}\rho_\star)||_{\infty} = O_p(\kappa_{n}^{\gamma}).
\end{equation*}
As a result, under $\kappa_n^{\gamma} = o(\lambda_n)$, the final line of Equation \eqref{eq:aux_rho} remains unchanged, and the rest of the proof follows.
\end{proof}

\begin{corollary}\label{cor:rr_bound_diag_wm_nonlin}
Suppose Assumptions \ref{ass:boundedness}, \ref{ass:sparse_approx}--\ref{ass:re_pgmm}, and \ref{ass:M_bound_nonlin} hold. Assume that populations moments are non-degenerate and have non-zero variance. Let $\kappa_n^{\gamma} = o(\lambda_n)$ and $\bar{s}^{2}\lambda_n = o(1)$. Then, under the diagonal weight matrix \eqref{eq:pgmm_diag_wm},
\begin{equation*}
	||\hat\alpha_{L} - \alpha_{0}|| = O_p(\kappa_{n}^{\alpha}), \text{ where } \kappa_{n}^{\alpha} = \bar{s}\lambda_{n}.
\end{equation*}
\end{corollary}

\begin{proof}
We only need to adjust the bound $||\hat{G}'\hat\Omega^{d}_{q}(\hat{M} - \hat{G}\rho_\star)||_{\infty}$ for the rest of the proof of Theorem \ref{thm:rr_bound_diag_wm} to follow. Under Assumption \ref{ass:M_bound_nonlin}, $||\hat{M} - M||_{\infty} = O_p(\kappa_{n}^{\gamma})$ and $||M||_{\infty} = O(1)$. Hence, we get a new bound 
\begin{equation*}
	||\hat{G}'\hat\Omega^{d}_{q}(\hat{M} - \hat{G}\rho_{\star})||_{\infty} = O_p(\varepsilon^{G}_n \kappa_{n}^{\gamma}).
\end{equation*}
Since $\kappa_n^{\gamma} = o(\lambda_n)$, the final line of Equation \eqref{eq:aux_rho} remains unchanged, establishing the result.
\end{proof}

\subsection{Asymptotic properties of the ADMLIV estimator}

\begin{lemma} \label{lem:rho_hat_L_bound}
If Assumptions \ref{ass:weight_matrix}--\ref{ass:M_conv} and \ref{ass:m_bound} are satisfied and $\varepsilon_{n} = o(\lambda_{n})$, then
\begin{equation*}
|\hat\rho_{L}|_{1} = O_p(1).
\end{equation*}
\end{lemma}

\begin{proof}
Recall Equation \eqref{eq:aux_rho} from the proof of Theorem \ref{thm:rr_bound} which implies
\begin{equation*}
2\lambda_{n}|\hat{\rho}_{L}|_{1} \leq 2o_p(\lambda_{n})|\hat{\rho}_{L} - \rho_{\star}|_{1} + 2\lambda_{n}|\rho_{\star}|_{1}. 
\end{equation*}
Dividing both sides by $2\lambda_{n}$ and applying the triangle inequality gives
\begin{equation*}
|\hat{\rho}_{L}|_{1} \leq o_p(1)|\hat{\rho}_{L} - \rho_{\star}|_{1} + |\rho_{\star}|_{1} \leq |\rho_{\star}|_{1} + o_p(1)(|\hat\rho_{L}|_{1} + |\rho_{\star}|_{1}),
\end{equation*}
which implies that with probability approaching one,
\begin{equation*}
|\hat{\rho}_{L}|_{1} \leq |\rho_{\star}|_{1} + \frac{1}{2}(|\hat\rho_{L}|_{1} + |\rho_{\star}|_{1}).
\end{equation*} 
Subtracting $|\hat{\rho}_{L}|_{1}/2$ from both sides and multiplying by $2$ gives with probability approaching one
\begin{equation*}
|\hat{\rho}_{L}|_{1} \leq 3|\rho_{\star}|_{1} = O(1).
\end{equation*}
\end{proof}

\subsubsection{Proof of Theorem \ref{thm:asy_norm_lin}}

We prove the first conclusion by verifying the conditions of Lemma 15 of \cite{chernozhukov2022locally}. Let $g(w,\,\gamma,\,\alpha,\,\theta)$ and $\phi(w,\,\gamma,\,\alpha,\,\theta)$ in \cite{chernozhukov2022locally} be $m(w,\,\gamma) - \theta$ and $\alpha(z)[y - \gamma(x)]$ here, respectively. First, $\E[\psi(W_{i},\,\gamma_{0},\,\alpha_{0},\,\theta_{0})^{2}] < \infty$ follows from Assumption \ref{ass:asy_primitives}. Moreover, note that by Assumptions \ref{ass:asy_primitives} and \ref{ass:mse_consist_gamma}, Theorem \ref{thm:rr_bound}, and the law of iterated expectations,
\begin{align*}
\int[\phi(w,\,\hat\gamma_{\ell},\,\alpha_{0}) - \phi(w,\,\gamma_{0},\,\alpha_{0})]^{2}F_{0}(dw) & = \int \alpha^{2}_{0}(z)[\hat\gamma_{\ell}(x) - \gamma_{0}(x)]^{2}F_{0}(dw) \\
& \leq C||T(\hat\gamma_{\ell} - \gamma_{0})||^{2} \xrightarrow{p} 0 \\
\int[\phi(w,\,\gamma_{0},\,\hat\alpha_{\ell}) - \phi(w,\,\gamma_{0},\,\alpha_{0})]^{2}F_{0}(dw) & = \int [\hat\alpha_{\ell}(z) - \alpha_{0}(z)]^{2}[y - \gamma_{0}(x)]^{2}F_{0}(dw) \\
& = \int [\hat\alpha_{\ell}(z) - \alpha_{0}(z)]^{2}\E[[y - \gamma_{0}(x)]^{2}|z]F_{0}(dz) \\ 
& \leq C||\hat\alpha_{\ell} - \alpha_{0}||^{2} \xrightarrow{p} 0.
\end{align*}
Also, it follows from Assumption \ref{ass:mse_consist_gamma} that
\begin{equation*}
\int[m(w,\,\hat\gamma_{\ell}) - m(w,\,\gamma_{0})]^{2}F_{0}(dw) \xrightarrow{p} 0.
\end{equation*} 
Thus, all the conditions of Assumption 1 of \cite{chernozhukov2022locally} are satisfied.

Next, for each $\ell$ let
\begin{align*}
\hat\Delta_{\ell}(w) & = \phi(w,\,\hat\gamma_{\ell},\,\hat\alpha_{\ell}) - \phi(w,\,\gamma_{0},\,\hat\alpha_{\ell}) - \phi(w,\,\hat\gamma_{\ell},\,\alpha_{0}) + \phi(w,\,\gamma_{0},\,\alpha_{0}) \\ & = [\hat\alpha_{\ell}(z) - \alpha_{0}(z)][\hat\gamma_{\ell}(x) - \gamma_{0}(x)].
\end{align*}
Since $\alpha_{0}$ is bounded by Assumption \ref{ass:asy_primitives} and $\sup_{z}|\hat\alpha_{\ell}(z)| = O_p(1)$ by Lemma \ref{lem:rho_hat_L_bound}, 
\begin{align*}
	\int\hat\Delta_{\ell}^{2}(w)F_{0}(dw) & = \int[\hat\alpha_{\ell}(z) - \alpha_{0}(z)]^{2}[\hat\gamma_{\ell}(x) - \gamma_{0}(x)]^{2}F_{0}(dw) \\
	& \leq O_p(1)\int[\hat\gamma_{\ell}(x) - \gamma_{0}(x)]^{2}F_{0}(dw) \xrightarrow{p} 0,
\end{align*}
where the conclusion follows from Assumption \ref{ass:mse_consist_gamma}. Furthermore, by Cauchy-Schwarz inequality and Assumption \ref{ass:conv_rates_lin},
\begin{align*}
\left|\sqrt{n}\int\hat\Delta_{\ell}(w)F_{0}(dw)\right| & =	\sqrt{n}\left|\int[\hat\alpha_{\ell}(z) - \alpha_{0}(z)][\hat\gamma_{\ell}(x) - \gamma_{0}(x)]F_{0}(dw)\right| \\
& = \sqrt{n}\left|\int[\hat\alpha_{\ell}(z) - \alpha_{0}(z)]\E[\hat\gamma_{\ell}(x) - \gamma_{0}(x)|z]F_{0}(dz)\right| \\
& \leq \sqrt{n}||\hat{\alpha}_{\ell} - \alpha_{0}||\;||T(\hat\gamma_{\ell} - \gamma_{0})|| = O_p(n^{1/2}\kappa_{n}^{\alpha}\kappa_{n}^{\gamma}) \xrightarrow{p} 0,
\end{align*}
which renders Assumption 2(iii) of \cite{chernozhukov2022locally} satisfied.

Also, by construction,
\begin{align*}
\int \hat\alpha_{\ell}(z)[y - \gamma_{0}(x)]F_{0}(dw) = \E[\hat\alpha_{\ell}(z)\E[y - \gamma_{0}(x)]|z] = 0.
\end{align*}
Combined with $m(w,\,\gamma)$ being affine in $\gamma$ verifies Assumption 3 of \cite{chernozhukov2022locally} is satisfied. As a result, we get the first conclusion.

To get the second conclusion, we need to show that $\hat V$ is a consistent estimator of $V$. This part of the proof is very similar to the proof of Theorem 3 in \cite{chernozhukov2022adml}. We start with
\begin{equation*}
\hat V = \frac{1}{n}\sum_{i=1}^{n}\hat\psi_{i}^{2} = \frac{1}{n}\sum_{i=1}^{n}(\hat\psi_{i} - \psi_{i})^{2} + \frac{2}{n}\sum_{i=1}^{n}(\hat\psi_{i} - \psi_{i})\psi_{i} + \frac{1}{n}\sum_{i=1}^{n}\psi_{i}^{2},
\end{equation*}  
hence, by re-arranging the terms and Cauchy-Schwarz inequality,
\begin{align}\label{eq:var_decomposition}
	\hat V - V & = \frac{1}{n}\sum_{i=1}^{n}(\hat\psi_{i} - \psi_{i})^{2} + \frac{2}{n}\sum_{i=1}^{n}(\hat\psi_{i} - \psi_{i})\psi_{i} \nonumber \\ 
	& \leq \frac{1}{n}\sum_{i=1}^{n}(\hat\psi_{i} - \psi_{i})^{2} + 2\sqrt{\frac{1}{n}\sum_{i=1}^{n}(\hat\psi_{i} - \psi_{i})^{2}}\sqrt{\frac{1}{n}\sum_{i=1}^{n}\psi_{i}^{2}}.
\end{align}
Using the triangle inequality, for $i \in I_{\ell}$,
\begin{equation*}
(\hat\psi_{i} - \psi_{i})^{2} \leq C\sum_{j=1}^{4}R_{ij} = C\sum_{j=1}^{3}R_{ij} + o_p(1),
\end{equation*}
where
\begin{align*}
R_{i1} & = [m(W_{i},\,\hat\gamma_{\ell}) - m(W_{i},\,\gamma_{0})]^{2}, \\
R_{i2} & = \hat\alpha_{\ell}^{2}(Z_{i})[\hat\gamma_{\ell}(X_{i}) - \gamma_{0}(X_{i})]^{2}, \\
R_{i3} & = [\hat\alpha_{\ell}(Z_{i}) - \alpha_{0}(Z_{i})]^{2}[Y_{i} - \gamma_{0}(X_{i})]^{2}, \\
R_{i4} & = (\hat\theta - \theta_{0})^{2}.
\end{align*}
The first conclusion implies $R_{i4} \xrightarrow{p} 0$. Let $I_{-\ell}$ denote observations not in $I_{\ell}$.

By Markov's inequality, for some $\delta > 0$,
\begin{equation*}
\P\left(\frac{1}{n}\sum_{i=1}^{n}(\hat\psi_{i} - \psi_{i})^{2} > \delta\right) \leq \frac{\E\left[\frac{1}{n}\sum_{i=1}^{n}(\hat\psi_{i} - \psi_{i})^{2}\right]}{\delta}.
\end{equation*}
Note that the cross-fitting allows us to write
\begin{align*}
\E\left[\frac{1}{n}\sum_{i=1}^{n}(\hat\psi_{i} - \psi_{i})^{2}\right] & \leq \E\left[\frac{C}{n} \sum_{\ell=1}^{L} \sum_{i\in I_{\ell}}\sum_{j=1}^{3}R_{ij} \right] + o_p(1) \\ & = C\sum_{\ell=1}^{L}\frac{n_{\ell}}{n}\sum_{j=1}^{3}\E[\E[R_{ij}|I_{-\ell}]] + o_p(1).
\end{align*}
Furthermore, by H\"{o}lder's inequality and Assumption \ref{ass:boundedness},
\begin{equation*}
\max_{i \in I_{\ell}}|\hat\alpha_{\ell}(Z_{i})| \leq |\hat\rho_{L}|_{1}\max_{i\in I_{\ell}}||b(Z_{i})||_{\infty} \leq C_{b}|\hat\rho_{L}|_{1}.
\end{equation*}
By Lemma \ref{lem:rho_hat_L_bound},
\begin{equation*}
\max_{i}|\hat\alpha_{\ell}(Z_{i})| =  C_{b}O_p(\bar{A}_n) = O_p(1).
\end{equation*}
Then for $i\in I_{\ell}$ by Assumptions \ref{ass:asy_primitives}, \ref{ass:mse_consist_gamma}, and iterated expectations,
\begin{align*}
\E[R_{i1}|I_{-\ell}] & = \int[m(W_{i},\,\hat\gamma_{\ell}) - m(W_{i},\,\gamma_{0})]^{2}F_{0}(dW) \xrightarrow{p} 0, \\
\E[R_{i2}|I_{-\ell}] & \leq O_p(1)\int[\hat\gamma_{\ell}(X_{i}) - \gamma_{0}(X_{i})]^{2}F_{0}(dX) \xrightarrow{p} 0, \\
\E[R_{i3}|I_{-\ell}] & = \E\left[\E\left[[\hat\alpha_{\ell}(Z_{i}) - \alpha_{0}(Z_{i})]^{2}[Y_{i} - \gamma_{0}(X_{i})]^{2}|Z_{i},\,I_{-\ell}\right]|I_{-\ell}\right] \\
& = \E\left[[\hat\alpha_{\ell}(Z_{i}) - \alpha_{0}(Z_{i})]^{2}\E[[Y_{i} - \gamma_{0}(X_{i})]^{2}|Z_{i}]|I_{-\ell}\right] \\
& \leq C||\hat\alpha_{\ell} - \alpha_{0}||^{2} \xrightarrow{p} 0.
\end{align*}
As a result,
\begin{equation*}
\frac{1}{n}\sum_{i=1}^{n}(\hat\psi_{i} - \psi_{i})^{2} \xrightarrow{p} 0.
\end{equation*}
Furthermore, $\E[\psi_{i}^{2}]<\infty$ by Assumptions \ref{ass:asy_primitives} and \ref{ass:mse_consist_gamma}. Thus, the conclusion follows from \eqref{eq:var_decomposition} and the central limit theorem.

\qed

\subsubsection{Proof of Lemma \ref{lem:M_bound_nonlin}}

The proof is similar to the proof of Lemma 8 of \cite{chernozhukov2022adml}. We start with defining
\begin{align*}
\hat{M}_{\ell} & = (\hat{M}_{\ell 1},\dots,\,\hat{M}_{\ell q})',\; \hat{M}_{\ell j} = \frac{1}{n - n_{\ell}} \sum_{\ell'\neq\ell}\sum_{i \in I_{\ell'}}D(W_{i},\,d_{j},\,\hat\gamma_{\ell,\ell'}), \\
\bar{M}_{\ell}(\gamma) & = (\bar{M}_{\ell 1}(\gamma),\dots,\,\bar{M}_{\ell q}(\gamma))',\; \bar{M}_{\ell j} = \int D(w,\,d_{j},\,\gamma)F_{0}(dw).
\end{align*}
Note that $M = \bar{M}(\gamma_{0})$. Let $\Gamma_{\ell,\ell'} = \{||\hat\gamma_{\ell,\ell'} - \gamma_{0}||\leq\varepsilon\}$, and note that $\P(\Gamma_{\ell,\ell'}) \rightarrow 1$ for each $\ell$ and $\ell'$ by Assumption \ref{ass:M_bound_nonlin}. When $\Gamma_{\ell,\ell'}$ occurs,
\begin{equation*}
\max_{1\leq j \leq q}|D(W,\,d_{j},\,\gamma)| \leq C
\end{equation*}
by Assumption \ref{ass:M_bound_nonlin}. For $i \in I_{\ell'}$ define
\begin{equation*}
T_{ij}(\gamma) = D(W_{i},\,d_{j},\,\gamma) - \bar{M}(\gamma),\; U_{ij}(\gamma) = \frac{1}{n_{\ell'}}\sum_{i \in I_{\ell'}}T_{ij}(\gamma).
\end{equation*}
Note that for any constant $\bar{C}$ and the event $\mathcal A=\{\max_{1 \leq j \leq q}|U_{ij}(\gamma)| \geq \bar{C}\varepsilon_{n}\}$ where $\varepsilon_{n} = \sqrt{\log q/n}$,
\begin{align}\label{eq:aux_prob}
\P(\mathcal A) & = \P(\mathcal A|\Gamma_{\ell,\ell'})\P(\Gamma_{\ell,\ell'}) + \P(\mathcal A|\Gamma^{c}_{\ell,\ell'})[1 - \P(\Gamma_{\ell,\ell'})] \\
& \leq \P\left(\max_{1 \leq j \leq q}|U_{ij}(\hat\gamma_{\ell,\ell'})| \geq \bar{C}\varepsilon_{n}|\Gamma_{\ell,\ell'}\right) + [1 - \P(\Gamma_{\ell,\ell'})]. \notag
\end{align}
Moreover,
\begin{equation*}
\P\left(\max_{1 \leq j \leq q}|U_{ij}(\hat\gamma_{\ell,\ell'})| \geq \bar{C}\varepsilon_{n}|\Gamma_{\ell,\ell'}\right) \leq q\max_{1 \leq j \leq q}\P\left(|U_{ij}(\hat\gamma_{\ell,\ell'})| \geq \bar{C}\varepsilon_{n}|\Gamma_{\ell,\ell'}\right).
\end{equation*}
Note that $\E[T_{ij}(\hat\gamma_{\ell,\ell'})|\hat\gamma_{\ell,\ell'}] = 0$ for $i \in I_{\ell'}$. Furthermore, conditional on $\Gamma_{\ell,\ell'}$, for $i \in I_{\ell'}$,
\begin{equation*}
|T_{ij}(\hat\gamma_{\ell,\ell'})| \leq 2C.
\end{equation*}
Hence, $T_{ij}$ is bounded. Similar to the proof of Lemma \ref{lem:g_rate}, define $K = 2C / \log2 \geq ||T_{ij}||_{\Psi_2}$. By Hoeffding's inequality (see Theorem 2.6.2 in \cite{vershynin2018high}) and the independence of $\{W_{i}\}_{i \in I_{\ell'}}$ and $\hat\gamma_{\ell,\ell'}$, there is a constant $c$ such that
\begin{align*}
q\max_{1 \leq j \leq q}\P\left(|U_{ij}(\hat\gamma_{\ell,\ell'})| \geq \bar{C}\varepsilon_{n}|\Gamma_{\ell,\ell'}\right) & = q\E\left[\left.\max_{1 \leq j \leq q}\P\left(|U_{ij}(\hat\gamma_{\ell,\ell'})| \geq \bar{C}\varepsilon_{n}|\hat\gamma_{\ell,\ell'}\right)\right|\Gamma_{\ell,\ell'}\right] \\
& \leq 2q\E\left[\left.\exp\left(-\frac{c(n_{\ell'}\bar{C}\varepsilon_{n})^{2}}{n_{\ell'}K^{2}}\right)\right|\Gamma_{\ell,\ell'} \right] \\
& \leq 2q\exp\left(-\frac{c n_{\ell'}\bar{C}^{2}\log q}{L n_{\ell'}K^{2}}\right) \\
& \leq 2\exp\left(\log q\left[1 - \frac{c\bar{C}^{2}}{LK^{2}}\right]\right) \rightarrow 0
\end{align*}
for any $\bar{C} > K\sqrt{L/c}$. Let $U_{\ell'}(\gamma) = (U_{\ell'1},\dots,\,U_{\ell'q})'$. Then it follows from \eqref{eq:aux_prob} that for large $\bar{C}$, $\P(|U_{\ell'}(\hat\gamma_{\ell,\ell'})| \geq \bar{C}\varepsilon_{n}) \rightarrow 0$, meaning that $||U_{\ell'}(\hat\gamma_{\ell,\ell'})||_{\infty} = O_p(\varepsilon_{n})$. 

Next, for each $\ell$ by the triangle inequality we have,
\begin{equation*}
	||\hat M_{\ell} - M||_{\infty} \leq ||\hat M_{\ell} - \bar{M}(\hat\gamma_{\ell,\ell'})||_{\infty} + ||\bar{M}(\hat\gamma_{\ell,\ell'}) - M||_{\infty}.
\end{equation*}
Furthermore, $n - n_{\ell} = \sum_{\ell'\neq\ell}n_{\ell'}$ and
\begin{equation*}
||\hat M_{\ell} - \bar{M}(\hat\gamma_{\ell,\ell'})||_{\infty} = \left\|\hat M_{\ell} - \sum_{\ell'\neq\ell}\frac{n_{\ell'}}{n - n_{\ell}}\bar{M}(\hat\gamma_{\ell,\ell'})\right\|_{\infty} \leq \sum_{\ell'\neq\ell}\frac{n_{\ell'}}{n - n_{\ell}}||U_{\ell'}(\hat\gamma_{\ell,\ell'})||_{\infty} = O_p(\varepsilon_{n}).
\end{equation*} 
Also, by Assumption \ref{ass:M_bound_nonlin}(ii) and $\P(\Gamma_{\ell,\ell'}) \rightarrow 1$ for each $\ell$ and $\ell'$,
\begin{equation*}
||\bar{M}(\hat\gamma_{\ell,\ell'}) - M||_{\infty} \leq \left\| \sum_{\ell'\neq\ell}\frac{n_{\ell'}}{n - n_{\ell}}[\bar{M}(\hat\gamma_{\ell,\ell'}) - M]\right\|_{\infty} \leq C\sum_{\ell'\neq\ell}\frac{n_{\ell'}}{n - n_{\ell}}||\hat\gamma_{\ell,\ell'} - \gamma_{0}|| = O_p(\kappa_{n}^{\gamma}).
\end{equation*}
The conclusion follows from $\kappa_{n}^{\gamma}$ being a slower rate than $\varepsilon_{n}$.
\qed

\subsubsection{Proof of Theorem \ref{thm:asy_norm_nonlin}}

The proof is analogous to the proof of Theorem \ref{thm:asy_norm_lin}. We obtain the first conclusion by verifying the conditions of Lemma 15 of \cite{chernozhukov2022locally}. First, it follows from the proof of Theorem \ref{thm:asy_norm_lin} that the conditions of Assumptions 1 and 2 of \cite{chernozhukov2022locally} are satisfied.

Next, by Assumptions \ref{ass:M_frechet} and \ref{ass:conv_rate_gamma_nonlin},
\begin{align*}
\sqrt{n}|\bar{\psi}(w,\,\hat\gamma_{\ell},\,\alpha_{0},\,\theta_{0})| & = \sqrt{n}\left| \int[m(w,\,\hat\gamma_{\ell}) - \theta_{0} + \alpha_{0}(z)[y - \hat\gamma_{\ell}(x)]]F_{0}(dw) \right| \\
& = \sqrt{n}\left| \int[m(w,\,\hat\gamma_{\ell}) - m(w,\,\gamma_{0}) + \alpha_{0}(z)[y - \hat\gamma_{\ell}(x)]]F_{0}(dw) \right| \\
& = \sqrt{n}\left| \int[m(w,\,\hat\gamma_{\ell}) - m(w,\,\gamma_{0}) + \alpha_{0}(z)[\gamma_{0}(x) - \hat\gamma_{\ell}(x)]]F_{0}(dw) \right| \\
& = \sqrt{n}\left| \int[m(w,\,\hat\gamma_{\ell}) - m(w,\,\gamma_{0}) - D(w,\,\gamma_{0},\,\hat\gamma_{\ell}-\gamma_{0})]F_{0}(dw) \right| \\
& \leq C\sqrt{n}||\hat\gamma_{\ell}-\gamma_{0}||^{2} \\ 
& = \sqrt{n}o_{p}((n^{-1/4})^{2}) = o_p(1).
\end{align*}
Moreover, as in the proof of Theorem \ref{thm:asy_norm_lin},
\begin{align*}
\int \hat\alpha_{\ell}(z)[y - \gamma_{0}(x)]F_{0}(dw) = 0.
\end{align*}
Thus, all the conditions of Assumption 3 of \cite{chernozhukov2022locally} are satisfied, which combined with the results above gives us the first conclusion. The second conclusion follows exactly as in the proof of Theorem \ref{thm:asy_norm_lin}.
\qed

\subsubsection{Proofs of Corollaries \ref{cor:asy_norm_lin_wm}--\ref{cor:asy_norm_nonlin_wm}}
These results follow directly from the corresponding main Theorems \ref{thm:asy_norm_lin}--\ref{thm:asy_norm_nonlin} and Corollaries \ref{cor:rr_bound_nonlin}--\ref{cor:rr_bound_diag_wm_nonlin}.

\section{Symmetric Inverse Demand Representation} \label{app:symmetric_inverse}

In this section, we provide a formal derivation of the symmetric inverse demand representation in equation~\eqref{eq:est_eq2}. The derivation synthesizes the nonparametric identification results of \cite{berry2014identification} with the dimensionality reduction techniques of \cite{gandhi_houde2019}.

\subsection{Assumptions}

We maintain the following assumptions throughout.

\begin{assumption}\label{ass:index}
The mean utility of product $j \in \mathcal{J} \backslash \{0\}$ in market $t$ takes the additive form
\begin{equation*}
	\delta_{jt} = x_{jt}^{(1)} + \xi_{jt},
\end{equation*}
where $x_{jt}^{(1)} \in \mathbb{R}$ is an observed characteristic entering with coefficient normalized to one, and $\xi_{jt}$ is an unobserved product-market specific shock.
\end{assumption}

This index restriction, standard in the literature following \cite{berry2014identification}, ensures that $x_{jt}^{(1)}$ enters with a known coefficient. This is essential for using $x_{jt}^{(1)}$ as an instrument and for separating the structural error $\xi_{jt}$ from observables.

\begin{assumption}\label{ass:connected}
Products $\{0, 1, \ldots, J\}$ are connected substitutes in both $-p_t$ and $\delta_t$:
\begin{enumerate}
	\item $\sigma_k(\delta_t, p_t)$ is nonincreasing in $\delta_{jt}$ and in $-p_{jt}$ for all $k \neq j$;
	\item For any nonempty $K \subseteq \{1, \ldots, J\}$, there exist $k \in K$ and $\ell \notin K$ such that $\sigma_\ell(\delta_t, p_t)$ is strictly decreasing in $\delta_{kt}$.
\end{enumerate}
\end{assumption}

The connected substitutes condition, introduced by \cite{berry2013connected}, ensures global invertibility of the demand system. Part (i) requires weak gross substitutability; part (ii) rules out fragmentation of the choice set into disconnected components.

\begin{assumption}\label{ass:linear_utility}
Consumer $i$'s indirect utility for product $j$ in market $t$ takes the form
\begin{equation*}
	u_{ijt} = \delta_{jt} + x_{jt}^{(2)\prime}\lambda_i - \alpha_i p_{jt} + \epsilon_{ijt},
\end{equation*}
where $\delta_{jt} = x_{jt}^{(1)} + \xi_{jt}$ is the mean utility index, $\lambda_i \in \mathbb{R}^{d_{x_2}}$ are random coefficients on nonlinear characteristics $x_{jt}^{(2)}$, $\alpha_i > 0$ is the random price sensitivity, and $\epsilon_{ijt}$ is an idiosyncratic taste shock.
\end{assumption}

Linearity in characteristics is crucial for the symmetry results that follow. Under this specification, consumer preferences depend on characteristic differences rather than levels, which generates the exchangeability property exploited by \cite{gandhi_houde2019}.

\begin{assumption}\label{ass:normalization}
The outside good $j = 0$ has normalized characteristics and utility:
\begin{equation*}
	x_{0t}^{(1)} = 0, \quad x_{0t}^{(2)} = 0, \quad p_{0t} = 0, \quad \xi_{0t} = 0, \quad \text{so that} \quad \delta_{0t} = 0.
\end{equation*}
\end{assumption}

This is the standard normalization in discrete choice models, which pins down the location of utilities.

\subsection{Symmetric Representation}

Under Assumptions~\ref{ass:index}--\ref{ass:connected}, \cite{berry2013connected} establish global invertibility of the demand system.

\begin{lemma}[Berry, Gandhi, and Haile, 2013] \label{lem:invertibility}
Consider any price vector $p$ and any market share vector $s = (s_0, s_1, \ldots, s_J)$ such that $s_j > 0$ for all $j$ and $\sum_{j=1}^{J} s_j < 1$. Under Assumptions~\ref{ass:index}--\ref{ass:connected}, there is at most one vector $\delta = (\delta_1, \ldots, \delta_J)$ such that $\sigma_j(\delta, p) = s_j$ for all $j$.
\end{lemma}

This lemma allows us to write the inverse demand system as
\begin{equation} \label{eq:inverse_demand_app}
	\delta_{jt} = \sigma_j^{-1}(s_t, p_t, x_t^{(2)}), \quad j = 1, \ldots, J.
\end{equation}
The key step is to show that the inverse demand admits a symmetric representation. Define the state vector for product $j$ as
\begin{equation*}
	\omega_{jt} = \{(s_{kt}, \Delta_{jkt})\}_{k \neq j},
\end{equation*}
where $\Delta_{jkt} = (p_{jt} - p_{kt}, x_{jt}^{(2)} - x_{kt}^{(2)})$ and $k$ ranges over $\{0, 1, \ldots, J\} \setminus \{j\}$.

\begin{lemma}[Gandhi and Houde, 2019] \label{lem:symmetry}
Under Assumptions~\ref{ass:index}--\ref{ass:linear_utility}, there exists a function $G: \mathcal{W} \to \mathbb{R}$, symmetric in its arguments (i.e., invariant to permutations of rival products), and a market-specific constant $C_t$, such that
\begin{equation} \label{eq:symmetric_rep}
	\sigma_j^{-1}(s_t, p_t, x_t^{(2)}) = G(\omega_{jt}) + C_t
\end{equation}
for all inside goods $j = 1, \ldots, J$.
\end{lemma}

Note that the market constant $C_t$ in Lemma~\ref{lem:symmetry} is not a free parameter, it is determined by the outside good normalization.

\begin{lemma}\label{lem:constant}
Under Assumptions~\ref{ass:index}--\ref{ass:normalization}, the market constant satisfies
\begin{equation} \label{eq:Ct_pinned}
	C_t = -G(\omega_{0t}),
\end{equation}
where $\omega_{0t} = \{(s_{kt}, \Delta_{0kt})\}_{k=1}^{J}$ is the state vector from the outside good's perspective, with $\Delta_{0kt} = (0 - p_{kt}, 0 - x_{kt}^{(2)}) = (-p_{kt}, -x_{kt}^{(2)})$.
\end{lemma}

\begin{proof}
Apply the symmetric representation~\eqref{eq:symmetric_rep} to the outside good $j = 0$,
\begin{equation*}
	\sigma_0^{-1}(s_t, p_t, x_t^{(2)}) = G(\omega_{0t}) + C_t.
\end{equation*}
By definition, $\sigma_0^{-1}(s_t, p_t, x_t^{(2)}) = \delta_{0t}$. Under Assumption~\ref{ass:normalization}, $\delta_{0t} = 0$. Therefore,
\begin{equation*}
	G(\omega_{0t}) + C_t = 0 \quad \Longrightarrow \quad C_t = -G(\omega_{0t}).
\end{equation*}
\end{proof}

\subsection{Proof of Theorem~\ref{thm:symmetric_inverse}}

Substituting~\eqref{eq:Ct_pinned} into~\eqref{eq:symmetric_rep} gives
\begin{equation} \label{eq:delta_difference}
	\delta_{jt} = \sigma_j^{-1}(s_t, p_t, x_t^{(2)}) = G(\omega_{jt}) - G(\omega_{0t}).
\end{equation}
The inverse demand is thus a difference of the symmetric function $G$ evaluated at two state vectors.

Next, we express $\omega_{0t}$ in terms of $\omega_{jt}$. For any inside good rival $k \neq j$,
\begin{align*}
	\Delta_{0kt} &= (-p_{kt}, -x_{kt}^{(2)}) \\
	&= (p_{jt} - p_{kt} - p_{jt}, x_{jt}^{(2)} - x_{kt}^{(2)} - x_{jt}^{(2)}) \\
	&= \Delta_{jkt} - \Delta_{j0t}.
\end{align*}
Since $\omega_{jt}$ contains both $\Delta_{jkt}$ for all $k \neq j$ and $\Delta_{j0t} = (p_{jt}, x_{jt}^{(2)})$, we can write $\omega_{0t} = h(\omega_{jt})$ for a deterministic function $h$. Therefore, $G(\omega_{0t})$ is a function of $\omega_{jt}$.

Define
\begin{equation} \label{eq:gamma_def}
	\gamma(\omega_{jt}) \equiv \log\left(\frac{s_{jt}}{s_{0t}}\right) - G(\omega_{jt}) + G(\omega_{0t}).
\end{equation}
Since $s_{jt} = 1 - \sum_{k \neq j} s_{kt}$ and $s_{0t} \in \{s_{kt}\}_{k \neq j}$ are both determined by $\omega_{jt}$, and $G(\omega_{0t}) = G(h(\omega_{jt}))$, the function $\gamma$ depends only on $\omega_{jt}$.

From~\eqref{eq:delta_difference} and $\delta_{jt} = x_{jt}^{(1)} + \xi_{jt}$ we get
\begin{equation*}
	x_{jt}^{(1)} + \xi_{jt} = G(\omega_{jt}) - G(\omega_{0t}).
\end{equation*}
Rearranging and using~\eqref{eq:gamma_def} gives
\begin{align*}
	\log\left(\frac{s_{jt}}{s_{0t}}\right) &= \gamma(\omega_{jt}) + G(\omega_{jt}) - G(\omega_{0t}) \\
	&= \gamma(\omega_{jt}) + x_{jt}^{(1)} + \xi_{jt}.
\end{align*}
\qed 

\subsection{Empirical moment basis functions} \label{app:em_bfs}

Here we discuss how to leverage the symmetry of the inverse demand function to construct informative basis functions to approximate the empirical distribution of the characteristic differences. We have a function $\gamma\left(\omega_{jt}\right)$ we need to approximate, where 
\[
\omega_{jt}=\left(\omega'_{j,1,t},\dots,\,\omega'_{j,j-1,t},\,\omega'_{j,j+1,t},\dots,\,\omega'_{j,J,t}\right)'
\]
is a vector representing the ``state'' of product $j$ in market
$t$ (the shares and product characteristic differences with respect
to the rivals in the same market). Given the vector symmetric theory
underlying demand across markets, without loss of generality we can express 
\[
\gamma\left(\omega_{jt}\right) = g\left(F\left(\omega_{jt}\right)\right)
\]
where $F$ is the empirical distribution of the variables in $\omega_{jt}$.

An approximation strategy for $\gamma$ can be structured as following.
For simplicity, write $F_{jt}=F\left(\omega_{jt}\right)$ and let
us approximate the distribution $F_{jt}$ by a finite set of moments
$m_{1}\left(F_{jt}\right),\dots,\,m_{L}\left(F_{jt}\right)$. Then our
approximation to $\gamma$ can be expressed as 
\[
\gamma\left(\omega_{jt}\right)\approx g\left(m_{1}\left(F_{jt}\right),\dots,\,m_{L}\left(F_{jt}\right)\right).
\]
There are two issues we need to resolve to implement this approximation: 
\begin{enumerate}
\item The choice of moments $m_{1}\left(F_{jt}\right),\dots,\,m_{L}\left(F_{jt}\right)$ 
\item The choice of a predictive function $g$ 
\end{enumerate}
Let us first deal with the choice of $m_{l}$, $l=1,\dots,\,L$. Let
us define $M_{jt}\left(\tau\right)$ as the MGF associated with
$F_{jt}$, where $\tau=\left(\tau_{1},\dots,\,\tau_{d_{x_{2}}+1}\right)$
and $d_{x_{2}}$ is the dimension of $\text{\ensuremath{x^{(2)}}}$.
Then define the moment 
\[
m_{p_{1},\dots,\,p_{d_{x_{2}}+1}}^{jt}=\left.\frac{\partial^{p_{1}+\dots p_{d_{x_{2}}+1}}}{\partial t_{1}^{p_{1}}\dots\partial t_{d_{x_{2}}+1}^{p_{d_{x_{2}}+1}}}M_{jt}(\tau)\right|_{\tau=0}
\]
This class of moments is defined by the multi-index $p_{1},\dots p_{d_{x_{2}}+1}$
for $p_{k}\in\mathbb{Z}_{+}$. We can define the set of $n^{th}$
order moments to be 
\[
B_{n}^{jt}=\left\{ m_{p_{1},\dots,\, p_{d_{x_{2}}+1}}^{jt}:\sum_{k=1}^{d_{x_{2}}+1}p_{k}=n\mbox{ and }n\geq2\mbox{ and \ensuremath{p_{1}>0} and }\exists k>1\ s.t.\ p_{k}>0\right\} .
\]

Observe that we restrict shares which are the first dimension of the
state vector $\omega_{jt}$ to never enter with a zero power, e.g.,
each moment has some interaction with shares. In addition, shares must
interact with at least one dimension of differentiation. Then the set
of moments entering the $n^{th}$ order approximation for each $t$
is 
\[
\bigcup_{i=2}^{n}B_{i}^{jt}
\]

The choice of $g$ can be determine by any functional form that allows
for a flexible approximation from the predictors $m_{1}\left(F_{jt}\right),\dots,\,m_{L}\left(F_{jt}\right)$, such as polynomials, B-splines, wavelets, etc.


\section{Own-price elasticity functional}\label{app:price_el}

Recall the semiparametric inverse demand system from equation \eqref{eq:demand_main_eq}:
\begin{equation*}
    \log\left(\frac{s_{jt}}{s_{0t}}\right) = x_{jt}^{(1)} + \gamma(\omega_{jt}) + \xi_{jt}, \quad j = 1, \ldots, J,
\end{equation*}
where $s_{jt}$ is the market share of product $j$, $s_{0t}$ is the outside good share, $x_{jt}^{(1)}$ are exogenous characteristics entering linearly, $\gamma(\cdot)$ is an unknown function, $\omega_{jt}$ contains endogenous variables (shares, prices, and nonlinear characteristic differences), and $\xi_{jt}$ is the structural error.

Define the inverse demand residuals as
\begin{equation*}
    \Upsilon_{jt}(s_t, p_t; \gamma) \equiv \log(s_{jt}/s_{0t}) - x_{jt}^{(1)} - \gamma(\omega_{jt}) - \xi_{jt} = 0.
\end{equation*}
By the implicit function theorem (IFT), the share response to prices is
\begin{equation*}
    \nabla_{p_t} s_t = -[\nabla_{s_t} \Upsilon_t]^{-1} \nabla_{p_t} \Upsilon_t.
\end{equation*}
Thus, the own-price elasticity for product $j$ is given by
\begin{equation}\label{eq:el_ift_app}
    \varepsilon_{jj}(\gamma) = \frac{p_{jt}}{s_{jt}} \frac{\partial s_{jt}}{\partial p_{jt}} = -\frac{p_{jt}}{s_{jt}} \left([\nabla_{s_t} \Upsilon_t]^{-1} \nabla_{p_t} \Upsilon_t\right)_{jj}.
\end{equation}

Our goal is to derive the Gateaux derivative of $\varepsilon_{jj}(\gamma)$ with respect to $\gamma$ and prove that it is linear in the perturbation direction so that the results of Theorem~\ref{thm:asy_norm_nonlin} and Corollary~\ref{cor:asy_norm_nonlin_wm} are applicable to the elasticity functional.

First, we will rewrite \eqref{eq:el_ift_app} in matrix form for convenience. Let us take a closer look at the price and share Jacobians. The $(j,k)$ element of the price Jacobian is
\begin{equation*}
    (\nabla_{p_t} \Upsilon_t)_{jk} = \frac{\partial \Upsilon_{jt}}{\partial p_{kt}} = -\frac{\partial \gamma(\omega_{jt})}{\partial p_{kt}}.
\end{equation*}
Define matrix $\Gamma^p \in \R^{J \times J}$ with elements
\begin{equation*}
    \Gamma^p_{jk} \equiv \frac{\partial \gamma(\omega_{jt})}{\partial p_{kt}}.
\end{equation*}
Then,
\begin{equation*}
    \nabla_{p_t} \Upsilon_t = -\Gamma^p.
\end{equation*}

The $(j,k)$ element of the share Jacobian is
\begin{equation*}
    (\nabla_{s_t} \Upsilon_t)_{jk} = \frac{\partial \Upsilon_{jt}}{\partial s_{kt}} = \frac{\partial \log(s_{jt}/s_{0t})}{\partial s_{kt}} - \frac{\partial \gamma(\omega_{jt})}{\partial s_{kt}}.
\end{equation*}
Define the following matrices: (i) $L \in \R^{J \times J}$ with $L_{jk} \equiv \frac{\partial \log(s_{jt}/s_{0t})}{\partial s_{kt}}$, and (ii) $\Gamma^s \in \R^{J \times J}$ with $\Gamma^s_{jk} \equiv \frac{\partial \gamma(\omega_{jt})}{\partial s_{kt}}$. Then, we can rewrite the share Jacobian as
\begin{equation*}
    \nabla_{s_t} \Upsilon_t = L - \Gamma^s.
\end{equation*}
Hence, the share derivative matrix becomes 
\begin{equation*}
	\nabla_{p_t} s_t = (L - \Gamma^s)^{-1} \Gamma^p.
\end{equation*}
As a result, we can rewrite \eqref{eq:el_ift_app} in matrix form as
\begin{equation}
    \varepsilon_{jj}(\gamma) = \frac{p_{jt}}{s_{jt}} \left[(L - \Gamma^s(\gamma))^{-1} \Gamma^p(\gamma)\right]_{jj}.
\end{equation}

\begin{remark*}
For the logit model, $\gamma(\omega_{jt}) = \alpha p_{jt} + x_{jt}^{(2)}\beta$, so $\Gamma^s = 0$ and $\Gamma^p = \alpha I$. The share Jacobian reduces to $L$, which has the standard form $L = D_s^{-1} + \frac{1}{s_{0t}}\mathbf{1}\mathbf{1}'$ where $D_s = \text{diag}(s_1, \ldots, s_J)$.
\end{remark*}

\subsection{The Gateaux Derivative}

Here we introduce a useful lemma to prove the main statement below.

\begin{lemma}\label{lem:inverse_expansion}
For a matrix $M$ and perturbation $N$, with $\|M^{-1}N\| < 1$:
\begin{equation*}
    (M - \epsilon N)^{-1} = M^{-1} + \epsilon M^{-1}NM^{-1} + O(\epsilon^2).
\end{equation*}
\end{lemma}

\begin{proof}
Using the Neumann series expansion:
\begin{align*}
    (M - \epsilon N)^{-1} &= (M(I - \epsilon M^{-1}N))^{-1} \\
    &= (I - \epsilon M^{-1}N)^{-1}M^{-1} \\
    &= (I + \epsilon M^{-1}N + O(\epsilon^2))M^{-1} \\
    &= M^{-1} + \epsilon M^{-1}NM^{-1} + O(\epsilon^2). \qedhere
\end{align*}
\end{proof}

\begin{proposition}\label{prop:gateaux}
The Gateaux derivative of the own-price elasticity functional $f(\gamma) = \varepsilon_{jj}(\gamma)$ in direction $\zeta$ is given by
\begin{equation}\label{eq:gateaux}
    D_\gamma \varepsilon_{jj}[\zeta] = \frac{p_{jt}}{s_{jt}}\left[\left(A^{-1}Z^p\right)_{jj} + \left(A^{-1}Z^s A^{-1}\Gamma^p\right)_{jj}\right]
\end{equation}
where $A = L - \Gamma^s$, and $Z^p$, $Z^s$ are the price and share derivative matrices of the perturbation $\zeta$.
\end{proposition}

\begin{proof}
Let $\zeta$ be a perturbation direction. Consider $\gamma \to \gamma + \epsilon \zeta$ for small $\epsilon > 0$. 
Define the derivatives of the perturbation as
\begin{align*}
    Z^s \in \R^{J \times J} \text{ with } Z^s_{jk} \equiv \frac{\partial \zeta(\omega_{jt})}{\partial s_{kt}}, \\
    Z^p \in \R^{J \times J} \text{ with }  Z^p_{jk} \equiv \frac{\partial \zeta(\omega_{jt})}{\partial p_{kt}}.
\end{align*}
Under the perturbation $\gamma \to \gamma + \epsilon\zeta$,
\begin{align*}
    \Gamma^s(\gamma + \epsilon\zeta) &= \Gamma^s(\gamma) + \epsilon Z^s, \\
    \Gamma^p(\gamma + \epsilon\zeta) &= \Gamma^p(\gamma) + \epsilon Z^p.
\end{align*}
Define $A(\gamma) \equiv L - \Gamma^{s}(\gamma)$, then $A(\gamma + \epsilon\zeta) = L - \Gamma^s - \epsilon Z^s = A - \epsilon Z^s$.

By Lemma~\ref{lem:inverse_expansion}, the perturbed share derivative matrix is given by
\begin{align*}
    \nabla_{p_t} s_t (\gamma \to \gamma + \epsilon\zeta) & = (A - \epsilon Z^s)^{-1}(\Gamma^p + \epsilon Z^p) \\
    &= \left(A^{-1} + \epsilon A^{-1}Z^s A^{-1} + O(\epsilon^2)\right)\left(\Gamma^p + \epsilon Z^p\right) \\
    &= A^{-1}\Gamma^p + \epsilon\left(A^{-1}Z^p + A^{-1}Z^s A^{-1}\Gamma^p\right) + O(\epsilon^2).
\end{align*}
Hence, the perturbed elasticity is
\begin{align*}
	\varepsilon_{jj}(\gamma + \epsilon\zeta) &= \frac{p_{jt}}{s_{jt}}\left[(A - \epsilon Z^s)^{-1}(\Gamma^p + \epsilon Z^p)\right]_{jj} \\
    &= \frac{p_{jt}}{s_{jt}}\left[(A^{-1}\Gamma^p)_{jj} + \epsilon\left(A^{-1}Z^p + A^{-1}Z^s A^{-1}\Gamma^p\right)_{jj}\right] + O(\epsilon^2).
\end{align*}
Finally, by definition of the Gateaux derivative,
\begin{align*}
    D_\gamma \varepsilon_{jj}[\zeta] &= \frac{d}{d\epsilon}\bigg|_{\epsilon=0} \varepsilon_{jj}(\gamma + \epsilon\zeta) \\
    &= \frac{d}{d\epsilon}\bigg|_{\epsilon=0} \frac{p_{jt}}{s_{jt}}\left[(A^{-1}\Gamma^p)_{jj} + \epsilon\left(A^{-1}Z^p + A^{-1}Z^s A^{-1}\Gamma^p\right)_{jj} + O(\epsilon^2)\right] \\
    &= \frac{p_{jt}}{s_{jt}}\left[\left(A^{-1}Z^p\right)_{jj} + \left(A^{-1}Z^s A^{-1}\Gamma^p\right)_{jj}\right]. \qedhere
\end{align*}
\end{proof}

\begin{proposition}\label{prop:linearity}
The Gateaux derivative $D_\gamma \varepsilon_{jj}[\zeta]$ is linear in $\zeta$. That is, for any perturbations $\zeta_1, \zeta_2$ and scalars $c_1, c_2 \in \R$,
\begin{equation*}
    D_\gamma \varepsilon_{jj}[c_1\zeta_1 + c_2\zeta_2] = c_1 D_\gamma \varepsilon_{jj}[\zeta_1] + c_2 D_\gamma \varepsilon_{jj}[\zeta_2].
\end{equation*}
\end{proposition}

\begin{proof}
We prove linearity by showing that each component of the Gateaux derivative is linear in $\zeta$. First, by linearity of the differentiation operator, for any functions $\zeta_1, \zeta_2$ and scalars $c_1, c_2$ we have
\begin{align*}
    \frac{\partial (c_1\zeta_1 + c_2\zeta_2)(\omega_{jt})}{\partial p_{kt}} &= c_1 \frac{\partial \zeta_1(\omega_{jt})}{\partial p_{kt}} + c_2 \frac{\partial \zeta_2(\omega_{jt})}{\partial p_{kt}} \\
    \frac{\partial (c_1\zeta_1 + c_2\zeta_2)(\omega_{jt})}{\partial s_{kt}} &= c_1 \frac{\partial \zeta_1(\omega_{jt})}{\partial s_{kt}} + c_2 \frac{\partial \zeta_2(\omega_{jt})}{\partial s_{kt}}
\end{align*}
In matrix notation, if $\zeta = c_1\zeta_1 + c_2\zeta_2$, then
\begin{align*}
    Z^p &= c_1 Z^p_1 + c_2 Z^p_2, \\
    Z^s &= c_1 Z^s_1 + c_2 Z^s_2.
\end{align*}

Second, note that matrices $A = L - \Gamma^s(\gamma)$ and $\Gamma^p(\gamma)$ depend on the current function $\gamma$ at which we evaluate the Gateaux derivative. They do \emph{not} depend on the perturbation direction $\zeta$. Therefore, both $A^{-1}Z^p$ and $A^{-1}Z^s A^{-1}\Gamma^p$ terms are linear in $\zeta$. Since both terms are linear in $\zeta$, their sum is linear in $\zeta$, which establishes the result.
\end{proof}

\begin{remark*}
	The results in this section trivially extend to cross-price elasticities.
\end{remark*}

\subsection{Estimation details}\label{app:el_estimation}

This section provides a detailed procedure for estimating the average own-price elasticity functional.

\subsubsection{Constructing the Moment Condition for PGMM}

Let $\zeta$ be a perturbation direction. From Proposition~\ref{prop:gateaux}, the Gateaux derivative of the elasticity functional in direction $\zeta$ is
\begin{equation*}
    D_\gamma \varepsilon_{jj}[\zeta] = \frac{p_{jt}}{s_{jt}}\left[\left(A^{-1}Z^p_\zeta\right)_{jj} + \left(A^{-1}Z^s_\zeta A^{-1}\Gamma^p\right)_{jj}\right],
\end{equation*}
where $Z^p_\zeta$ and $Z^s_\zeta$ are the price and share Jacobians of the perturbation $\zeta$. The Riesz representer $\alpha_{jj,0}(z)$ satisfies the population moment condition
\begin{equation*}
    \E\left[D_\gamma \varepsilon_{jj}[\zeta] - \alpha_{jj,0}(z_{jt})\zeta(\omega_{jt})\right] = 0 \quad \text{for all } \zeta.
\end{equation*}
This follows from the linearity of $D_\gamma \varepsilon_{jj}[\cdot]$ established in Proposition~\ref{prop:linearity}.

Let $d(\omega) = (d_1(\omega), \ldots, d_q(\omega))'$ be a $q$-dimensional dictionary of basis functions that represent deviations from $\gamma_0$. These basis functions span the space of possible perturbation directions. Using $d(\omega)$, we can form a vector of $q$ moment conditions by setting $\zeta = d_k$ for $k = 1, \ldots, q$,
\begin{equation}\label{eq:elast_mom_system}
    \E\left[D_\gamma \varepsilon_{jj}[d_k] - \alpha_{jj,0}(z_{jt})d_k(\omega_{jt})\right] = 0, \quad k = 1, \ldots, q.
\end{equation}
The sample moment condition corresponding to \eqref{eq:elast_mom_system} is
\begin{equation*}
    \hat\psi_\gamma(d_k, \rho) = \frac{1}{T}\sum_{t=1}^{T}\left\{D_{\hat\gamma}\varepsilon_{jj,t}[d_k] - d_k(\omega_{jt})b(z_{jt})'\rho\right\} = 0, \quad k = 1, \ldots, q,
\end{equation*}
where $D_{\hat\gamma}\varepsilon_{jj,t}[d_k]$ denotes the Gateaux derivative evaluated at the current estimate $\hat\gamma$ using $d_k$ as the perturbation direction.

\subsubsection{Double-Splitting for $\hat{M}_\ell$}

Since the elasticity functional is nonlinear in $\gamma$, the construction of $\hat{M}_\ell$ requires further sample splitting to obtain an unbiased estimator. The key difference from linear functionals is that $\hat{M}_\ell$ depends on an estimator of $\gamma$. To ensure $\hat{M}_\ell$ is unbiased, we use observations that were not used to construct $\hat\gamma$ appearing in the Gateaux derivative.

Let $\hat\gamma_{\ell,\ell'}$ denote an estimator of $\gamma$ based on observations not in either $\cT_\ell$ or $\cT_{\ell'}$. The unbiased estimator $\hat{M}_\ell$ is then
\begin{align*}
    \hat{M}_\ell &= (\hat{M}_{\ell,1}, \ldots, \hat{M}_{\ell,q})',\\
    \hat{M}_{\ell,k} &= \frac{1}{T - T_\ell}\sum_{\ell' \neq \ell}\sum_{t \in \cT_{\ell'}} D_{\hat\gamma_{\ell,\ell'}}\varepsilon_{jj,t}[d_k], 
\end{align*}
where $D_{\hat\gamma_{\ell,\ell'}}\varepsilon_{jj,t}[d_k]$ is the Gateaux derivative of the elasticity evaluated at $\hat\gamma_{\ell,\ell'}$ in direction $d_k$.

\subsubsection{Complete Estimation Procedure}

The estimation procedure is as follows:

\begin{enumerate}
    \item Assuming the data $\{W_t\}_{t=1}^{T}$ is i.i.d.\ across markets, let $\cT_\ell$, $\ell = 1, \ldots, L$, be a partition of the observation index set $\{1, \ldots, T\}$ into $L$ distinct subsets of approximately equal size. Let $T_\ell$ denote the number of observations in fold $\ell$.
    
    \item For each pair of folds $(\ell, \ell')$ with $\ell \neq \ell'$, estimate $\hat\gamma_{\ell,\ell'}$ using an MLIV estimator on $\{W_t : t \notin \cT_\ell \cup \cT_{\ell'}\}$.
    
    \item For each fold $\ell = 1, \ldots, L$:
    \begin{enumerate}
        \item[(a)] Compute $\hat{G}_\ell$: For each $t \notin \cT_\ell$, form
        \begin{equation*}
            \hat{G}_\ell = \frac{1}{T - T_\ell}\sum_{t \notin \cT_\ell} d(\omega_{jt})b(Z_{jt})'.
        \end{equation*}
        
        \item[(b)] Compute $\hat{M}_\ell$ using double-splitting: For each $\ell' \neq \ell$ and each $t \in \cT_{\ell'}$:
        \begin{itemize}
            \item Compute Jacobians of $\hat\gamma_{\ell,\ell'}$ at $\omega_t$: $\hat\Gamma^p_{\ell,\ell',t}$, $\hat\Gamma^s_{\ell,\ell',t}$
            \item Form $\hat{A}_{\ell,\ell',t} = L_t - \hat\Gamma^s_{\ell,\ell',t}$
            \item For each basis function $d_k$, compute its Jacobians $Z^p_{d_k,t}$, $Z^s_{d_k,t}$
            \item Compute the Gateaux derivative:
            \begin{equation*}
                D_{\hat\gamma_{\ell,\ell'}}\varepsilon_{jj,t}[d_k] = \frac{p_{jt}}{s_{jt}}\left[(\hat{A}_{\ell,\ell',t}^{-1}Z^p_{d_k,t})_{jj} + (\hat{A}_{\ell,\ell',t}^{-1}Z^s_{d_k,t}\hat{A}_{\ell,\ell',t}^{-1}\hat\Gamma^p_{\ell,\ell',t})_{jj}\right]
            \end{equation*}
        \end{itemize}
        Then aggregate:
        \begin{equation*}
            \hat{M}_{\ell,k} = \frac{1}{T - T_\ell}\sum_{\ell' \neq \ell}\sum_{t \in \cT_{\ell'}} D_{\hat\gamma_{\ell,\ell'}}\varepsilon_{jj,t}[d_k].
        \end{equation*}
        
        \item[(c)] Estimate $\hat\alpha_{jj,\ell} = b(Z)'\hat\rho_\ell$ via PGMM:
        \begin{equation*}
            \hat\rho_\ell = \argmin_{\rho \in \R^p} (\hat{M}_\ell - \hat{G}_\ell\rho)'\hat\Omega_q(\hat{M}_\ell - \hat{G}_\ell\rho) + 2\lambda_T|\rho|_1.
        \end{equation*}
    \end{enumerate}
    
    \item For each fold $\ell$, estimate $\hat\gamma_\ell$ using observations not in $\cT_\ell$.
    
    \item The debiased estimator $\hat\theta_{jj}$ and its asymptotic variance estimator $\hat{V}_{jj}$ are
    \begin{align*}
        \hat\theta_{jj} &= \frac{1}{T}\sum_{\ell=1}^{L}\sum_{t \in \cT_\ell}\left\{\varepsilon_{jj,t}(\hat\gamma_\ell) + \hat\alpha_{jj,\ell}(z_{jt})[y_{jt} - \hat\gamma_\ell(\omega_{jt})]\right\}, \label{eq:db_elast_final}\\
        \hat{V}_{jj} &= \frac{1}{T}\sum_{\ell=1}^{L}\sum_{t \in \cT_\ell}\hat\psi_{jj,\ell,t}^2, \notag
    \end{align*}
    where 
    \begin{equation*}
        \hat\psi_{jj,\ell,t} = \varepsilon_{jj,t}(\hat\gamma_\ell) - \hat\theta_{jj} + \hat\alpha_{jj,\ell}(z_{jt})[y_{jt} - \hat\gamma_\ell(\omega_{jt})].
    \end{equation*}
\end{enumerate}

\begin{remark*}
The double-splitting in Step 3 requires estimating $\hat\gamma_{\ell,\ell'}$ for each pair $(\ell, \ell')$, resulting in $L(L-1)$ estimators of $\gamma$ for constructing $\hat{M}$. Additionally, Step 4 requires $L$ estimators $\hat\gamma_\ell$ for the final debiased estimator. For $L=5$ folds, this means estimating $\gamma$ a total of $5 \times 4 + 5 = 25$ times. That said, we can greatly improve upon that. Note that for any two folds $\ell$ and $\ell'$, $\ell \neq \ell'$, $\hat\gamma_{\ell,\ell'} = \hat\gamma_{\ell',\ell}$. As a result, we only need $\binom{L}{2} = \binom{5}{2} = 10$ distinct estimators in Step 3 leading to a total of $15$ estimators of $\gamma$, which is a $40\%$ reduction in compute cost. Alas, we cannot completely mitigate the additional compute burden which is the price paid for valid inference on nonlinear functionals.
\end{remark*}


\section{Empirical application}

\subsection{Data cleaning and aggregation details}\label{app:data}

\subsubsection{Imputations}

Data on product characteristics have a lot of missing observations in the type of sweetener and caffeine level. We use the following heuristics to impute those values:
\begin{itemize}
	\item Type of sweetener:
	\begin{itemize}
		\item if the calorie level is \emph{Regular}, then the type of sweetener will be \emph{Sugar};
		\item if the calorie level is \emph{Calorie-free}, then the type of sweetener will be \emph{Unsweetened};
		\item if the calorie level is \emph{Diet} and the flavor is not \emph{Cola}, then the type of sweetener will be \emph{Sweetener}.
	\end{itemize}
	\item Caffeine info:
	\begin{itemize}
		\item if flavor is \emph{Grapefruit}, \emph{Lemon Lime}, \emph{Natural}, \emph{Strawberry}, \emph{Pineapple}, \emph{Grape}, \emph{Fruit Punch}, it is \emph{Caffeine-free};
		\item if flavor is \emph{Dew} , \emph{Pepper}, \emph{Cherry Cola}, it is \emph{Caffeine}.
	\end{itemize}
\end{itemize}

We also replace zero sales with ones and impute corresponding missing prices with the average price of all other observed products in a particular store in a particular week.

\subsubsection{Product characteristics aggregation}

All product characteristics are categorical variables, to facilitate computations we group product attributes into larger groups which can be coded up as dummy variables. We use the following heuristics:
\begin{itemize}
	\item Flavor/scent:
	\begin{itemize}
		\item \emph{Cola} (such as \emph{Cherry Cola}, \emph{Wild Cherry Cola}, \emph{Cola with Lemon} and so on, basically everything with \emph{Cola})
		\item \emph{Lemonade} (such as \emph{Lemonade}, \emph{Lemon Lime}, \emph{Mandarine Lime}, \emph{Citrus}, \emph{Tangerine}, \emph{Punch}, etc.)
		\item \emph{Alcohol-free beer} (such as \emph{Root Beer},\emph{Birch Beer}, etc.)
		\item \emph{Berries} (\emph{Strawberry}, \emph{Raspberry} , \emph{Cherry}, etc.)
		\item \emph{Fruit} (fruity flavors except berries or lemon, such as \emph{Pineapple}, \emph{Grape}, \emph{Peach}, \emph{Watermelon}, etc.)
		\item \emph{Cream Soda} (\emph{Cream Soda}, \emph{Red Cream Soda}, etc.)
		\item \emph{Others} 
	\end{itemize}
	\item Caffeine level:
	\begin{itemize}
		\item caffeine-free and 55\% caffeine-free are considered \emph{Caffeine-free}
		\item other beverages are considered to contain caffeine
	\end{itemize}
	\item Calorie level:
	\begin{itemize}
		\item calorie-free and diet beverages are considered to be \emph{Diet} 
		\item other beverages are considered to be \emph{Regular}
	\end{itemize}
	\item Type of sweetener:
	\begin{itemize}
		\item \emph{Sugar-free}
		\item \emph{Sweetener} (non-saccharin): Nutra, aspartame, sucralose, splenda
		\item \emph{Sugar}: all entries corresponding to corn sweeteners and sugar/saccharin containing products
	\end{itemize}
\end{itemize}


\subsection{Discussion of Hausman Instruments} \label{app:hausman}

A natural candidate for instrumenting endogenous prices in demand estimation is Hausman instruments, which use the price of the same product in other geographic markets to isolate cost-driven price variation. In our setting, the Hausman instrument for product $j$ in DMA $d$ during month $t$ is
\[
  z_{jdt}^{\text{H}} = \frac{1}{D-1}\sum_{d' \neq d} p_{jd't},
\]
where $D = 50$ is the number of DMAs. The identifying assumption is that cross-market price correlation reflects common cost shocks (changes in input prices, national wholesale pricing, distribution costs, etc.) that are orthogonal to local demand unobservables. 

However, as pointed out by \citet{nevo2001measuring}, this assumption can be easily violated in practice if a positive demand shock hits a product nationally through an advertising campaign, seasonal demand increase, or the end of a promotional period.  Since the Hausman instrument averages prices across DMAs within the same month, it inherits these common demand shocks, creating a positive correlation between the instrument and the structural error. 

\begin{table}[htbp]
\centering
\begin{threeparttable}
\caption{First-Stage Regressions: Price on Hausman Instrument}
\label{tab:hausman_first_stage}
\begin{tabular}{lcccc}
\toprule
 & (1) & (2) & (3) & (4) \\
\midrule
$\hat{\beta}_{\text{Hausman}}$ & 0.833 & 0.781 & 0.903 & 0.872 \\
$R^2$ & 0.077 & 0.078 & 0.470 & 0.471 \\
Partial $F$ & 502.1 & 333.9 & 1016.4 & 718.4 \\[4pt]
Product FE & & \checkmark & & \checkmark \\
DMA FE & & & \checkmark & \checkmark \\
\bottomrule
\end{tabular}
\begin{tablenotes}
\small
\item $N = 5{,}993$ observations (10 products, 50 DMAs, 12 months). Partial $F$ is the $F$-statistic for the excluded instrument in each specification.
\end{tablenotes}
\end{threeparttable}
\end{table}

We evaluate the identifying assumption for the IRI carbonated beverage data and find strong evidence that it fails. First, we note that the within-product-month standard deviation of the Hausman instrument is only 0.000592, while the within-product standard deviation of 0.0089. Nearly all variation in the instrument is at the product$\times$month level, not the DMA level. This is precisely the level at which common demand shocks would operate, casting immediate doubt on the exclusion restriction.

Table~\ref{tab:hausman_first_stage} reports first-stage regressions of product price on the Hausman instrument under different sets of fixed effects. The instrument is highly relevant in all specifications, with $F$-statistics ranging from 334 to 718, well above conventional thresholds for instrument strength. 

\begin{table}[htbp]
\centering
\begin{threeparttable}
\caption{Logit Demand Estimates: BLP vs.\ Hausman Instruments}
\label{tab:hausman_logit}
\begin{tabular}{lccccc}
\toprule
 & (1) & (2) & (3) & (4) & (5) \\
 & BLP IVs & \multicolumn{4}{c}{Hausman IV} \\
\cmidrule(lr){3-6}
\midrule
$\hat{\alpha}$ (price) & $-13.53$ & $+0.35$ & $+0.20$ & $+0.34$ & $+0.20$ \\
 & $(7.64)$ & $(1.39)$ & $(1.39)$ & $(1.23)$ & $(1.21)$ \\[4pt]
Mean own-price elasticity & $-3.09$ & $+0.08$ & $+0.05$ & $+0.08$ & $+0.05$ \\[4pt]
Product FE & & & \checkmark & & \checkmark \\
DMA FE & & & & \checkmark & \checkmark \\
\bottomrule
\end{tabular}
\begin{tablenotes}
\small
\item All specifications include five product characteristics (caffeine, cola, lemon-lime, sugar, pepper); columns without product FE also include a constant. BLP IVs are the two non-collinear rival characteristic sums (caffeine and sugar). Mean own-price elasticity is averaged across market-level elasticities.
\end{tablenotes}
\end{threeparttable}
\end{table}

Table~\ref{tab:hausman_logit} compares 2SLS logit estimates using Hausman instruments against those obtained with BLP-style differentiation instruments. The Hausman instrument produces a positive price coefficient of $+0.35$, with the implied mean own-price elasticity of $+0.08$. This result is robust to the inclusion of fixed effects. By contrast, BLP instruments, which satisfy the exclusion restriction by construction, yield $\hat{\alpha} = -13.53$, corresponding to a mean own-price elasticity of $-3.09$. This evidence suggests that despite being strong [redictors of local prices Hausman instruments violate the exclusion restriction leading to the upward bias in the price coefficient estimate.


\subsection{IIA-test} \label{app:iia_test}

The logit demand model imposes the Independence of Irrelevant Alternatives (IIA) property: the ratio of choice probabilities between any two products is independent of the characteristics of all other products. We test this restriction using the differentiation instruments test of \citet{gandhi_houde2019}.

Under IIA, a product's market share depends only on its own characteristics and the inclusive value, which serves as a sufficient statistic for the competitive environment. In particular, differentiation instruments, which capture how similarly positioned competing products are in characteristic space, should have no explanatory power for market shares conditional on own characteristics. The test exploits this implication.

Following Section 3.3 in \citet{gandhi_houde2019}, we estimate the logit model by 2SLS while controlling for the differentiation IVs. Specifically,
\begin{equation*}
	\log \left(\frac{s_{jt}}{s_{0t}}\right) = \alpha p_{jt} + x_{jt}'\beta + A_{j}(x_t)'\gamma + e_{jt},
\end{equation*}
where $s_{jt}$ and $s_{0t}$ are the product and outside-good shares, $p_{jt}$ is price, and $A_{j}(x_t)$ is a vector of local differentiation instruments. The differentiation IVs are constructed from the five product characteristics with cross-interactions. For instance, $\sum_{k \neq j} \mathbf{1}[Caffeine_k = Caffeine_j] \cdot (Cola_k - Cola_j)$ capture whether proximity along one characteristic dimension (Caffeine) combined with differentiation along another (Cola) predicts substitution. After removing instruments collinear with $(1, x_{jt})$, six differentiation instruments survive. To identify the price coefficient, we use BLP instruments. 

IIA is tested via the null hypothesis $H_0\colon \gamma = 0$. With the $F$-statistic of $68.5$ ($p < 0.0001$), and the robust Wald statistic of $991.4$ ($p < 0.0001$), the test strongly rejects the null, providing strong evidence against IIA.


\subsection{Parametric demand estimates} \label{app:param_demand}

\begin{table}[htbp]
\centering
\begin{threeparttable}
\caption{Parametric Demand Estimates}
\label{tab:blp_estimates}
\begin{tabular}{lccc}
\toprule
 & (1) & (2) & (3) \\
\midrule
\multicolumn{4}{l}{\textit{Linear Parameters ($\beta$)}} \\[3pt]
Constant & $-1.403$ & $6.014$ & $5.555$ \\
 & $(1.913)$ & $(3.629)$ & $(5.344)$ \\[2pt]
Price & $-13.533$ & $-44.145$ & $-42.382$ \\
 & $(7.639)$ & $(14.445)$ & $(24.743)$ \\[2pt]
Sugar & $0.885$ & $0.993$ & $1.008$ \\
 & $(0.032)$ & $(0.066)$ & $(0.179)$ \\[2pt]
Caffeine & $0.461$ & $1.907$ & $2.257$ \\
 & $(0.031)$ & $(1.640)$ & $(3.146)$ \\[2pt]
Cola & $2.553$ & $2.398$ & $2.408$ \\
 & $(0.070)$ & $(0.116)$ & $(0.139)$ \\[2pt]
Lemonade & $1.389$ & $0.971$ & $0.964$ \\
 & $(0.066)$ & $(0.117)$ & $(0.146)$ \\[2pt]
Pepper & $1.250$ & $1.542$ & $1.526$ \\
 & $(0.096)$ & $(0.174)$ & $(0.259)$ \\[2pt]
\midrule
\multicolumn{4}{l}{\textit{Taste Heterogeneity ($\Sigma$)}} \\[3pt]
$\sigma_{Sugar}$ & \textemdash & $2.095$ & $2.054$ \\
 &  & $(0.481)$ & $(0.488)$ \\[2pt]
$\sigma_{Caffeine}$ & \textemdash & $9.234$ & $10.901$ \\
 &  & $(8.291)$ & $(11.515)$ \\[2pt]
$\sigma_{Price}$ & \textemdash & \textemdash & $1.050$ \\
 &  &  & $(42.301)$ \\[2pt]
\midrule
Mean own-price elasticity & $-3.09$ & $-8.93$ & $-8.54$ \\
Elasticity range & $[-3.30,\, -2.74]$ & $[-10.38,\, -7.13]$ & $[-9.96,\, -6.80]$ \\
Instruments & BLP & Diff (local) & Diff (local) \\
\bottomrule
\end{tabular}
\begin{tablenotes}[flushleft]
\footnotesize
\item Demand model estimated via 2-step GMM on monthly DMA-level data for 10 carbonated beverage products across 50 DMAs (600 market-months, 5{,}993 observations). Standard errors in parentheses. Column~(1) is a plain logit with BLP instruments (rival characteristic sums). Columns~(2)--(3) are random coefficients logit with diagonal $\Sigma$, using local differentiation instruments with cross-interactions. ``\textemdash'' indicates parameter not present in specification.
\end{tablenotes}
\end{threeparttable}
\end{table}

Table~\ref{tab:blp_estimates} reports parametric demand estimates that serve as benchmarks for the semiparametric approach. We estimate three specifications of increasing flexibility: a standard logit with BLP instruments, a random coefficients (RC) logit with taste heterogeneity on sugar and caffeine, and an RC logit that additionally allows heterogeneity in price sensitivity. All models are estimated with the PyBLP package \citep{conlon2020pyblp} available in Python.

First, the price coefficient roughly triples in magnitude when random coefficients are introduced, rising from $-13.5$ (logit) to $-44.1$ and $-42.4$ in columns~(2) and~(3). Given the absence of external price instruments, this behavior is not surprising. Plus, the price coefficient is borderline (in-)significant across specifications.

Second, identification of the random coefficients is uneven. The sugar heterogeneity parameter $\sigma_{\text{Sugar}} \approx 2.1$ is precisely estimated (SE $\approx 0.48$) in both specifications, consistent with strong heterogeneity in preferences for sugary versus diet beverages. By contrast, $\sigma_{\text{Caffeine}}$ is large but imprecise, and $\sigma_{\text{Price}}$ in column~(3) is essentially unidentified (SE $= 42.3$). The lack of cost-side instruments limits the ability to separate price sensitivity heterogeneity from unobserved demand shocks.

Third, the RC models imply mean own-price elasticities of $-8.9$ to $-8.5$, substantially larger in magnitude compared to logit ($-3.1$) and carbonated beverage literature \citep[see e.g.,][]{dube2004multiple}. Given high sensitivity of RC models to instrument strength and parametric assumptions on the distribution of taste heterogeneity, we tend to treat logit results as a more reliable benchmark to compare against the semiparametric estimates.


\subsection{Robustness checks}\label{app:robustness}

Here we provide additional results to check robustness of the baseline specification in Section \ref{subsec:iri}.

\subsubsection{Choice of the penalty parameter}

The penalty parameter $c_1$ plays an important role in Riesz representer estimation. Since we already use double cross-fitting for the own-price elasticity functional, applying the cross-validation procedure \eqref{eq:pgmm_split_penalty} on top of it would introduce additional computational complexity and estimation instability. Instead, we run the ADML estimator over a grid of tuning parameters $\{10^{-4},\,10^{-5},\,10^{-6},\,10^{-7},\,10^{-8}\}$, with $c_1 = 10^{-6}$ as our baseline value. Table \ref{tab:robustness_pgmm_c} reports the results. 

Overall, results are robust to the choice of $c_1$. For most products, the debiased estimates are nearly invariant to $c_1$ across the full grid. For a few products (e.g., Diet Coke, Sprite), the magnitude of the debiasing correction and the corresponding standard error vary with $c_1$, reflecting the bias-variance tradeoff: larger $c_1$ penalizes the Riesz representer more heavily, yielding estimates closer to the plug-in with smaller standard errors, while smaller $c_1$ permits larger corrections at the cost of increased variance. Notably, estimates stabilize between $c_1 = 10^{-7}$ and $c_1 = 10^{-8}$, suggesting convergence of the Riesz representer at the lower end of the grid.

\begin{table}[htbp]
\small
\centering
\begin{threeparttable}
\caption{Sensitivity to PGMM Penalty Parameter}
\label{tab:robustness_pgmm_c}
\begin{tabular}{lcccccc}
\toprule
 & & \multicolumn{5}{c}{ADML} \\
\cmidrule(lr){3-7}
Product & PI & $c_1 = 10^{-4}$ & $c_1 = 10^{-5}$ & $\mathbf{c_1 = 10^{-6}}$ & $c_1 = 10^{-7}$ & $c_1 = 10^{-8}$ \\
\midrule
Coke Classic & $-3.456$ & $-3.454$ & $-3.455$ & $\mathbf{-3.436}$ & $-3.395$ & $-3.387$ \\
  & $(0.021)$ & $(0.021)$ & $(0.021)$ & $\mathbf{(0.022)}$ & $(0.026)$ & $(0.027)$ \\[2pt]
Pepsi Classic & $-3.468$ & $-3.469$ & $-3.469$ & $\mathbf{-3.472}$ & $-3.487$ & $-3.493$ \\
  & $(0.024)$ & $(0.024)$ & $(0.024)$ & $\mathbf{(0.024)}$ & $(0.025)$ & $(0.025)$ \\[2pt]
Diet Coke & $-3.765$ & $-3.747$ & $-3.617$ & $\mathbf{-3.448}$ & $-3.404$ & $-3.400$ \\
  & $(0.018)$ & $(0.018)$ & $(0.026)$ & $\mathbf{(0.089)}$ & $(0.117)$ & $(0.121)$ \\[2pt]
Diet Pepsi & $-3.838$ & $-3.848$ & $-3.845$ & $\mathbf{-3.965}$ & $-3.970$ & $-3.969$ \\
  & $(0.022)$ & $(0.021)$ & $(0.022)$ & $\mathbf{(0.025)}$ & $(0.027)$ & $(0.028)$ \\[2pt]
Caffeine-free Diet Coke & $-3.889$ & $-3.840$ & $-3.880$ & $\mathbf{-3.767}$ & $-3.754$ & $-3.752$ \\
  & $(0.018)$ & $(0.018)$ & $(0.018)$ & $\mathbf{(0.021)}$ & $(0.021)$ & $(0.021)$ \\[2pt]
Caffeine-free Diet Pepsi & $-3.937$ & $-4.011$ & $-3.951$ & $\mathbf{-4.108}$ & $-4.134$ & $-4.132$ \\
  & $(0.020)$ & $(0.020)$ & $(0.020)$ & $\mathbf{(0.021)}$ & $(0.034)$ & $(0.038)$ \\[2pt]
Dr Pepper & $-3.543$ & $-3.542$ & $-3.543$ & $\mathbf{-3.543}$ & $-3.543$ & $-3.543$ \\
  & $(0.019)$ & $(0.020)$ & $(0.019)$ & $\mathbf{(0.020)}$ & $(0.020)$ & $(0.020)$ \\[2pt]
Mountain Dew Classic & $-3.927$ & $-3.927$ & $-3.927$ & $\mathbf{-3.927}$ & $-3.928$ & $-3.931$ \\
  & $(0.032)$ & $(0.032)$ & $(0.032)$ & $\mathbf{(0.032)}$ & $(0.032)$ & $(0.031)$ \\[2pt]
Sprite & $-3.702$ & $-3.701$ & $-3.684$ & $\mathbf{-3.600}$ & $-3.587$ & $-3.586$ \\
  & $(0.017)$ & $(0.017)$ & $(0.017)$ & $\mathbf{(0.022)}$ & $(0.023)$ & $(0.023)$ \\[2pt]
Mountain Dew Other & $-3.992$ & $-3.992$ & $-3.992$ & $\mathbf{-3.992}$ & $-3.992$ & $-3.992$ \\
  & $(0.023)$ & $(0.023)$ & $(0.023)$ & $\mathbf{(0.023)}$ & $(0.023)$ & $(0.023)$ \\[2pt]
\bottomrule
\end{tabular}
\begin{tablenotes}[flushleft]
\footnotesize
\item PI = plug-in estimator; ADML = automatic debiased estimator with PGMM-estimated Riesz representer. The baseline specification ($c_1 = 10^{-6}$, in bold) is used in the main results. Standard errors in parentheses.
\end{tablenotes}
\end{threeparttable}
\end{table}

\subsubsection{Choice of the special regressor} 

Note that the choice of a special regressor is fundamentally different from the choice of a hyperparameter (penalty). The latter governs the quality of the Riesz representer fit, while the former is a structural restriction ensuring identification of the inverse demand function. Hence, different $x^{(1)}$ choices are expected to have a significant impact on the estimates. We run the ADML procedure\footnote{Plug-in estimates also depend on the choice of $x^{(1)}$ and are omitted for brevity} to estimate own-price elasticities under different choices of the special regressor $x^{(1)} \in \{$\emph{Lemonade}, \emph{Cola}, \emph{Pepper}, \emph{Caffeine}, \emph{Sugar}$\}$, with $x^{(1)}$ = \emph{Lemonade} as our baseline. Table \ref{tab:robustness_x1} reports the results. 

We observe that elasticity estimates substantially across $x^{(1)}$ choices. For $x^{(1)} \in \{$\emph{Lemonade}, \emph{Cola}, \emph{Caffeine}$\}$, elasticity estimates are larger in magnitude compared to the logit ones. In contrast, under $x^{(1)} \in \{$\emph{Pepper}, \emph{Sugar}$\}$, estimates are smaller in magnitude than the logit estimates, which is not plausible. This suggests that \emph{Pepper} and \emph{Sugar} have nonlinear effects on the inverse demand and should be included in $\gamma$ instead.  

\begin{table}[htbp]
\centering
\begin{threeparttable}
\caption{Sensitivity to Choice of Special Regressor}
\label{tab:robustness_x1}
\begin{tabular}{lccccc}
\toprule
& \multicolumn{5}{c}{Special regressor} \\
\cmidrule(lr){2-6}
Product & \textbf{\emph{Lemonade}} & \emph{Cola} & \emph{Pepper} & \emph{Caffeine} & \emph{Sugar} \\
\midrule
Coke Classic & $\mathbf{-3.436}$ & $-3.696$ & $-1.884$ & $-3.377$ & $-1.550$ \\
  & $\mathbf{(0.022)}$ & $(0.022)$ & $(0.020)$ & $(0.023)$ & $(0.012)$ \\[2pt]
Pepsi Classic & $\mathbf{-3.472}$ & $-3.792$ & $-1.913$ & $-3.306$ & $-1.577$ \\
  & $\mathbf{(0.024)}$ & $(0.023)$ & $(0.022)$ & $(0.042)$ & $(0.013)$ \\[2pt]
Diet Coke & $\mathbf{-3.448}$ & $-3.614$ & $-1.928$ & $-3.331$ & $-1.679$ \\
  & $\mathbf{(0.089)}$ & $(0.228)$ & $(0.060)$ & $(0.096)$ & $(0.011)$ \\[2pt]
Diet Pepsi & $\mathbf{-3.965}$ & $-4.169$ & $-2.233$ & $-3.693$ & $-1.732$ \\
  & $\mathbf{(0.025)}$ & $(0.047)$ & $(0.024)$ & $(0.041)$ & $(0.012)$ \\[2pt]
Caffeine-free Diet Coke & $\mathbf{-3.767}$ & $-4.065$ & $-2.095$ & $-3.825$ & $-1.716$ \\
  & $\mathbf{(0.021)}$ & $(0.022)$ & $(0.021)$ & $(0.029)$ & $(0.011)$ \\[2pt]
Caffeine-free Diet Pepsi & $\mathbf{-4.108}$ & $-4.469$ & $-2.266$ & $-4.011$ & $-1.778$ \\
  & $\mathbf{(0.021)}$ & $(0.024)$ & $(0.023)$ & $(0.040)$ & $(0.011)$ \\[2pt]
Dr Pepper & $\mathbf{-3.543}$ & $-4.055$ & $-2.014$ & $-3.926$ & $-1.466$ \\
  & $\mathbf{(0.020)}$ & $(0.023)$ & $(0.025)$ & $(0.034)$ & $(0.014)$ \\[2pt]
Mountain Dew Classic & $\mathbf{-3.927}$ & $-4.230$ & $-2.023$ & $-4.479$ & $-1.670$ \\
  & $\mathbf{(0.032)}$ & $(0.028)$ & $(0.026)$ & $(0.058)$ & $(0.022)$ \\[2pt]
Sprite & $\mathbf{-3.600}$ & $-3.893$ & $-1.884$ & $-3.870$ & $-1.562$ \\
  & $\mathbf{(0.022)}$ & $(0.019)$ & $(0.020)$ & $(0.029)$ & $(0.012)$ \\[2pt]
Mountain Dew Other & $\mathbf{-3.992}$ & $-4.467$ & $-2.128$ & $-4.205$ & $-1.772$ \\
  & $\mathbf{(0.023)}$ & $(0.026)$ & $(0.024)$ & $(0.035)$ & $(0.020)$ \\[2pt]
\bottomrule
\end{tabular}
\begin{tablenotes}[flushleft]
\footnotesize
\item The special regressor $x_1$ is excluded from the nonparametric input space $\omega$ and enters the model linearly. Columns compare choices of $x^{(1)} \in \{$\emph{Lemonade}, \emph{Cola}, \emph{Pepper}, \emph{Caffeine}, \emph{Sugar}$\}$. The baseline specification ($x^{(1)}$ = \emph{Lemonade}, in bold) is used in the main results. Standard errors in parentheses.
\end{tablenotes}
\end{threeparttable}
\end{table}


\section{Implementation details}\label{app:implementation}

This subsection summarizes the key implementation choices and hyperparameters used in the Monte Carlo experiments and the empirical application. All experiments use 5-fold cross-fitting unless stated otherwise. The intercept PGMM penalty parameter is set to $c_0=0.1$ throughout. Additional details can be found in the \href{https://github.com/edbakhitov/ADMLIV}{GitHub} repo.

\paragraph{Average derivative Monte Carlo (Section~\ref{sec:toy_mc}).}
The MLIV estimator is a Double Lasso with a degree-3 polynomial basis with pairwise interactions. The first-stage regularization parameter is set to $\alpha = 10^{-4}$, and the second stage uses 3-fold cross-validation over a grid of 100 values log-spaced from $10^{-7}$ to $10^{-1}$. The PGMM estimator of the Riesz representer uses a degree-3 polynomial basis with pairwise interactions for both $b(Z)$ and $d(X)$, adaptive weights, and two-stage estimation. Penalty parameters are set as follows: $c_1 = 10^{-2}$ for $k=2$, $c_1 = 10^{-3}$ for $k=5$, and $c_1 = 10^{-4}$ for $k=10$. As a benchmark, the DML analytical estimator uses the identity score estimator of \cite{chen2023ann_npiv} with the same polynomial basis. Results are based on $1000$ replications with sample sizes $n \in \{100,\, 500,\, 1000,\, 10000\}$ and dimensions $k \in \{2,\, 5,\, 10\}$.

\paragraph{Elasticity Monte Carlo (Section~\ref{subsec:demand_mc}).}
The MLIV estimator is Kernel IV (KIV) with the standard median-distance bandwidth heuristic and bandwidth scale factor of 25.0. Both $d(\omega)$ and $b(z)$ use degree-2 polynomial bases with and without pairwise interactions, respectively. The PGMM penalty is $c_1 = 10^{-7}$ with adaptive weights and two-stage estimation. Double cross-fitting is used for the nonlinear elasticity functional. Results are based on 500 replications per configuration.

\paragraph{IRI scanner data application.}
The data come from the IRI Academic Database \citep{bronnenberg2008IRI}, version 2.3, Year~3 (2003), aggregated to monthly DMA level for the top 10 carbonated beverage products, yielding an unbalanced panel of approximately 564 markets after dropping 7 incomplete market-months. The MLIV estimator is KIV with the standard deviation bandwidth heuristic and bandwidth scale factor of 25.0. The PGMM uses an empirical moment (EM) basis featurizer, which constructs basis functions from the market-level structure of the omega transformation (see Appendix~\ref{app:em_bfs}). The omega basis $d(\omega)$ uses moments $\bigcup_{n=2}^{3} B_n^{jt}$ and the IV basis $b(z)$ uses moments $\bigcup_{n=1}^{2} B_n^{jt}$, where $B_n^{jt}$ is as defined in Appendix~\ref{app:em_bfs}; both use a degree-2 polynomial with pairwise interactions as the outer approximation $g$. The PGMM penalty is $c = 10^{-6}$ with adaptive weights and two-stage estimation. 


\end{document}